\documentclass[runningheads,envcountsame]{llncs}

\usepackage[T1]{fontenc}
\usepackage[english]{babel}

\usepackage{xspace}

\usepackage{lmodern}

\usepackage{comment}

\usepackage{authblk}

\usepackage{scalerel}

\usepackage{graphicx}

\usepackage{dcolumn}

\usepackage{booktabs}

\usepackage[shortlabels, inline]{enumitem} %

\usepackage{multicol}
\usepackage{float}

\usepackage{tabularx}

\usepackage{xltabular}

\usepackage{multirow}

\usepackage{amsfonts, amsmath, amsthm, amssymb} 

\spnewtheorem{notation}[theorem]{Notation}{\bfseries}{}

\usepackage{mathtools} 

\usepackage{mathrsfs}

\usepackage[all]{xy}

\usepackage{xfrac}
\usepackage{faktor}

\usepackage[nice]{nicefrac}

\usepackage{stackrel}

\usepackage{colonequals}

\usepackage[all]{xy}

\usepackage{tikz}
\usetikzlibrary{
    arrows, arrows.meta,
    automata, backgrounds, calc, cd, chains,
    decorations.fractals, decorations.markings, decorations.pathreplacing, decorations.pathmorphing, decorations.text, 
    fadings, fit, folding, mindmap, 
    patterns, patterns, patterns.meta,
    plotmarks, positioning, shadows,
    shapes.geometric, shapes.symbols,
    through, trees, 3d
}
\usepackage{tikz-3dplot}

\usepackage{extarrows}

\usepackage{bussproofs}

\usepackage{pict2e}

\usepackage{accents}

\usepackage{xcolor}
\usepackage{color}

\usepackage{adjustbox}

\usepackage{hyperref}
\usepackage[capitalize]{cleveref} 

\usepackage{lineno}
 
\usepackage{csquotes}

\usepackage{breakcites}

\usepackage{makecell}

\usepackage{changepage}

\makeatletter
\def\orcidID#1{\href{http://orcid.org/#1}{\protect\raisebox{-1.25pt}{\protect\includegraphics{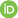}}}}
\makeatother

\newcommand{\customqed}{\hfill $\blacksquare$}

\newcommand{\qmarks}[1]{``#1''}

\DeclareTextFontCommand{\myemph}{\bfseries\em}

\newcolumntype{Y}{>{\centering\arraybackslash}X}

\newenvironment{centeritemize*}[1][]
  {\par\centering\begin{itemize*}[itemjoin=\qquad,#1]}
  {\end{itemize*}\par}

\colorlet{linkcolor}{red!60!black} 
\hypersetup{
    colorlinks=true,
    linkcolor=linkcolor,
    citecolor=linkcolor,
    filecolor=linkcolor,      
    urlcolor=linkcolor,
}

\newcommand{\replabel}{\label} 

\ExplSyntaxOn
\NewDocumentCommand{\repeattheorem}{m}%
 {%
  \group_begin:
  \renewcommand{\thetheorem}{\ref{#1}}
  \renewcommand{\replabel}[1]{\tag{\ref{##1}}}
  \prop_item:Nn \g_reptheorem_prop { #1 }
  \endtheorem
  \group_end:
 }%

\NewDocumentEnvironment{reptheorem}{m+b}
 {%
  \prop_gput:Nnn \g_reptheorem_prop { #1 } { \theorem #2 \endtheorem }
  \theorem#2\unskip\label{#1}\endtheorem%
 }{}%

\prop_new:N \g_reptheorem_prop
\ExplSyntaxOff

\ExplSyntaxOn
\NewDocumentCommand{\repeatlemma}{m}
 {
  \group_begin:
  \renewcommand{\thelemma}{\ref{#1}}
  \renewcommand{\replabel}[1]{\tag{\ref{##1}}}
  \prop_item:Nn \g_replemma_prop { #1 }
  \endlemma
  \group_end:
 }

\NewDocumentEnvironment{replemma}{m+b}
 {
  \prop_gput:Nnn \g_replemma_prop { #1 } { \lemma #2 \endlemma }
  \lemma#2\unskip\label{#1}\endlemma
 }{}

\prop_new:N \g_replemma_prop

\ExplSyntaxOff

\pgfdeclarepatternformonly{south west lines}{\pgfqpoint{-0pt}{-0pt}}{\pgfqpoint{3pt}{3pt}}{\pgfqpoint{3pt}{3pt}}{
        \pgfsetlinewidth{0.4pt}
        \pgfpathmoveto{\pgfqpoint{0pt}{0pt}}
        \pgfpathlineto{\pgfqpoint{3pt}{3pt}}
        \pgfpathmoveto{\pgfqpoint{2.8pt}{-.2pt}}
        \pgfpathlineto{\pgfqpoint{3.2pt}{.2pt}}
        \pgfpathmoveto{\pgfqpoint{-.2pt}{2.8pt}}
        \pgfpathlineto{\pgfqpoint{.2pt}{3.2pt}}
        \pgfusepath{stroke}}

\pgfdeclarepatternformonly{south east lines}{\pgfqpoint{-0pt}{-0pt}}{\pgfqpoint{3pt}{3pt}}{\pgfqpoint{3pt}{3pt}}{
    \pgfsetlinewidth{0.4pt}
    \pgfpathmoveto{\pgfqpoint{0pt}{3pt}}
    \pgfpathlineto{\pgfqpoint{3pt}{0pt}}
    \pgfpathmoveto{\pgfqpoint{.2pt}{-.2pt}}
    \pgfpathlineto{\pgfqpoint{-.2pt}{.2pt}}
    \pgfpathmoveto{\pgfqpoint{3.2pt}{2.8pt}}
    \pgfpathlineto{\pgfqpoint{2.8pt}{3.2pt}}
    \pgfusepath{stroke}}

\tikzcdset{scale cd/.style={every label/.append style={scale=#1},
    cells={nodes={scale=#1}}}}

\DeclarePairedDelimiter\set{\{}{\}}
\DeclarePairedDelimiterX\setvbar[2]{\{}{\}}{#1 \nonscript\;\delimsize \vert \nonscript\; #2}
\DeclarePairedDelimiterX\setcolon[2]{\{}{\}}{#1 : #2}

\newcommand{\mycirc}{\cdot}

\newcommand{\inv}{^{\raisebox{.2ex}{$\scriptscriptstyle-1$}}}

\DeclarePairedDelimiter{\Card}{\lvert}{\rvert} 

\makeatletter
 \newcommand{\xmapsfrom}[2][]{%
    \ext@arrow3095\leftarrowfill@{#1}{#2}\mapsfromchar
}
\makeatother

\newcommand\noloc{%
	\nobreak
	\mspace{6mu plus 1mu}
	{:}
	\nonscript\mkern-\thinmuskip
	\mathpunct{}
	\mspace{2mu}
}

\DeclareMathOperator{\ima}{im}

\DeclareMathOperator{\id}{id}

\let\emptyset\varnothing

\newcommand\rsetminus{\mathbin{\mathpalette\rsetminusaux\relax}}
\newcommand\rsetminusaux[2]{\mspace{-4mu}
  \raisebox{\rsmraise{#1}\depth}{\rotatebox[origin=c]{-20}{$#1\smallsetminus$}}
 \mspace{-4mu}
}
\newcommand\rsmraise[1]{%
  \ifx#1\displaystyle .8\else
    \ifx#1\textstyle .8\else
      \ifx#1\scriptstyle .6\else
        .45%
      \fi
    \fi
  \fi}
\renewcommand{\setminus}{\rsetminus}

\renewcommand{\leq}{\leqslant}
\renewcommand{\geq}{\geqslant}
\newcommand*\meet{\wedge}
\newcommand*\join{\vee}

\newcommand*\bigjoin{\bigvee}

\newcommand{\tss}{\textsuperscript}

\newcommand{\bb}[1]{\mathbb{#1}}
\newcommand{\N}{\ensuremath{\bb{N}}\xspace}

\newcommand{\Q}{\ensuremath{\bb{Q}}\xspace}
\newcommand{\R}{\ensuremath{\bb{R}}\xspace}

\newcommand\restrict[1]{\raisebox{-.6ex}{$|$}_{#1}}

\let\isom\cong

\newcommand{\defeq}{\vcentcolon=}
\newcommand{\eqdef}{=\vcentcolon}

\newcommand{\dittotikz}{%
	\tikz{
		\draw [line width=0.12ex] (-0.2ex,0) -- +(0,0.8ex)
		(0.2ex,0) -- +(0,0.8ex);
		\draw [line width=0.08ex] (-0.6ex,0.4ex) -- +(-1.5em,0)
		(0.6ex,0.4ex) -- +(1.5em,0);
	}%
}

\newcommand{\true}{\ensuremath{\mathsf{True}}\xspace}
\newcommand{\false}{\ensuremath{\mathsf{False}}\xspace}

\let\lnottemp\lnot
\renewcommand{\lnot}{\lnottemp \hspace*{0.1em}}

\let\oldexists\exists
\let\exists\relax
\newcommand{\exists}{\hspace*{0em}\oldexists\hspace*{0.07em}}
\let\oldforall\forall
\let\forall\relax 
\newcommand{\forall}{\hspace*{0em}\oldforall\hspace*{0.07em}}

  {\gdef\scalefactor{#1}\begin{center}\proofSkipAmount \leavevmode}%
  {\scalebox{\scalefactor}{\DisplayProof}\proofSkipAmount \end{center} }
  
\newcommand{\cat}[1]{\mathsf{#1}} 
\newcommand{\Set}{\cat{Set}}

\newcommand{\Graph}{\cat{Graph}}
\newcommand{\ReflGraph}{\cat{ReflGraph}}

\newcommand{\opcat}[1]{#1^{\mathrm{op}}}

\newcommand{\inclusion}{{\mathsf{incl}}}

\newcommand{\HeytingAlgebra}{\cat{HeytAlg}}

\newcommand{\FunctorCat}[2]{{#1}^{#2}}
\newcommand{\Presheaf}[1]{\FunctorCat{\Set}{\opcat{#1}}}
\newcommand{\PresheafHat}[1]{\hat{#1}}
\newcommand{\indexcategory}{\cat{C}}
\newcommand{\PresheafDefault}{\Presheaf{\indexcategory}}
\newcommand{\PresheafHatDefault}{\PresheafHat{\indexcategory}}

\newcommand{\sSet}{\cat{sSet}}
\newcommand{\semisSet}{\cat{sSet}_+}
\newcommand{\ndimsSet}[1][n]{\cat{sSet}^{\leq #1}}
\newcommand{\ndimsemisSet}[1][n]{\cat{sSet}^{\leq #1}_+}
\newcommand{\semi}{_{+}}
\newcommand{\leqn}[1][n]{^{\leq #1}}
\newcommand{\degen}{\mathsf{degen}}
\newcommand{\ithface}[1][i]{\hat{#1}}

\newcommand{\ynhollow}[1][n]{{y(#1)^\partial}} 

\newcommand{\Sheaf}{\mathsf{Sh}}
\newcommand{\Separated}{\mathsf{Sep}}

\newcommand{\labels}{\mathcal{L}}

\newcommand{\FuzzySet}[1]{\cat{FuzzySet(#1)}}
\newcommand{\FuzzySetDefault}{\FuzzySet{\labels}}

\makeatletter
\newcommand{\adjunction}[4]{%
    #1\colon #2%
    \mathrel{\vcenter{%
        \offinterlineskip\m@th
        \ialign{%
            \hfil$##$\hfil\cr
            \longrightarrow\cr
            \noalign{\kern-.3ex}
            {\scriptscriptstyle\bot}\cr
            \longleftarrow\cr
        }%
    }}%
    #3 \noloc #4%
}

\DeclareMathOperator{\Hom}{Hom}
\DeclareMathOperator{\MonoClass}{Mono}
\newcommand{\mono}{\rightarrowtail}

\newcommand{\Sub}{\mathsf{Sub}}
\newcommand{\pullbackalong}[1]{{#1}^\leftarrow}
\newcommand{\charac}[1]{\chi^{#1}}

\newcommand{\idrefl}{
    \begin{tikzcd}[sep=small, ampersand replacement=\&]
        \id_0 \ar[loop, refl]
    \end{tikzcd}
}
\newcommand{\strefl}{
    \begin{tikzcd}[sep=small, ampersand replacement=\&]
        s \ar[loop, refl] \& t \ar[loop, refl]
    \end{tikzcd}
}
\newcommand{\stotrefl}{
    \begin{tikzcd}[sep=small, ampersand replacement=\&]
        s \ar[loop, refl] \ar[r] \& t \ar[loop, refl]
    \end{tikzcd}
}

\newcommand{\BiColGraph}{\cat{BiColGraph}}
\colorlet{fcolor}{blue}
\colorlet{scolor}{red}
\definecolor{fcolor}{rgb}{0,0,1}
\definecolor{scolor}{rgb}{1,0,0}
\definecolor{amber}{rgb}{1.0, 0.75, 0.0}
\newcommand{\stot}{1} 
\newcommand{\stotp}{1'}

\newcommand{\refl}{\mathsf{refl}}

\newcommand{\topology}{\tau}
\newcommand{\topologybar}{\overline}

\makeatletter
\DeclareFontEncoding{LS2}{}{\@noaccents}
\makeatother
\DeclareFontSubstitution{LS2}{stix}{m}{n}

\DeclareSymbolFont{largesymbolsstix}{LS2}{stixex}{m}{n}

\DeclareMathDelimiter{\lParen}{\mathopen}{largesymbolsstix}{"DE}{largesymbolsstix}{"02}
\DeclareMathDelimiter{\rParen}{\mathclose}{largesymbolsstix}{"DF}{largesymbolsstix}{"03}
\DeclareMathDelimiter{\lBrack}{\mathopen}{largesymbolsstix}{"E0}{largesymbolsstix}{"06}
\DeclareMathDelimiter{\rBrack}{\mathclose}{largesymbolsstix}{"E1}{largesymbolsstix}{"07}
\DeclareMathDelimiter{\lBrace}{\mathopen}{largesymbolsstix}{"E8}{largesymbolsstix}{"0E}
\DeclareMathDelimiter{\rBrace}{\mathclose}{largesymbolsstix}{"E9}{largesymbolsstix}{"0F}
\DeclareMathDelimiter{\lbrbrak}{\mathopen} {largesymbolsstix}{"EE}{largesymbolsstix}{"14}
\DeclareMathDelimiter{\rbrbrak}{\mathclose}{largesymbolsstix}{"EF}{largesymbolsstix}{"15}

\DeclareFontFamily{OMX}{MnSymbolE}{}
\DeclareFontShape{OMX}{MnSymbolE}{m}{n}{
    <-6>  MnSymbolE5
   <6-7>  MnSymbolE6
   <7-8>  MnSymbolE7
   <8-9>  MnSymbolE8
   <9-10> MnSymbolE9
  <10-12> MnSymbolE10
  <12->   MnSymbolE12}{}
\DeclareFontShape{OMX}{MnSymbolE}{b}{n}{
    <-6>  MnSymbolE-Bold5
   <6-7>  MnSymbolE-Bold6
   <7-8>  MnSymbolE-Bold7
   <8-9>  MnSymbolE-Bold8
   <9-10> MnSymbolE-Bold9
  <10-12> MnSymbolE-Bold10
  <12->   MnSymbolE-Bold12}{}
 
\DeclareSymbolFont{largesymbolsmnsymbol}  {OMX}{MnSymbolE}{m}{n}

\DeclareMathDelimiter{\langlebar}{\mathopen}{largesymbolsmnsymbol}{'152}{largesymbolsmnsymbol}{'152}
\DeclareMathDelimiter{\ranglebar}{\mathclose}{largesymbolsmnsymbol}{'157}{largesymbolsmnsymbol}{'157}

\makeatletter
\newcommand{\doublearrowtail@inner}[2]{%
  \vcenter{\offinterlineskip
    \halign{%
      ##\cr
      $\m@th#1\rightarrowtail$\cr
      \makebox[\widthof{$\m@th#1\rightarrowtail$}][s]{%
        $\m@th#1\leftarrowtail$%
      }\cr
    }%
  }%
}

\DeclareRobustCommand{\doublearrowtail}{%
  \mathrel{\mathpalette\doublearrowtail@inner\relax}%
}
\makeatother

\makeatletter
\newcommand{\rightleftmapsto@inner}[2]{
  \vcenter{\offinterlineskip
    \halign{
      ##\cr
      $\m@th#1\mapsto$\cr
      \makebox[\widthof{$\m@th#1\mapsto$}][s]{
        $\m@th#1\mapsfrom$
      }\cr
    }
  }
}

\DeclareRobustCommand{\rightleftmapsto}{
  \mathrel{\mathpalette\rightleftmapsto@inner\relax}
}
\makeatother

\makeatletter
\newcommand{\mapstofrom}{\mathpalette\@mapstofrom\relax}
\newcommand*{\@mapstofrom}[2]{%
   \dimen@\fontdimen8
       \ifx#1\displaystyle\textfont\else
       \ifx#1\textstyle\textfont\else
       \ifx#1\scriptstyle\scriptfont\else
       \scriptscriptfont\fi\fi\fi 3
   \mathrel{%
      \vcenter{%
         \vbox{%
            \baselineskip\z@skip
            \lineskip\z@
            \ialign{##\cr$#1\mapstochar\varrightarrow$\cr
            \noalign{\kern\dimen@}%
            $#1\varleftarrow\mapsfromchar$\cr}%
         }%
      }%
   }%
}
\makeatother

\DeclareFontFamily{U} {MnSymbolA}{}

\DeclareFontShape{U}{MnSymbolA}{m}{n}{
  <-6> MnSymbolA5
  <6-7> MnSymbolA6
  <7-8> MnSymbolA7
  <8-9> MnSymbolA8
  <9-10> MnSymbolA9
  <10-12> MnSymbolA10
  <12-> MnSymbolA12}{}
\DeclareFontShape{U}{MnSymbolA}{b}{n}{
  <-6> MnSymbolA-Bold5
  <6-7> MnSymbolA-Bold6
  <7-8> MnSymbolA-Bold7
  <8-9> MnSymbolA-Bold8
  <9-10> MnSymbolA-Bold9
  <10-12> MnSymbolA-Bold10
  <12-> MnSymbolA-Bold12}{}

\DeclareSymbolFont{MnSyA} {U} {MnSymbolA}{m}{n}
\DeclareMathSymbol{\rightmapsto}{\mathrel}{MnSyA}{40}
\DeclareMathSymbol{\leftmapsto}{\mathrel}{MnSyA}{42}

\newcommand{\pbpostrong}{PBPO$^{+}$\xspace}

\newenvironment{mysidepicture}[4]{
\nointerlineskip\noindent
\begin{tikzpicture}[overlay]%
  \node at (\textwidth,0) [anchor=north east,rectangle,inner sep=0,outer sep=0,xshift=#2,yshift=#3] {#4};%
\end{tikzpicture}%
\begin{adjustwidth}{0cm}{#1}%
}{%
\end{adjustwidth}%
}

\tikzset{vertex/.style={circle,fill=black,minimum size=0.8mm,inner sep=0mm,outer sep=0.5mm}}
\tikzset{metaedge/.style={amber,very thick}}

\tikzset{graphborder/.style={rectangle, rounded corners=2mm, draw=black, minimum size=4mm}}

\newcommand{\graphnode}[6][]{
  \node [graphborder,#1,outer sep=1mm] (#2) {\graphnodeinner{#3}{#4}{#5}{#6}};
}

\newcommand{\graphnodeinner}[4]{%
  \begin{tikzpicture}
    \useasboundingbox (-#1/2,-#2/2) rectangle (#1/2,#2/2);
    \begin{scope}[#3]
    #4
    \end{scope}
  \end{tikzpicture}%
}

\newcommand{\edge}[3][]{\draw[#1] (#2) to (#3);}%

\newcommand{\labelnorth}[1]{\node[anchor=south] at (current bounding box.north) {${#1}$};}%
\newcommand{\labelsouth}[1]{\node[anchor=north] at (current bounding box.south) {${#1}$};}%

\tikzset{refl/.style={loop, densely dotted, distance=.8em, in=110, out=70}}
\tikzset{downrefl/.style={loop, densely dotted, distance=.8em, in=290, out=250}}

\begin{document}

\title{Characterisation of Lawvere-Tierney Topologies on Simplicial Sets, Bicolored Graphs, and Fuzzy~Sets}
\titlerunning{Characterisation of LT-Topologies on Simplicial Sets}

\author{
    Alo\"is Rosset \inst{1} \orcidID{0000-0002-7841-2318} \and
    Helle Hvid Hansen \inst{2} \orcidID{0000-0001-7061-1219} \and
    J\"org Endrullis \inst{1} \orcidID{0000-0002-2554-8270}
}
\authorrunning{A.~Rosset, H.H.~Hansen, J.~Endrullis}

\institute{
    Vrije Universiteit Amsterdam, Amsterdam, Netherlands \\
    \email{\{a.rosset, j.endrullis\}@vu.nl}
    \and
    University of Groningen, Groningen, Netherlands \\
    \email{h.h.hansen@rug.nl}
}

\maketitle

\begin{abstract}
    Simplicial sets generalise many categories of graphs.
    In this paper, we give a complete characterisation of the Lawvere-Tierney topologies on (semi-)simplicial sets, on bicoloured graphs, and on fuzzy sets. 
    We apply our results to establish that `partially simple' simplicial sets and `partially simple' graphs form quasitoposes.
    \keywords{Lawvere-Tierney topology \and Simplicial set \and Quasitopos \and Graph rewriting}
\end{abstract}

\section{Introduction}

Simplicial sets are a generalization of graphs that allow for higher-dimensional objects such as triangles ($2$-dimensional), tetrahedra ($3$-dimensional), and so forth \cite{Friedman_2012_simplicial_sets,Riehl_2011_leisurely_intro_simplicial_sets}.
They subsume, for instance, multigraphs, reflexive graphs, symmetric graphs, undirected graphs, certain forms of hypergraphs, etc.
Simplicial sets play an important role in computer science and mathematics as geometrical and topological spaces are often visualised and analysed via triangulation.
Moreover, simplicial sets form the basis of simplicial homology \cite{Eilenberg_1950_simplicial_sets_singular_homology}.

The category of simplicial sets is a presheaf category and therefore a topos~\cite{MacLane_Moerdijk_1994,Rosset_Overbeek_Endrullis_2023_Fuzzy_presheaves_are_quasitoposes}.
Toposes are extensively studied in mathematics, as they offer a unifying framework between topology, algebraic geometry, logic and set theory.
Sometimes toposes are too restrictive, as certain categories of graphs, such as simple graphs, are quasitoposes and not toposes.
\emph{Quasitoposes}~\cite{Wyler_1991} generalise toposes, and they share many of the desirable properties. 
Quasitoposes have applications in computer science.
In particular, they provide a natural setting for non-linear algebraic graph transformation \cite{Behr_Harmer_Krivine_2021_Concurrency_theorems_quasitoposes} and the study of properties such as termination and confluence~\cite{Overbeek_Endrullis_Rosset_2023_PBPO+_Quasitopos,Overbeek_2023_Termination_of_graph_transformation}.
Graph transformation also involves the transformation of higher-order structures such as triangles or tetrahedra, see further~\cite{Brown_Patterson_Hanks_Fairbanks_2022_conference,Plump_1995_On_termination,plump2018modular}.
Quasitoposes are also employed in the study of intuitionistic logic~\cite[Chapter 3]{Wyler_1991}.

Determining whether a category is a quasitopos can be very challenging.
For this purpose, it can be helpful if the category under consideration can be obtained from existing (quasi)toposes via some transformation that preserve quasitoposes.
Unfortunately, only few such transformations are known.
One well-known technique is Artin Glueing \cite{Carboni_Johnstone_1995_Artin_glueing,Johnstone_Lack_Sobocinski_2007_Quasitoposes_quasiadhesive_artin_glueing}.
However, it is not clear that Artin glueing is applicable in the setting of simplicial sets.
In this paper, we employ a different technique: Given a \emph{Lawvere-Tierney topology} $j$ on a (quasi)topos, the subcategories $\Separated_j$ and $\Sheaf_j$ of $j$-separated elements and $j$-sheaves, respectively, are again quasitoposes, as proven in \cite[Thm.~10.1]{Johnstone_1979_On_a_topological_topos} and \cite[Thm.~45.5]{Wyler_1991}.

The contributions of the paper are as follows.
We provide a complete characterisation of the Lawvere-Tierney topologies on 
(i) ($n$-dimensional) {(semi-)} simplicial sets in \cref{sec:top-simplicial}, 
(ii) bicolored graphs in \cref{sec:bicolour}, and 
(iii) fuzzy sets in \cref{sec:fuzzy_case}.
We briefly summarise the characterisations of topologies formulated in terms of their corresponding closure operators \cite[21.1]{McLarty_1992_Elementary_categories_elementary_toposes}.
\begin{enumerate}[label=(\roman*)]
    \item 
    Topologies on simplicial sets have a binary choice for each dimension $k\in\N$: the closure of the topology can either add all possible elements of dimension $k$ or not. 
	Whenever elements of dimension $k$ are added, then the separated elements are $k$-simple simplicial sets, i.e., there are no parallel elements of dimension $k$. 
	This generalises the notion of simple graphs.
	We refer to the separated elements as \emph{partially simple simplicial sets}.

    \item 
    Topologies on bicoloured graphs ensure that \emph{partially simple graphs} form a quasitopos.
	A partially simple graph is a graph with two types of edges. One type allows parallel edges, and the other does not. 

    \item 
    Topologies on $\FuzzySetDefault$ are either trivial or are in bijection with \emph{nuclei} $\phi \colon \labels \to \labels$.
\end{enumerate}

\section{Preliminaries}

We assume the reader is familiar with basic notions of category theory
\cite{Awodey_2006, MacLane_1971}.

\subsection{Graph theory}

In this paper, a \emph{graph} is synonymous with a \emph{directed multigraph}, as is usually the case in the graph rewriting literature.
In directed multigraphs, edges are directed, and parallel edges and loops are allowed.

\begin{definition}
	\label{def:graph}
	$\Graph$ denotes the category of graphs and graph homomorphisms.
	\begin{itemize}[topsep=2pt, itemsep=1pt]
		\item 
		A \myemph{graph} $A$ consists of a set $A(V)$ of vertices, a set $A(E)$ of edges, and source and target functions $A(s),A(t) \colon A(E) \to A(V)$.

		\item 
		A \myemph{graph homomorphism} $f \colon A \to B$ between graphs $A,B$ is a pair  of functions ${f_V \colon A(V) \to B(V)}$ and $f_E \colon A(E) \to B(E)$ that respect sources and targets, i.e., 
		\(
			(B(s),B(t)) \cdot f_E = (f_V \times f_V) \cdot (A(s),A(t)).
		\)
	\end{itemize}
\end{definition}

\newpage

\begin{definition}
	\label{def:reflexive_graph}
	$\ReflGraph$ denotes the category of reflexive graphs and their homomorphisms.
	\begin{itemize}[topsep=2pt, itemsep=1pt]
		\item 
		A \myemph{reflexive graph} consists of a graph $A$ and a function ${A(\refl) \colon A(V) \to A(E)}$ that satisfies $A(s) \cdot A(\refl) = A(t) \cdot A(\refl) = \id_{A(V)}$.
		In other words, $\refl$ points out for each vertex $v \in A(V)$ a loop on $v$.
		We call the loops in the image of $\refl$ \emph{reflexive loops} and denote them with dotted arrows, e.g. \smash{\(
			\begin{tikzcd}
				v 
				\arrow[loop, "\refl(v)"', densely dotted, distance=.8em, in=200, out=160]
			\end{tikzcd}
			\)},
		to distinguish them from other edges and loops that are not in the image of $\refl$.
			
		\item 
		A \myemph{reflexive graph homomorphism} $f \colon A \to B$ is a graph homomorphism that preserves reflexive loops,  i.e., $f_E \cdot A(\refl) = B(\refl) \cdot f_V$.
	\end{itemize}
\end{definition}

		

\subsection{Category theory}

\begin{definition}
    \label{def:presheaf}
    A \myemph{presheaf} on a category $\indexcategory$ is a functor $F \colon \opcat{\indexcategory} \to \Set$.
    A morphism of presheaves is a natural transformation between functors.
    The category of presheaves on $\indexcategory$ is usually denoted $\smash{\PresheafDefault}$ or $\PresheafHatDefault$. 
\end{definition}

\begin{example}
    \label{ex:graph_reflexive_graph_are_presheaf_categories}
    $\Graph$ and $\ReflGraph$ are two well-known examples of presheaves.
    Here are their respective index categories:
    \begin{equation*}
        \opcat{\indexcategory} = \quad
        \begin{tikzcd}
            0 & 1
            \arrow["t", shift left=1, from=1-2, to=1-1]
            \arrow["s"', shift right=1, from=1-2, to=1-1] 
        \end{tikzcd}
        ,\quad \text{and} \quad
        \opcat{\indexcategory} = \quad
        \begin{tikzcd}
            0 & 1
            \arrow["{\refl}"{description}, from=1-1, to=1-2]
            \arrow["t", shift left=2.5, from=1-2, to=1-1]
            \arrow["s"', shift right=2.5, from=1-2, to=1-1] 
        \end{tikzcd}
        \text{ where }
        \left\{
        \begin{aligned}
            s \cdot \refl &= \id_{V} \\
            t \cdot \refl &= \id_V
        \end{aligned}
        \right.
    \end{equation*}
\end{example}

Our goal is to show that certain subcategories of presheaf categories are quasitoposes.
We thus recall what a quasitopos is, along with the required notion of (strong) subobject classifiers which we will use throughout the paper.

\medskip
\noindent\begin{minipage}{.78\linewidth}
	\begin{definition}[{\cite[Def.~8.15]{Awodey_2006}}]
		\label{def:subobject}
		Monomorphisms into an object $C$ form a class denoted $\MonoClass(C)$.
		There is a preorder on this class: $m \leq m'$ if there exists $u$ such that $m = m'u$. 
		The equivalence relation $m \simeq m' \defeq (m \leq m' \; \meet \; m \geq m')$ is equivalent to having $m = m'u$ for some isomorphism $u$.
		The \myemph{subobjects} of $C$ are the equivalence classes of monomorphisms.
		If each representative $m \colon A \mono C$ of a subobject of $C$ is a strong monomorphism, then the subobject is called a \emph{strong subobject}.
	\end{definition}
\end{minipage}
\hfill
\begin{minipage}{.2\linewidth}
	\begin{center}
		\begin{tikzcd}[column sep=2mm, row sep=5mm,ampersand replacement=\&]
			A \&\& B \\
			\& C \& 
			\ar[from=1-1, to=2-2, tail, "m"']
			\ar[from=1-1, to=1-3, dotted, "u"]
			\ar[from=1-3, to=2-2, tail, "m'"]
		\end{tikzcd}
	\end{center}
\end{minipage}
\medskip

\begin{example}
	In $\Set$, $\Graph$, and more generally any presheaf category $\PresheafDefault$, a subobject $A' \mono A$ can always be considered w.l.o.g.~to be of the form $A' \subseteq A$.
\end{example}

\medskip
\noindent\begin{minipage}{.78\linewidth}
	\begin{definition}[{\cite[Def.~14.1]{Wyler_1991}}]
		\label{def:subobject_classifier}
		Let $\cat{C}$ be a category with finite limits, and $\mathcal{M}$ be a class of monomorphisms.
		An \myemph{$\mathcal{M}$-subobject classifier} in $\cat{C}$ is a monomorphism $\true \colon 1 \mono \Omega$, where $1$ denotes the terminal object, such that for every monomorphism $m \colon A \mono B$ in $\mathcal{M}$, there exists a unique $\charac{A} \colon B \to \Omega$, also denoted $\charac{m}$, called \myemph{characteristic function}, such that
		the diagram on the right is a pullback square.
		The object $\Omega$ is called the \myemph{classifying object}.
		When unspecified, $\mathcal{M} = \MonoClass$ is the class of all monomorphisms.
	\end{definition}
\end{minipage}
\hfill
\begin{minipage}{.2\linewidth}
	\begin{center}
		\begin{tikzcd}[column sep=5mm, row sep=5mm,ampersand replacement=\&]
			A \& 1 \\
			B \& \Omega
			\ar[from=1-1, to=1-2, "!"]
			\ar[from=1-1, to=2-1, tail, "m"']
			\ar[from=1-2, to=2-2, "\true"]
			\ar[from=2-1, to=2-2, "\charac{A}"']
			\ar[from=1-1, to=2-2, "\lrcorner"{anchor=center, pos=0.125}, draw=none]
		\end{tikzcd}
	\end{center}
\end{minipage}
\medskip

\begin{example}
    \label{ex:classifying_object_Omega_in_Set_and_Graph}
    Here are some examples of classifying object.
    \begin{itemize}[topsep=0pt]
        \item 
        In $\Set$, we have $\Omega = \set{0,1}$ and $\true \colon \set{\cdot} \to \Omega \colon \cdot \mapsto 1$.
        Given a subset $A \subseteq B$, the function $\charac{A} \colon B \to \Omega$ is the usual \emph{characteristic function} that sends every element of $A$ to $1$ and the rest to $0$.
    \end{itemize}
    \begin{mysidepicture}{3.5cm}{-4ex}{-2ex}{%
        \begin{tikzcd}[ampersand replacement=\&, scale cd=.9]
            0
                \arrow[loop, distance=1em, in=205, out=155, "0"']
                \arrow[r, "t"', bend right=15]
            \& 1
                \arrow[loop, distance=1em, in=295, out=245, "{(s, t)}"']
                \arrow[loop, distance=1em, in=115, out=65, "s \to t"']
                \arrow[l, "s"', bend right=15]
        \end{tikzcd}
    }
    \begin{itemize}[topsep=0pt]
        \item
        In $\Graph$~\cite{Vigna_2003_Simplegraph_quasitopos}, the terminal object is
        \(
            \set{ 
            \begin{tikzcd}
                \cdot \arrow[loop, distance=.8em, in=200, out=160]
            \end{tikzcd}
            }
        \), 
        the classifying object $\Omega$ is the graph shown on the right, and
        \smash{\(
            \true \colon
            \begin{tikzcd}
                \cdot 
                \arrow[loop, distance=.8em, in=200, out=160]
            \end{tikzcd}
            \mapsto 
            \begin{tikzcd}
                1
                \arrow[loop, distance=.8em, in=20, out=-20, "s \to t"'] 
            \end{tikzcd}
        \)}.
        Given a graph $A$ and a subgraph $A' \subseteq A$,  the characteristic function $\chi_{A'}$ sends vertices in $A'$ to $1$ and edges in $A'$ to $s \to t$.
        Vertices
    \end{itemize}
    \end{mysidepicture}
    \begin{itemize}[topsep=-2ex]
        \item []
        not in $A'$ are sent to $0$; and edges not in $A'$ are sent to $(s,t)$ if both endpoints are in $A'$, to $0$ if neither endpoint is in $A'$, and to $s$ or $t$ if only the source or target is in $A'$, respectively.
    \end{itemize}
\end{example}

The definitions of (elementary) \emph{toposes} and \emph{quasitoposes} are standard in the literature, see {e.g.~\cite[Def.~9]{Johnstone_Lack_Sobocinski_2007_Quasitoposes_quasiadhesive_artin_glueing} or \cite[Def.~19.3]{Wyler_1991}}.
For general references, see also~\cite{Johnstone_elephant_2002,MacLane_Moerdijk_1994,McLarty_1992_Elementary_categories_elementary_toposes}.

\begin{definition}
    \label{def:quasitopos}
    A \myemph{topos} is a category $\cat{E}$ satisfying items $1$-$3$ below. 
    Consequently, a topos also satisfies $4$-$6$. 
    A \myemph{quasitopos} is a category satisfying $1,3$-$6$.
    \begin{multicols}{2}
    \begin{enumerate}[leftmargin=.5cm, topsep=2pt, noitemsep]
        \item $\cat{E}$ admits all finite limits,
        \item $\cat{E}$ has a subobject classifier, and 
        \item $\cat{E}$ is cartesian closed.
        \item $\cat{E}$ admits every finite colimit,
        \item $\cat{E}$ has a strong-subobject classifier,
        \item $\cat{E}$ is locally cartesian closed.
    \end{enumerate}
    \end{multicols}
\end{definition}

The Yoneda embedding is central to the subobject classifier in presheaf categories.
Here is its definition.

\begin{definition}
    \label{def:contravariant_hom_functor}
    \label{def:yoneda_embedding}
    Let $\indexcategory$ be a locally small category.
    \begin{itemize}[topsep=2pt, label=$\bullet$]
        \item 
        The \myemph{contravariant $\Hom$-functor}
        ${\indexcategory(-,c) \colon \opcat{\indexcategory} \to \Set}$ is defined 
        on objects $b \in \opcat{\indexcategory}$ as the set of $\indexcategory$-morphisms $\indexcategory(b,c)$, and
        on $\opcat{\indexcategory}$-morphisms $\opcat{f}\colon a \rightarrow b$ (i.e., $f \colon b \to a$ in $\indexcategory$)
        as the precomposition by $f$:
        \[
             \indexcategory(b,c) \xrightarrow{\indexcategory(\opcat{f},c) = - \circ f} \indexcategory(a,c)\colon
             \;
             (b \xrightarrow{g} c)
             \ \xmapsto \ \ 
             (a \xrightarrow{f} b \xrightarrow{g} c).
        \]

        \item 
        The \myemph{Yoneda embedding} $y \colon \indexcategory \to \PresheafDefault$ acts as follows.
        For an object $c \in \indexcategory$, $y(c \in \indexcategory) = \indexcategory(-,c)$ is the contravariant $\Hom$-functor.
        For a morphism $g\colon b \to c$ in $\indexcategory$, $y(g) : \indexcategory(-,b) \to \indexcategory(-,c)$ is the natural transformation that postcomposes by $g$,
        i.e., for every $a \in \indexcategory$:
        \[
            \indexcategory(a,b) \xrightarrow{y(g)_a \defeq g \circ -} \indexcategory(a,c)\colon
            \;
            (a \xrightarrow{f} b)
            \ \xmapsto \ \ 
            (a \xrightarrow{f} b \xrightarrow{g} c).
        \]
    \end{itemize}
\end{definition}

The Yoneda embedding is known to give the \qmarks{building blocks} or \qmarks{generic figures} \cite{Reyes_2004_Generic_figures} of a presheaf category $\PresheafDefault$.

\begin{example}
    \label{ex:yoneda_building_block_in_graph}
    We describe the image of the Yoneda embedding in $\Graph$ and $\ReflGraph$.
    \begin{itemize}
        \item 
        In $\Graph$:
        \(
            y(0) = 
            {%
                \newcommand{\vx}{\node (vx) at (0mm,0mm) [vertex] {};}%
                \begin{tikzpicture}[node distance=15mm,baseline=-1mm,loop/.style={->,densely dotted,distance=3mm}]
                    \graphnode[]{gx}{6mm}{2mm}{}{\vx}
                \end{tikzpicture}%
            }
        \)
        and 
        \(
            y(1) =
            {%
                \newcommand{\va}{\node (va) at (-2mm,0mm) [vertex] {};}%
                \newcommand{\vb}{\node (vb) at (2mm,0mm) [vertex] {};}%
                \begin{tikzpicture}[node distance=15mm,baseline=-1mm,loop/.style={->,densely dotted,distance=3mm}]
                    \begin{scope}[local bounding box=a]
                        \graphnode[]{eab}{6mm}{2mm}{}{\va\vb\edge[->]{va}{vb}}
                    \end{scope}
                \end{tikzpicture}%
            }
        \).

        \item
        In $\ReflGraph$:
        \(
            y(0) = 
            {%
                \newcommand{\vx}{\node (vx) at (0mm,0mm) [vertex] {};}%
                \begin{tikzpicture}[node distance=15mm,baseline=-1mm,loop/.style={->,densely dotted,distance=3mm}]
                    \graphnode[]{gx}{6mm}{2mm}{}{\vx \draw [loop] (vx) to[out=180-35,in=180+35] (vx);}
                \end{tikzpicture}%
            }
        \)
        and
        \(
            y(1) =
            {%
                \newcommand{\va}{\node (va) at (-2mm,0mm) [vertex] {};}%
                \newcommand{\vb}{\node (vb) at (2mm,0mm) [vertex] {};}%
                \begin{tikzpicture}[node distance=15mm,baseline=-1mm,loop/.style={->,densely dotted,distance=3mm}]
                    \begin{scope}[local bounding box=a]
                        \graphnode[]{eab}{10mm}{2mm}{}{\va\vb\edge[->]{va}{vb}
                            \draw [loop] (va) to[out=180-35,in=180+35] (va);
                            \draw [loop] (vb) to[out=  0-35,in=  0+35] (vb);}
                    \end{scope}
                \end{tikzpicture}%
            }
        \).
    \end{itemize}
    
\end{example}

\subsection{Topologies}
\label{subsec:topologies}

Toposes admit enough structure to have an internal logic, whose principles are those of \emph{intuitionistic logic}~\cite{Goldblatt_2006_Topoi,McLarty_1992_Elementary_categories_elementary_toposes}, and the same observation can be made about quasitoposes~\cite[Chapter 3]{Wyler_1991}.
The \emph{truth values} are the global elements $1 \to \Omega$ of the classifying object $\Omega$.
For instance, $\true \colon 1 \to \Omega$ is already part of the definition of subobject classifier (\cref{def:subobject_classifier}).
The other standard logical connectives $\false \colon 1 \to \Omega$, negation $\lnot \colon \Omega \to \Omega$, conjunction, implication and disjunction $\meet, \Rightarrow, \join \colon \Omega \times \Omega \to \Omega$ can all be defined as characteristic functions of certain morphisms, see e.g.~\cite[Section 6.6]{Goldblatt_2006_Topoi}.

We give some details about these connective in the case of presheaf categories $\PresheafDefault$.
The terminal object $1$ is given by $1(c) = \set{\cdot}$ for all $c \in \cat{C}$.
When $\indexcategory$ is small, there is a functor that returns the set of all subobjects 
$\Sub_{\PresheafHatDefault}(B) \defeq \set{\text{subobjects } m \colon A \mono B}$.
For presheaf categories, the type of $\Sub$ is $\smash{\opcat{\PresheafHatDefault}} \to \HeytingAlgebra$, because the set of subobjects of a presheaf always form a Heyting algebra, see e.g.~\cite[I.8.~Prop.~5]{MacLane_Moerdijk_1994} or \cite[1.4.13]{Johnstone_elephant_2002}.
The classifying object is the presheaf $\Omega \defeq \Sub_{\PresheafHatDefault} (y(-)) \colon \opcat{\indexcategory} \to \Set$, and we have $\true_c \colon 1(c) \to \Omega(c) \colon \cdot \mapsto y(c)$ \cite[I.3, I.4]{MacLane_Moerdijk_1994}.
For $c \in \indexcategory$, $\Omega(c)$ is thus the Heyting algebra of all subpresheaves of $y(c)$, which are also called \textit{sieves}.
The connectives $\meet, \join \colon \Omega \times \Omega \to \Omega$ are then given pointwise as the intersection, respectively union, of two subpresheaves.
For $\lnot$ and $\Rightarrow$, see \cite[I.8.~Prop.~5]{MacLane_Moerdijk_1994}.

\begin{example}
    \label{ex:Hasse_diagram_Omega_in_Graph}
    We illustrate the classifying object $\Omega$ and its Heyting algebra structure in $\Graph$.
    Each Heyting algebra $(\labels, \leq)$ is visualised as a horizontal Hasse diagram, where the order increases to the right.\\[2mm]
    \[
        \begin{array}{rcl}
			\Omega(0) = \Sub(y(0)) &=&
			{%
				\newcommand{\vx}{\node (vx) at (0mm,0mm) [vertex] {};}%
				\begin{tikzpicture}[node distance=15mm,baseline=-1mm, scale=1, every node/.style={scale=1}]
					\graphnode[]{empty}{6mm}{2mm}{}{\labelnorth{0}}
					\graphnode[right of=empty]{gx}{6mm}{2mm}{}{\vx \labelnorth{1}}
					\begin{scope}[metaedge]
						\draw (empty) to (gx);
					\end{scope}
				\end{tikzpicture}%
			}
			\\[6mm]
			\Omega(1) = \Sub(y(1)) &=&
			{%
				\newcommand{\va}{\node (va) at (-2mm,0mm) [vertex] {};}%
				\newcommand{\vb}{\node (vb) at (2mm,0mm) [vertex] {};}%
				\begin{tikzpicture}[node distance=15mm,baseline=-1mm, scale=1, every node/.style={scale=1}]
					\begin{scope}[local bounding box=a]
						\graphnode[]{empty}{6mm}{2mm}{}{\labelnorth{0}}
						\graphnode[right of=empty,yshift=3mm]{ga}{6mm}{2mm}{}{\va\labelnorth{s}}
						\graphnode[right of=empty,yshift=-3mm]{gb}{6mm}{2mm}{}{\vb\labelsouth{t}}
						\graphnode[right of=ga,yshift=-3mm]{gab}{6mm}{2mm}{}{\va\vb\labelnorth{(s,t)}}
						\graphnode[right of=gab]{eab}{6mm}{2mm}{}{\va\vb\edge[->]{va}{vb}\labelnorth{s\to t}}
						\begin{scope}[metaedge]
							\draw (empty) to (ga);
							\draw (empty) to (gb);
							\draw (ga) to (gab);
							\draw (gb) to (gab);
							\draw (gab) to (eab);
						\end{scope}
					\end{scope}
				\end{tikzpicture}%
			}
			\\[7mm]
		\end{array}
    \]
\end{example}

The next definition recalls the notion of Lawvere-Tierney topologies (LT-topologies), also referred to as \emph{(local) operators} or \emph{(local) modalities}.
For references, see e.g.~\cite{Lawvere_1970_LT_topologies} \cite[V.1]{MacLane_Moerdijk_1994} or \cite[A.4.4]{Johnstone_elephant_2002}.
In logical terms, LT-topologies address \qmarks{local truth} within the internal logic of a topos.

\begin{definition} 
    \label{def:topology_on_topos}
    A (Lawvere-Tierney) \myemph{topology} on a topos $\cat{E}$ is an $\cat{E}$-morphism $j \colon \Omega \to \Omega$ satisfying \\*[\medskipamount]
    \begin{enumerate*}[label={(\arabic*)}, itemjoin={,\hspace*{15mm}}, topsep=0pt]
        \item
        \label{it:topology_true}
        $j \circ \true = \true$

        \item 
        \label{it:topology_idempotent}
        $j \circ j = j$

        \item 
        \label{it:topology_meet}
        $j \circ \land = \land \circ (j \times j).$
    \end{enumerate*}
\end{definition}

\begin{remark}
	\label{lem:topology_monotone}
	In every presheaf topos, $\Omega$ is a Heyting algebra, hence any LT-topology $j$ can also be called an idempotent bounded meet-semilattice homomorphism.
	In particular, it follows directly from \ref{it:topology_meet} that $j$ is monotone:
	\begin{enumerate}[label=$(4)$, topsep=0pt]
		\item \label{it:topology_monotone}
		$a \leq b \implies j(a) \leq j(b)$.
	\end{enumerate}
\end{remark}

There is a correspondence between LT-topologies and closure operators.
We recall the definition of closure operators and then state the correspondence.

\begin{definition}[{\cite[21.1]{McLarty_1992_Elementary_categories_elementary_toposes}}]
    \label{def:closure_operator}
    A \myemph{closure operator} on a (quasi)topos $\cat{E}$ is a family of mappings $\topology_A \colon \MonoClass(A) \to \MonoClass(A)$ for each object $A$ satisfying 
    \ref{def:topology_on_quasitopos_1_increasing}-\ref{def:topology_on_quasitopos_5_strong} 
    below.
    I.e., every mono $m \colon A' \rightarrowtail A$ is sent to another mono $\topology_A(m) \colon \topologybar{A'} \rightarrowtail A$ with same codomain.
    When no confusion can arise, we sometimes write $\topologybar{m}$ or $\topologybar{A'}$ instead of $\topology_A(m)$.
    \begin{enumerate}[label=(\roman*), noitemsep, topsep=0pt]
        \item
        \label{def:topology_on_quasitopos_1_increasing}
        $\topology_A $ is increasing\footnote{Also known as extensive or inflationary~\cite{McLarty_1992_Elementary_categories_elementary_toposes}.}: $m \leq \topologybar{m}$,
        
        \item
        \label{def:topology_on_quasitopos_2_idempotent}
        $\topology_A $ is idempotent: $\topologybar{\topologybar{m}} \simeq \topologybar{m}$,
    
        \item 
        \label{def:topology_on_quasitopos_3_monotone}
        $\topology_A $ is monotone: $m_1 \leq m_2$ implies $\topologybar{m_1} \leq \topologybar{m_2}$,
    
        \item 
        \label{def:topology_on_quasitopos_4_stable_under_PB}
        $\topology_A $ is stable under pullbacks: for all $f : B \to A$, $\topologybar{\pullbackalong{f} m} \simeq \pullbackalong{f}(\topologybar{m})$,
    
        \item 
        \label{def:topology_on_quasitopos_5_strong}
        $\topology_A $ preserves strongness: $m$ strong implies $\topologybar{m}$ strong.
    \end{enumerate}
\end{definition}

\begin{lemma}[{\cite[Thm 21.1]{McLarty_1992_Elementary_categories_elementary_toposes}, \cite[Thm 45.5]{Wyler_1991}}]
    \label{lem:equivalence_LT-topologies_and_closure_operators}
    In a topos $\cat{E}$, there is a bijection between LT-topologies and closure operators.
    Given an LT-topology $j$ and $A' \mono A$, the closure $\overline{A'} \mono A$ is defined by $\charac{\overline{A'}} \defeq j \circ \charac{A'}$.
    Given a closure operator $\tau$, the corresponding LT-topology is defined by $j \defeq \charac{\tau(\true)} \colon \Omega \to \Omega$.
\end{lemma}

The strongness-preservation requirement \ref{def:topology_on_quasitopos_5_strong} is superfluous in a topos, because every mono is strong, but we state it to pave the way for the quasitopos setting in \cref{sec:fuzzy_case}.
Intuitively, the closure $\topologybar{A'}$ contains those elements of $A$ that are \qmarks{locally} in $A'$.

In the next two examples, we give the list of all topologies on $\Set$ and $\Graph$, identifying them with their corresponding closure operators.
We already employ a notation that we will continue in subsequent sections: in a topology $j^w$, the superscript $w$ is a binary string that indicates which elements are added by the closure operator corresponding to $j^w$.

\begin{example}
    \label{ex:topologies_on_set}
    There are exactly two topologies on $\Set$.
    As seen in \cref{ex:classifying_object_Omega_in_Set_and_Graph}, $\Omega = \set{0,1}$ and $\true=1$.
    By \ref{it:topology_true}, we must have $j(1) = 1$.
    Therefore, there is a choice for $j(0) \in \set{0,1} = \Omega$ which gives us the two topologies:
    \begin{enumerate}[topsep=0pt, noitemsep]
        \item 
        The \emph{discrete} topology $j^0 = \id_\Omega\colon 0 \mapsto 0$, which adds nothing.
        
        \item 
        The \emph{trivial} topology $j^1 = \true_\Omega\colon 0 \mapsto 1$, which adds all elements.\footnote{
            The notation $\true_X$ is common to represent the morphism $X \smash{\xrightarrow{!}} 1 \smash{\xrightarrow{\true}} \Omega$.
        }
    \end{enumerate}
\end{example}

In fact, the discrete and trivial topologies exist in every topos.

\begin{example}
    \label{ex:topologies_on_graphs}
    There are exactly $4$ topologies on $\Graph$~\cite{Vigna_2003_Simplegraph_quasitopos}.
    The binary string $w=w_0 w_1$ in the superscript of the topologies is there to indicate with $w_0$ whether vertices are added, and with $w_1$ whether the edges are added.
    Recall its classifying object $\Omega$ in \cref{ex:classifying_object_Omega_in_Set_and_Graph} and its Heyting alegbra structure in \cref{ex:Hasse_diagram_Omega_in_Graph}.
    For each topology, we describe the closure $\topologybar{A'}$ of a subgraph $A' \subseteq A$.
    \begin{enumerate}[topsep=0pt, noitemsep]
        \item 
        The \emph{discrete} topology, $j^{00} = \id_{\Omega}$ adds nothing, $\topologybar{A'} = A'$
        
        \item 
        The topology $j^{10}$ adds to $A'$ all vertices of $A$, i.e., $\topologybar{A'}(0) = A(0)$.
        
        \item 
        The \emph{double negation} topology, $j^{01} = \lnot \lnot$ adds all edges whose source and target are already in the subgraph:
        \begin{equation}
            \label{eq:closure_double_negation_topology_on_Graph}
            \overline{A'}(E) = A(E) \cap (A'(V) \times A'(V)).
        \end{equation}
        
        \item The \emph{trivial} topology, $j^{11} = \true_\Omega$ adds all vertices and all edges, $\topologybar{A'} = A$.
    \end{enumerate}
\end{example}

he notions of density and separatedness from standard topology, can also be defined for LT-topologies.
Recall that $\Q$ is called dense in $\R$ because closing it under the standard metric topology gives the entire space $\R$.
More precisely, for every real number $r \in R$ and every open interval around $r$, there exists a rational number $q \in Q$ in the open interval.
In standard topology, recall also that every separated space $B$ (a.k.a.\ Hausdorff space) has the property that functions $f \colon A \to B$ are fully determined by their restriction to any dense subset of their domain.

\begin{definition}[{\cite[p.~221]{MacLane_Moerdijk_1994}}]
    \label{def:dense_subobject}
    Given a topology $j$ on a topos, a subobject $A' \stackrel{}{\mono} A$ is \myemph{$j$-dense} if $\topologybar{A' \subseteq A} = (A \subseteq A)$.
\end{definition}

\begin{example}
    \label{ex:dense_subobjects}
    Here are two obivous examples of dense subobjects.
    \begin{itemize}[topsep=0pt, noitemsep]
        \item 
        For the discrete topology, only $A \mono A$ is dense in itself.

        \item 
        For the trivial topology, every $A' \mono A$ is dense.
    \end{itemize}
\end{example}

\medskip
\noindent\begin{minipage}{.8\linewidth}
\begin{definition}[{\cite[p.~223]{MacLane_Moerdijk_1994}}]
    \label{def:separated_elements_complete_elements_sheaves}
    Let $\cat{E}$ be a topos, $j$ be a topology on $\cat{E}$, and $B$ be an object in $\cat{E}$.
    \begin{itemize}[topsep=0pt, noitemsep]
        \item $B$ is \myemph{$j$-separated } if for every $A$, for every $j$-dense subobject $m\colon A' \mono A$ and every morphism $f\colon A' \to B$, there exists at most one factorisation $g\colon A \to B$ of $f$ through $m$.
        \item $B$ is \myemph{$j$-complete} if such a factorisation always exists.
        \item $B$ is a \myemph{$j$-sheaf} if it is $j$-separated and $j$-complete.
    \end{itemize}
    We denote with $\Separated_j$ and $\Sheaf_j$ the full subcategories of $\cat{E}$ consisting of the $j$-separated objects, and the $j$-sheaves, respectively.
\end{definition}
\end{minipage}
\hfill
\begin{minipage}{.3\linewidth}
    \centering
    \begin{tikzcd}[ampersand replacement=\&]
        A' \& \\
        A \& B
        \ar[from=1-1, to=2-2, "f"]
        \ar[from=1-1, to=2-1, tail, "m"']
        \ar[from=2-1, to=2-2, dotted, "g"']
    \end{tikzcd}
\end{minipage}
\medskip

One of the principal motivations for studying LT-topologies is that it enables the identification of two subcategories that can be classified as quasitoposes.

\begin{lemma}[{\cite[Thm.~10.1]{Johnstone_1979_On_a_topological_topos}}]
    \label{lem:separated_elements_and_sheaves_form_quasitoposes}
    Let $\cat{E}$ be a topos and $j$ be a topology on $\cat{E}$.
    Then $\Separated_j$ and $\Sheaf_j$ are quasitoposes.
\end{lemma}

\begin{example}
    \label{ex:separated_elements_and_sheaves_in_Graph_and_discrete_trivial}
    Some examples of separated elements and sheaves in a topos:
    \begin{itemize}[topsep=0pt, noitemsep]
        \item 
        For the double negation topology on $\Graph$ (\cref{ex:topologies_on_graphs}), the separated elements are simple graphs, which therefore form a quasitopos \cite{Vigna_2003_Simplegraph_quasitopos}.
        
        \item 
        For the discrete topology, every object $B$ is a sheaf.

        \item 
        For the trivial topology, since the morphism from the initial object $0 \mono A$ is dense, a sheaf $B$ has thus exactly one morphism $A \to B$ for any $A$, meaning that $B$ is a terminal object.
        A separated object $B$ has at most one morphism $A \to B$ for any $A$, which means that $B$ is a \emph{subterminal objects}.
        Equivalently, subterminal objects can also be described as subobjects of $1$.
        In $\Set$, the sheaves are $\emptyset$ and $\set{\cdot}$.
        In $\Graph$, the sheaves are $\emptyset, \set{\cdot}$ and
        \(
            \set{ 
            \begin{tikzcd}
                \cdot \arrow[loop, distance=.8em, in=200, out=160]
            \end{tikzcd}
            }
        \).
    \end{itemize}
\end{example}

\section{Toplogies on Simplicial Sets}\label{sec:top-simplicial}

\subsection{Simplicial Sets}
\label{subsec:def_simplicial_sets}

As pointed out in \cref{ex:graph_reflexive_graph_are_presheaf_categories}, $\Graph$ and $\ReflGraph$ are presheaf categories $\PresheafDefault$.
\emph{Simplicial sets} are a generalisation of (reflexive) graphs obtained by extending the index category $\indexcategory$ with higher-dimensional elements beyond vertices ($0$-dimensional) and edges ($1$-dimensional): triangles ($2$-dimensonal), tetrahedra ($3$-dimensional), and so forth.
The generalisation of triangles to dimension $n$ is called an $n$-simplex, hence the name simplicial sets. 
The index category for simplicial sets looks as follows:
\[
    \opcat{\Delta} = \qquad
    \begin{tikzcd}[column sep=large]
        0 & 1 & 2 & \cdots
        \arrow["{s^0_0}"{description}, from=1-1, to=1-2]
        \arrow["d^1_0", shift left=3, from=1-2, to=1-1]
        \arrow["d^1_1"', shift right=3, from=1-2, to=1-1]
        \arrow["{s^1_0}"{description, pos=.7}, shift left=3, from=1-2, to=1-3]
        \arrow["{s^1_1}"{description, pos=.7}, shift right=3, from=1-2, to=1-3]
        \arrow["{d^2_0}"{description, pos=.7}, shift left=6, from=1-3, to=1-2]
        \arrow["{d^2_2}"{description, pos=.7}, shift right=6, from=1-3, to=1-2]
        \arrow["{d^2_1}"{description, pos=.7}, from=1-3, to=1-2]
        \arrow["\cdots"{description}, draw=none, from=1-3, to=1-4]
        \arrow["\cdots"{description}, shift left=3, draw=none, from=1-3, to=1-4]
        \arrow["\cdots"{description}, shift right=3, draw=none, from=1-3, to=1-4] 
\end{tikzcd}\]
The arrows $d^n_i$ are called \myemph{faces} and $s^n_i$ are called \myemph{degeneracies}, with the superscripts often omitted.
The functions $(d^1_1,d^1_0)$ are the source and target functions $(s,t)$.
The subscripts are in decreasing order, because a face map $d_i$ omits/projects away the vertex at position~$i$.%
\footnote{
    We count starting from zero, i.e., position $0$, $1$, etc.
}
Therefore, the source of a (directed) edge $(v_0,v_1)$ is $d_1(v_0,v_1) = v_0$.

Similarly, the triangles that we consider are ordered triples $(e_0, e_1, e_2)$ of edges that share endpoints in a certain way: 
$e_0 =(v_1,v_2)$, $e_1 = (v_0, v_2)$ and $e_2 = (v_0, v_1)$. 
Hence a triangle can also be described by the ordered triple of its vertices $(v_0,v_1,v_2)$.
Similarly as for source and target, the faces of a triangle are its three edges.
The sharing of endpoints of edges in a triangle is specified by the following equation:
\begin{equation}
	\label{eq:first_simplicial_identity}
	d^{n-1}_i d^n_j = d^{n-1}_{j-1} d^n_i \hspace*{21pt} \text{ if }i < j 
\end{equation}

Dually, a degeneracy map $s_i$ acts by duplicating the $i$\tss{th} vertex.
E.g.~$s_0(v_0) = (v_0,v_0)$ is a degenerate edge on vertex $v_0$, which we denoted $\refl(v_0)$ in $\ReflGraph$ and visualised by a reflexive loop, and $s_1 (v_0,v_1) = (v_0,v_1,v_1)$ is a degenerate triangle.%
\footnote{
	In geometry, a triangle is called \emph{degenerate} if its three points are colinear, i.e., on the same line.
}
Degenerate edges, triangles, etc., are often not depicted for practical reasons.
The degeneracy maps must satisfy the equations in~\eqref{eq:simplicial_identities}. 
\begin{equation}
    \label{eq:simplicial_identities}
    \begin{aligned}
        s_i s_j &= s_{j+1} s_i \text{ if }i \leq j \\
        d_i s_j &= 
        \begin{cases}
            s_{j-1} d_i &\text{ if }i < j \\
            \id &\text{ if } i=j \text{ or } i=j+1 \\
            s_j d_{i-1} &\text{ if } i > j+1
        \end{cases}
    \end{aligned}
\end{equation}
The equations in \eqref{eq:first_simplicial_identity} and \eqref{eq:simplicial_identities} are together called the \emph{simplicial identities}.
Here are some examples of the simplicial identities for a given triangle:

\[
    \begin{tikzcd}[ampersand replacement=\&, column sep=1.6em, row sep=1em]
        \&\& |[alias=b]| 1 \\
        |[alias=a]|  0 \\
        \&\& |[alias=c]| 2
        \arrow[""{name=0, anchor=center, inner sep=0}, from=2-1, to=3-3, "{(0,2)}"']
        \arrow[from=2-1, to=1-3, "{(0,1)}"]
        \arrow[from=1-3, to=3-3, "{(1,2)}"]
        \arrow[shorten <=2pt, "{(0,1,2)}" description, pos=.4, draw=none, from=0, to=1-3] 
    \end{tikzcd} 
    \qquad
    \begin{aligned}
        &d_1 d_2 (0,1,2) = d_1 (0,1) = 0 = d_1 (0,2) = d_1 d_1 (0,1,2) \\
        &s_0 s_0 (0) = s_0 (0,0) = (0,0,0) = s_1 (0,0) = s_1 s_0 (0) \\
        &d_0 s_1 (1,2) = d_0 (1,2,2) = (2,2) = s_0 (2) = s_0 d_0 (1,2)
    \end{aligned}
\]

The index category $\Delta$ is called the \myemph{simplex category}.
We have given an intuitive explanation of its opposite category $\opcat{\Delta}$ with a graph approach in mind.
For a detailed definition of $\Delta$, see e.g.~\cite{Friedman_2012_simplicial_sets,Riehl_2011_leisurely_intro_simplicial_sets}.
We will consider the following subcategories of $\Delta$:

\begin{itemize}[topsep=0pt, noitemsep]
    \item $\Delta\semi$ is the wide subcategory containing only the $d_i$ arrows.
    \item $\Delta\leqn$ is the full subcategory of objects $0, \ldots, n$.
    \item $\Delta\leqn\semi$ is the intersection between $\Delta\semi$ and $\Delta\leqn$.
\end{itemize}

\begin{remark}
	\label{rem:augmented_simplex_category}
	Certain references use the notation $\Delta\semi$ to denote instead the \emph{augmented} simplex category, which has an additional object $-1$ and a face morphism $0 \to -1$.
	Presheaves with the augmented simplex category as index category are called the \emph{augmented simplicial sets}.
\end{remark}

\begin{definition}
    We consider the following presheaf categories in this paper:
     \begin{enumerate}[(1), topsep=0pt, noitemsep]
        \item \myemph{Simplicial sets} are presheaves $\opcat{\Delta} \to \Set$.
        \item \myemph{Semi-simplicial sets} are presheaves $\opcat{\Delta\semi} \to \Set$.
        \item \myemph{$n$-dimensional simplicial sets} are presheaves $\opcat{\big(\Delta\leqn\big)} \to \Set$.
        \item \myemph{$n$-dimensional semi-simplicial sets} are presheaves $\smash{\opcat{\big(\Delta\leqn\semi\big)}} \to \Set$.
    \end{enumerate}
    Together with natural transformations as morphisms, they form the categories $(1)\ \sSet, (2)\ \semisSet, (3)\ \ndimsSet$ and $(4)\ \ndimsemisSet$.
\end{definition}

\begin{example}
    \label{ex:simplicial_sets}
    As already discussed, here are some examples of simplicial sets:
    \begin{itemize}[topsep=0pt, noitemsep]
        \item $\ndimsemisSet[0]=\ndimsSet[0] = \Set$ is the category of sets
        \item $\ndimsemisSet[1] = \Graph$ is the category of graphs, and
        \item $\ndimsSet[1] = \ReflGraph$ is the category of reflexive graphs.
    \end{itemize}
\end{example}

The next definition is the generalisation to all dimensions of the pair (source, target) for an edge.

\begin{definition}
    \label{def:incidence_tuple}
    Let $n \in \N$ and $\cat{D} \in \set{\ndimsemisSet, \ndimsSet, \semisSet, \sSet}$.
    For a simplicial set $X \in \cat{D}$, for $k \leq n$, and for a $k$-simplex $x \in X(k)$, we define the \myemph{incidence tuple} of $x$ as the $(k+1)$-tuple $\big(d_k(x), \ldots, d_0(x)\big)$.
\end{definition}

Take an arbitrary $n \in \Delta$.
The simplicial set $y(n) = \Delta(-,n) \in \sSet$ is the \emph{standard $n$-simplex}~\cite{Riehl_2011_leisurely_intro_simplicial_sets}.
The classifying object $\Omega(n)$ is given by the Heyting algebra of subobjects of $y(n)$, i.e., $\Omega(n) = \Sub(y(n))$ (as discussed in \cref{subsec:topologies}).

In \cref{ex:yoneda_building_block_in_graph}, we presented visualisations of $y(0)$ and $y(1)$, and in \cref{ex:Hasse_diagram_Omega_in_Graph} we provided visualisations of $\Omega(0)$ and $\Omega(1)$, including their Heyting algebra structures depicted through Hasse diagrams.
Here is a visualisation for the dimension $n=2$, where $y(2)$ is an oriented triangle, and $\Omega(2)$ is its subgraph Heyting algebra.
We omit the degenerate simplices as is usually the case.

\begin{equation}
    \begin{array}{rcl}
    y(2) &=&
    \begin{tikzcd}[
        ampersand replacement=\&,
        column sep=.8em, 
        row sep=.5em, 
        scale cd=.8,
        baseline={(a.base)},
        execute at end picture={
            \scoped[on background layer]
            \fill[rounded corners = 2mm, color=lightgray] (a.center) -- (b.center) -- (c.center) -- cycle;
        }
    ]
        \&\& |[alias=b]| \cdot \\
        |[alias=a]|  \cdot \\
        \&\& |[alias=c]| \cdot
        \arrow[""{name=0, anchor=center, inner sep=0}, from=2-1, to=3-3]
        \arrow[from=2-1, to=1-3]
        \arrow[from=1-3, to=3-3]
    \end{tikzcd} 
    \\
    \Omega(2) &=&
    {%
    \newcommand{\va}{\node (va) at (-60:2.5mm) [vertex] {};}%
    \newcommand{\vb}{\node (vb) at (60:2.5mm) [vertex] {};}%
    \newcommand{\vc}{\node (vc) at (180:2.5mm) [vertex] {};}%
    \begin{tikzpicture}[node distance=15mm,baseline=-1mm, scale=0.8, every node/.style={scale=0.8}]%
        \begin{scope}
            \graphnode[]{empty}{6mm}{6mm}{}{}
            \graphnode[right of=empty,yshift=10mm]{ga}{6mm}{6mm}{xshift=1mm}{\va}
            \graphnode[right of=empty,yshift=0mm]{gb}{6mm}{6mm}{xshift=1mm}{\vb}
            \graphnode[right of=empty,yshift=-10mm]{gc}{6mm}{6mm}{xshift=1mm}{\vc}
            \graphnode[right of=ga]{gab}{6mm}{6mm}{}{\va\vb}
            \graphnode[right of=gb]{gac}{6mm}{6mm}{}{\va\vc}
            \graphnode[right of=gc]{gbc}{6mm}{6mm}{}{\vb\vc}
            
            \begin{scope}[node distance=20mm]    
                \graphnode[right of=gab]{eab}{6mm}{6mm}{}{\va\vb\edge{va}{vb}}
                \graphnode[right of=gac]{eac}{6mm}{6mm}{}{\va\vc\edge{va}{vc}}
                \graphnode[right of=gbc]{ebc}{6mm}{6mm}{}{\vb\vc\edge{vb}{vc}}
                
                \graphnode[right of=eab]{eabx}{6mm}{6mm}{}{\va\vb\vc\edge{va}{vb}}
                \graphnode[right of=eac]{eacx}{6mm}{6mm}{}{\va\vb\vc\edge{va}{vc}}
                \graphnode[right of=ebc]{ebcx}{6mm}{6mm}{}{\va\vb\vc\edge{vb}{vc}}
            \end{scope}
            
            \graphnode[at=(eac),yshift=25mm]{gabc}{6mm}{6mm}{}{\va\vb\vc}
            
            \graphnode[right of=eabx]{hbc}{6mm}{6mm}{}{\va\vb\vc\edge{va}{vc}\edge{va}{vb}}
            \graphnode[right of=eacx]{hac}{6mm}{6mm}{}{\va\vb\vc\edge{va}{vb}\edge{vb}{vc}}
            \graphnode[right of=ebcx]{hab}{6mm}{6mm}{}{\va\vb\vc\edge{vc}{va}\edge{vb}{vc}}
            \graphnode[right of=hac]{triangle}{6mm}{6mm}{}{\va\vb\vc\edge{vc}{va}\edge{va}{vb}\edge{vb}{vc}}
            \graphnode[right of=triangle]{full}{6mm}{6mm}{}{
                \va\vb\vc
                \begin{pgfonlayer}{background}
                    \draw [draw=none,fill=black!25,rounded corners=0mm] (va.center) to (vb.center) to (vc.center) to cycle; 
                \end{pgfonlayer}
                \edge{vc}{va}\edge{va}{vb}\edge{vb}{vc}
            }
            
            \begin{scope}[metaedge]
                \draw (empty) to (ga);
                \draw (empty) to (gb);
                \draw (empty) to (gc);
                \draw (ga) to (gab); \draw (ga) to (gac);
                \draw (gb) to (gab); \draw (gb) to (gbc);
                \draw (gc) to (gac); \draw (gc) to (gbc);
                \draw (gab) to (eab);
                \draw (gac) to (eac);
                \draw (gbc) to (ebc);
                \draw (eab) to (eabx);
                \draw (eac) to (eacx);
                \draw (ebc) to (ebcx);
                \draw (eabx) to (hbc); \draw (eabx) to (hac);
                \draw (eacx) to (hab); \draw (eacx) to (hbc); 
                \draw (ebcx) to (hab); \draw (ebcx) to (hac);
                \draw (hab) to (triangle);
                \draw (hac) to (triangle);
                \draw (hbc) to (triangle);
                \draw (triangle) to (full);
                
                \begin{scope}[densely dotted]
                    \draw (gab.east) to (gabc.west);
                    \draw (gac.east) to (gabc.west);
                    \draw (gbc.east) to (gabc.west);
                    \draw (gabc.east) to (eabx.west);
                    \draw (gabc.east) to (eacx.west);
                    \draw (gabc.east) to (ebcx.west);
                \end{scope}
            \end{scope}
        \end{scope}
    \end{tikzpicture}%
    }
    \end{array}
\end{equation}

The goal of this section is to characterise all topologies on $\sSet$.
We first give in \cref{subsec:leqn_semi} a full characterisation of the topologies $j: \Omega \to \Omega$ on $\ndimsemisSet$.
Then, in \cref{subsec:leqn}, we observe which $j$'s from the previous subsection are topologies on $\ndimsSet$ when we also consider the degeneracies $s_i$.
Finally, in \cref{subsec:semi_and_general}, we explain how to remove the dimension restrictions and characterise topologies on $\semisSet$ and $\sSet$ using the two previous subsections.

\subsection{Topologies on \texorpdfstring{$\ndimsemisSet$}{n-dimensional semi-sSet} }
\label{subsec:leqn_semi}

In this subsection, we fix an integer $n \in \N$.
The goal of this subsection is to characterise all topologies $j : \Omega \to \Omega$ on $\ndimsemisSet$.
We first need some preliminary definitions and lemmas.

\begin{definition}
	\label{def:ith_face}
	For all $k,i \in \N$ such that $0 \leq k < k+1 \leq n$ and $0 \leq i \leq k+1$, the \myemph{$i$\tss{th} face} of $y(k+1)$ is the least subpresheaf $x$ of $y(k+1)$ such that 
	\[
		d^{k+1}_i 
		\in 
		x(k) 
		\subseteq y(k+1)(k) 
		= 
		\opcat{(\Delta\leqn\semi)}(k+1,k).
	\]
	We denote the $i$\tss{th} face of $y(k+1)$ by $\hat{i}_{k+1}$.
	The subscript is sometimes omitted if it is clear from context.
	Note that for all $0 \leq l \leq k+1$, $\hat{i}_{k+1}(l) \subseteq y(k+1)(l)=\opcat{(\Delta\leqn\semi)}(k+1,l)$.
	Concretely, the components of $\hat{i}_{k+1}$ are as follows:
	\begin{equation}
	\label{eq:i_hat}
	\begin{array}{lcl}
		\hat{i}_{k+1}(k+1)	&=& \emptyset \\
		\hat{i}_{k+1}(k)	&=& \set{d^{k+1}_i} \\
		\hat{i}_{k+1}(k-1)	&=& \setvbar{ d^{k}_{i_{k}} \cdot d^{k+1}_i  }{0 \leq i_{k} \leq k} \\
		\ldots \\
		\hat{i}_{k+1}(0)	&=& \setvbar{d^{1}_{i_{1}} \cdot \ldots \cdot d^{k}_{i_{k}} \cdot d^{k+1}_i}{0 \leq i_1 \leq 1, \ldots, 0 \leq i_{k} \leq k}.
	\end{array}
	\end{equation}
\end{definition}

\begin{example}
	Here are some examples of faces.
	\begin{itemize}
		\item 
		Take $n=k+1=1$, recall the Yoneda embedding for $\Graph$ from \cref{ex:yoneda_building_block_in_graph}.
		The notations correspond as follows: $0=V$, $1=E$, $d^1_1=s$, and $d^1_0=t$.
		\begin{align*}
			&\left\{
			\begin{aligned}
					y(0)(0) &= \set{\id_0} \\
					y(0)(1) &= \emptyset
				\end{aligned}
			\right.
			&
			y(0) &=
			{%
				\newcommand{\vx}{\node (vx) at (0mm,0mm) [vertex] {};}%
				\begin{tikzpicture}[node distance=15mm,baseline=-1mm,loop/.style={->,densely dotted,distance=3mm}]
					\graphnode[]{gx}{6mm}{2mm}{}{\vx}
				\end{tikzpicture}%
			}
			\\[2mm]
			&\left\{
			\begin{aligned}
				y(1)(0) &= \set{d^1_1, d^1_0} \\
				y(1)(1) &= \set{\id_1}
			\end{aligned}
			\right.
			&
			y(1) &=
			{%
				\newcommand{\va}{\node (va) at (-2mm,0mm) [vertex] {};}%
				\newcommand{\vb}{\node (vb) at (2mm,0mm) [vertex] {};}%
				\begin{tikzpicture}[node distance=15mm,baseline=-1mm,loop/.style={->,densely dotted,distance=3mm}]
					\begin{scope}[local bounding box=a]
						\graphnode[]{eab}{6mm}{2mm}{}{\va\vb\edge[->]{va}{vb}}
					\end{scope}
				\end{tikzpicture}%
			}
		\end{align*}
		The standard $1$-simplex $y(1)$ has two faces:
		\begin{itemize}[topsep=2pt, noitemsep]
			\item 
			\makebox[3.2cm][l]{the source vertex:}
			\(
				\hat{1}_1 =
				{%
					\newcommand{\va}{\node (va) at (-2mm,0mm) [vertex] {};}%
					\newcommand{\vb}{\node (vb) at (2mm,0mm) [vertex] {};}%
					\begin{tikzpicture}[node distance=15mm,baseline=-1mm,loop/.style={->,densely dotted,distance=3mm}]
						\begin{scope}[local bounding box=a]
							\graphnode[]{eab}{6mm}{2mm}{}{\va}
						\end{scope}
					\end{tikzpicture}%
				} 
				\subseteq
				{%
					\newcommand{\va}{\node (va) at (-2mm,0mm) [vertex] {};}%
					\newcommand{\vb}{\node (vb) at (2mm,0mm) [vertex] {};}%
					\begin{tikzpicture}[node distance=15mm,baseline=-1mm,loop/.style={->,densely dotted,distance=3mm}]
						\begin{scope}[local bounding box=a]
							\graphnode[]{eab}{6mm}{2mm}{}{\va\vb\edge[->]{va}{vb}}
						\end{scope}
					\end{tikzpicture}%
				}
			\)
			
			\item 
			\makebox[3.2cm][l]{the target vertex:}
			\(
				\hat{0}_1 =
				{%
					\newcommand{\va}{\node (va) at (-2mm,0mm) [vertex] {};}%
					\newcommand{\vb}{\node (vb) at (2mm,0mm) [vertex] {};}%
					\begin{tikzpicture}[node distance=15mm,baseline=-1mm,loop/.style={->,densely dotted,distance=3mm}]
						\begin{scope}[local bounding box=a]
							\graphnode[]{eab}{6mm}{2mm}{}{\vb}
						\end{scope}
					\end{tikzpicture}%
				} 
				\subseteq
				{%
					\newcommand{\va}{\node (va) at (-2mm,0mm) [vertex] {};}%
					\newcommand{\vb}{\node (vb) at (2mm,0mm) [vertex] {};}%
					\begin{tikzpicture}[node distance=15mm,baseline=-1mm,loop/.style={->,densely dotted,distance=3mm}]
						\begin{scope}[local bounding box=a]
							\graphnode[]{eab}{6mm}{2mm}{}{\va\vb\edge[->]{va}{vb}}
						\end{scope}
					\end{tikzpicture}%
				}
			\)
		\end{itemize}
		Notice that both faces $\hat{1}$ and $\hat{0}$ are isomorphic to
		\(
			y(0)
			=
			{%
				\newcommand{\vx}{\node (vx) at (0mm,0mm) [vertex] {};}%
				\begin{tikzpicture}[node distance=15mm,baseline=-1mm,loop/.style={->,densely dotted,distance=3mm}]
					\graphnode[]{gx}{6mm}{2mm}{}{\vx}
				\end{tikzpicture}%
			}.
		\)
		
		\item 
		For $n=k+1=2$, recall the Yoneda embedding for $2$-dimensional semi-simplicial sets from \cref{ex:yoneda_building_block_in_graph}.
		The notations correspond as follows: $0=V$, $1=E$, $2=T$, $d^2_2=f_1$, $d^2_1=f_2$, and $d^2_0=f_3$.
		\begin{align*}
			&\left\{
			\begin{aligned}
				y(2)(1) &= \set{d^2_2,d^2_1,d^2_0} \\
				y(2)(2) &= \set{\id_2}
			\end{aligned}
			\right.
			&
			y(2) &=
			{%
			\begin{tikzpicture}[baseline=-1mm, scale=0.9, every node/.style={scale=0.9}]
				\node[graphborder,outer sep=1mm] (gx) {
					\begin{tikzcd}[ampersand replacement=\&, column sep=1.6em, row sep=1em, execute at end picture={
							\scoped[on background layer]
							\fill[rounded corners = 2mm, color=lightgray] (a.center) -- (b.center) -- (c.center) -- cycle;
						}
						]
						\&\& |[alias=b]| \cdot \\
						|[alias=a]|  \cdot \\
						\&\& |[alias=c]| \cdot
						\arrow[""{name=0, anchor=center, inner sep=0}, "d^2_1"', from=2-1, to=3-3]
						\arrow[from=2-1, to=1-3, "d^2_2"]
						\arrow[from=1-3, to=3-3, "d^2_0"]
						\arrow[shorten <=2pt, background color=lightgray, pos=.4, draw=none, from=0, to=1-3] 
					\end{tikzcd} 
				};
			\end{tikzpicture}%
			}
		\end{align*}
		The graph $y(2)$ has three faces: 
		\begin{itemize}[topsep=2pt, noitemsep]
			\item $\hat{2}_2$, which contains $d^2_2$ and its endpoints,
			\item $\hat{1}_2$, which contains $d^2_1$ and its endpoints, and
			\item $\hat{0}_2$, which contains $d^2_0$ and its endpoints.
		\end{itemize}
		Each face is isomorphic to 
		\(
			y(1)
			=
			{%
				\newcommand{\va}{\node (va) at (-2mm,0mm) [vertex] {};}%
				\newcommand{\vb}{\node (vb) at (2mm,0mm) [vertex] {};}%
				\begin{tikzpicture}[node distance=15mm,baseline=-1mm,loop/.style={->,densely dotted,distance=3mm}]
					\begin{scope}[local bounding box=a]
						\graphnode[]{eab}{6mm}{2mm}{}{\va\vb\edge[->]{va}{vb}}
					\end{scope}
				\end{tikzpicture}%
			}.
		\)
	\end{itemize}
\end{example}

In both examples, we notice that each face $\hat{i}_{k+1}$ of $y(k+1)$ is isomorphic to $y(k)$, as objects in $\ndimsemisSet$.
Each face $\hat{i}_{k+1}$ of $y(k+1)$ is a subgraph of $y(k+1)$ and hence belongs to the subpresheaf Heyting algebra $\Omega(k+1)$.
Because of the functor $\Sub \colon \ndimsemisSet \to \cat{HeytAlg}$, 
the isomorphism $y(k) \isom \hat{i}_{k+1}$ in $\ndimsemisSet$ implies a Heyting algebra isomorphism $\Omega(k) \isom \Omega(k+1)_{\leq \hat{i}_{k+1}}$, where $\Omega(k+1)_{\leq \hat{i}_{k+1}} = \setvbar{x \in \Omega(k+1)}{x \leq \hat{i}_{k+1}}$.
We illustrate this for $k+1 \in \set{1,2}$ and $i=0$, where we highlight $\Omega(k+1)_{\leq \hat{0}}$ in red within $\Omega(k+1)$:

\[
\begin{array}{rcccl}
	\Omega(0) 
	&\isom& 
	\Omega(1)_{\leq \hat{0}}
	&\colon& 
	{%
		\newcommand{\va}{\node (va) at (-2mm,0mm) [vertex] {};}%
		\newcommand{\vb}{\node (vb) at (2mm,0mm) [vertex] {};}%
		\begin{tikzpicture}[node distance=15mm,baseline=0mm, scale=0.7, every node/.style={scale=0.7}]%
			\begin{scope}
				\graphnode[]{empty}{6mm}{2mm}{}{}
				\graphnode[right of=empty,yshift=3mm]{ga}{6mm}{2mm}{}{\va}
				\graphnode[right of=empty,yshift=-3mm]{gb}{6mm}{2mm}{}{\vb}
				\graphnode[right of=ga,yshift=-3mm]{gab}{6mm}{2mm}{}{\va\vb}
				\graphnode[right of=gab]{eab}{6mm}{2mm}{}{\va\vb\edge{va}{vb}}
				
				\begin{scope}[metaedge]
					\draw (empty) to (ga);
					\draw (empty) to (gb);
					\draw (ga) to (gab);
					\draw (gb) to (gab);
					\draw (gab) to (eab);
				\end{scope}
			\end{scope}

			\begin{pgfonlayer}{background}
				\draw [fill=red!20,draw=none,rounded corners=2mm] 
				(empty.south west) to
				(empty.south east) to
				(ga.south west) to
				(ga.south east) to
				(ga.north east) to
				(ga.north west) to
				(empty.north east) to
				(empty.north west) to
				cycle;
			\end{pgfonlayer}
		\end{tikzpicture}%
	}
	\\
	\Omega(1) 
	&\isom&
	\Omega(2)_{\leq \hat{0}}
	&\colon&
	{%
		\newcommand{\va}{\node (va) at (-60:2.5mm) [vertex] {};}%
		\newcommand{\vb}{\node (vb) at (60:2.5mm) [vertex] {};}%
		\newcommand{\vc}{\node (vc) at (180:2.5mm) [vertex] {};}%
		\begin{tikzpicture}[node distance=15mm,baseline=0mm, scale=0.6, every node/.style={scale=0.6}]%
			\begin{scope}
				\graphnode[]{empty}{6mm}{6mm}{}{}
				\graphnode[right of=empty,yshift=10mm]{ga}{6mm}{6mm}{xshift=1mm}{\va}
				\graphnode[right of=empty,yshift=0mm]{gb}{6mm}{6mm}{xshift=1mm}{\vb}
				\graphnode[right of=empty,yshift=-10mm]{gc}{6mm}{6mm}{xshift=1mm}{\vc}
				\graphnode[right of=ga]{gab}{6mm}{6mm}{}{\va\vb}
				\graphnode[right of=gb]{gac}{6mm}{6mm}{}{\va\vc}
				\graphnode[right of=gc]{gbc}{6mm}{6mm}{}{\vb\vc}
				
				\begin{scope}[node distance=20mm]    
					\graphnode[right of=gab]{eab}{6mm}{6mm}{}{\va\vb\edge{va}{vb}}
					\graphnode[right of=gac]{eac}{6mm}{6mm}{}{\va\vc\edge{va}{vc}}
					\graphnode[right of=gbc]{ebc}{6mm}{6mm}{}{\vb\vc\edge{vb}{vc}}
					
					\graphnode[right of=eab]{eabx}{6mm}{6mm}{}{\va\vb\vc\edge{va}{vb}}
					\graphnode[right of=eac]{eacx}{6mm}{6mm}{}{\va\vb\vc\edge{va}{vc}}
					\graphnode[right of=ebc]{ebcx}{6mm}{6mm}{}{\va\vb\vc\edge{vb}{vc}}
				\end{scope}
				
				\graphnode[at=(eac),yshift=25mm]{gabc}{6mm}{6mm}{}{\va\vb\vc}
				
				\graphnode[right of=eabx]{hbc}{6mm}{6mm}{}{\va\vb\vc\edge{va}{vc}\edge{va}{vb}}
				\graphnode[right of=eacx]{hac}{6mm}{6mm}{}{\va\vb\vc\edge{va}{vb}\edge{vb}{vc}}
				\graphnode[right of=ebcx]{hab}{6mm}{6mm}{}{\va\vb\vc\edge{vc}{va}\edge{vb}{vc}}
				\graphnode[right of=hac]{triangle}{6mm}{6mm}{}{\va\vb\vc\edge{vc}{va}\edge{va}{vb}\edge{vb}{vc}}
				\graphnode[right of=triangle]{full}{6mm}{6mm}{}{
					\va\vb\vc
					\begin{pgfonlayer}{background}
						\draw [draw=none,fill=black!25,rounded corners=0mm] (va.center) to (vb.center) to (vc.center) to cycle; 
					\end{pgfonlayer}
					\edge{vc}{va}\edge{va}{vb}\edge{vb}{vc}
				}
				
				\begin{pgfonlayer}{background}
					\draw [fill=red!20,draw=none,rounded corners=2mm] 
					(empty.south west) to
					(gb.south east) to
					(gab.south west) to
					(eab.south east) to
					(eab.north east) to
					(ga.north west) to
					(empty.north west) to 
					cycle;
				\end{pgfonlayer}
				
				\begin{scope}[metaedge]
					\draw (empty) to (ga);
					\draw (empty) to (gb);
					\draw (empty) to (gc);
					\draw (ga) to (gab); \draw (ga) to (gac);
					\draw (gb) to (gab); \draw (gb) to (gbc);
					\draw (gc) to (gac); \draw (gc) to (gbc);
					\draw (gab) to (eab);
					\draw (gac) to (eac);
					\draw (gbc) to (ebc);
					\draw (eab) to (eabx);
					\draw (eac) to (eacx);
					\draw (ebc) to (ebcx);
					\draw (eabx) to (hbc); \draw (eabx) to (hac);
					\draw (eacx) to (hab); \draw (eacx) to (hbc); 
					\draw (ebcx) to (hab); \draw (ebcx) to (hac);
					\draw (hab) to (triangle);
					\draw (hac) to (triangle);
					\draw (hbc) to (triangle);
					\draw (triangle) to (full);
					
					\begin{scope}[densely dotted]
						\draw (gab.east) to (gabc.west);
						\draw (gac.east) to (gabc.west);
						\draw (gbc.east) to (gabc.west);
						\draw (gabc.east) to (eabx.west);
						\draw (gabc.east) to (eacx.west);
						\draw (gabc.east) to (ebcx.west);
					\end{scope}
				\end{scope}
			\end{scope}
		\end{tikzpicture}%
	} 
\end{array}
\]

The next lemma formalises this observation.

\begin{replemma}{lem:isom_Omega(n)_Omega(n+1)_leq_delta_i}
    For all $k,i \in \N$ such that $0 \leq k < k+1 \leq n$ and $0 \leq i \leq k+1$, we have 
	\[
		\Omega(k) \isom \Omega(k+1)_{\leq \ithface_{k+1}},
	\]
    where $\Omega(k+1)_{\leq \hat{i}_{k+1}} = \setvbar{x \in \Omega(k+1)}{x \leq \hat{i}_{k+1}}$.
    \customqed%
    \footnote{The symbol $\blacksquare$ denotes that the proof is in the appendix.
    }
\end{replemma}

Each face map $\Omega(d^{k+1}_i) \colon \Omega(k+1) \to \Omega(k)$, when composed with the isomorphism from  \cref{lem:isom_Omega(n)_Omega(n+1)_leq_delta_i}, becomes the intersection with $\ithface_{k+1}$, the $i$\tss{th} face of $y(k+1)$.

\noindent\begin{minipage}{.52\linewidth}
\begin{replemma}{cor:incidence_Omega(n+1)_becomes_intersection}
    For all $k,i \in \N$ such that $0 \leq k < k+1 \leq n$ and $0 \leq i \leq k+1$, the diagram on the right commutes.
    \customqed
\end{replemma}
\end{minipage}
\hfill
\begin{minipage}{.45\linewidth}
    \[
	\begin{tikzcd}[scale cd=.9, row sep=tiny, column sep=tiny]
	& \Omega(k+1) & \\
	\Omega(k)
	& & \Omega(k+1)_{\leq \ithface_{k+1}}
	\ar[from=1-2, to=2-1, "{\Omega(d^{k+1}_i)}"']
	\ar[from=2-1, to=2-3, leftrightarrow, "\isom"' {name=0}]
	\ar[from=1-2, to=2-3, "{-\meet\ithface_{k+1}}"]
	\ar[from=1-2, to=0, phantom, "\circlearrowleft" description]
	\end{tikzcd}
	\]
\end{minipage}

\begin{corollary}
    \label{lem:face_map_di_commutes_with_meet}
    For all $k,i \in \N$ such that $0 \leq k < k+1 \leq n$ and $0 \leq i \leq k+1$,
    the face map $\Omega(d^{k+1}_i) \colon \Omega(k+1) \to \Omega(k)$ commutes with $\land$, i.e., the following diagram commutes:
    \[\begin{tikzcd}
		\Omega(k+1) \times \Omega(k+1)
		\ar[r, "\meet"']
		\ar[d, "{\Omega(d^{k+1}_i) \times \Omega(d^{k+1}_i)}"']
		& \Omega(k+1)
		\ar[d, "{\Omega(d^{k+1}_i)}"]
		\\
		\Omega(k) \times \Omega(k)
		\ar[r, "\meet"]
		& \Omega(k)
	\end{tikzcd}\]
\end{corollary}

\begin{proof}
	Given $x,x' \in \Omega(k+1)$, it suffices to go along the isomorphism of \cref{cor:incidence_Omega(n+1)_becomes_intersection} and to observe that $(x \meet \ithface) \meet (x' \meet \ithface) = (x \meet x') \meet \ithface$.
\end{proof}

\begin{definition}
	In $\ndimsemisSet$, for all $k \in \N$ such that $0 \leq k \leq n$,
	we define $\ynhollow[k]$ to be the subgraph of $y(k)$ which consists of all its faces 
	\(
		\ynhollow[k] \defeq \textstyle\bigjoin_{i=0, \ldots, k} \ithface_k.
	\)
\end{definition}

In short, $\ynhollow[k]$ is the \qmarks{hollow} version or \qmarks{boundary} of $y(k)$, i.e., the one element in the Heyting algebra $\Omega(k)$ just below $y(k)$.

\begin{example}
    Here are visualisations of $\ynhollow$ in lower dimensions.
    \bgroup
    \setlength\tabcolsep{4mm}
    \def\arraystretch{1.5}
    \begin{center}
        \begin{tabular}{c|c|c|c|c}
            $\ynhollow[0]$ & $\ynhollow[1]$ & $\ynhollow[2]$ & $\ynhollow[3]$ & $\ldots$  \\ \hline
            \makecell{empty\\set} $\emptyset$ 
            & \makecell{hollow\\edge} $\cdot \phantom{\to} \cdot$ 
            & \makecell{hollow\\triangle} 
            \(
                \begin{tikzcd}[ampersand replacement=\&, column sep=.8em, row sep=.5em, scale cd=.8]
            	\&\& \cdot \\
            	\cdot \\
            	\&\& \cdot
            	\arrow[""{name=0, anchor=center, inner sep=0}, from=2-1, to=3-3]
            	\arrow[from=2-1, to=1-3]
            	\arrow[from=1-3, to=3-3]
                \end{tikzcd}
            \)
            & \makecell{hollow\\tetrahedron} & $\ldots$
        \end{tabular}
    \end{center}
    \egroup
\end{example}

\begin{replemma}{lem:unique_incidence_except_yn_ynhollow}
    In $\ndimsemisSet$, for all $k \in \N$ such that $0 \leq k < k+1 \leq n$, the incidence tuple function 
    \[
        {(d_{k+1}, \ldots, d_0) \colon \Omega(k+1) \to \Omega(k)^{k+2}}
    \]
    is surjective.
    Additionally, the only two elements that prevent $(d_{k+1}, \ldots, d_0)$ from being injective are $y(k+1)$ and $\ynhollow[k+1]$, which both have $\big(y(k), \ldots, y(k)\big)$ as incidence tuple.
    \customqed
\end{replemma}

We now define our topologies on $\ndimsemisSet$.
We first define partial mappings $j^w$, and then show that they extend uniquely to topologies.

\begin{definition}
    \label{def:topologies_on_ndimsemisSet}
    For all binary words $w \in \set{0,1}^{n+1}$, we define a partial mapping $j^w \colon \Omega \to \Omega$ recursively as follows:
    \begin{itemize}
        \item 
        For $n=0$, $j^0$ and $j^1$ are the topologies on $\ndimsemisSet[0] = \Set$, see \cref{ex:topologies_on_set}. 

        \item 
        Suppose $j^w$ is defined for $w \in \set{0,1}^{n+1}$. 
        Define $j^{w0}$ and $j^{w1}$ by:
    	\begin{itemize}
    		\item 
    		For $k \leq n$: let $j^{w0}_k = j^{w1}_k \defeq j^w_k \colon \Omega(k) \to \Omega(k)$.
    		
    		\item
    		For $n+1$: define
    		\begin{equation}
    			\label{eq:topologies_on_ndimsSet_on_yn}
    			\hspace*{-10mm}
    			\left\{
    			\begin{array}{rcl}
    				\Omega(n+1) 	&\xrightarrow{j^{w0}_{n+1}}& \Omega(n+1) \\
    				y(n+1) 			&\mapsto& y(n+1) \\
    				\ynhollow[n+1] 	&\mapsto& \ynhollow[n+1]
    			\end{array}
    			\right.
    			\hspace*{10mm}
    			\left\{
    			\begin{array}{rcl}
    				\Omega(n+1)		&\xrightarrow{j^{w1}_{n+1}}& \Omega(n+1) \\
    				y(n+1) 			&\mapsto& y(n+1) \\
    				\ynhollow[n+1] 	&\mapsto& y(n+1)
    			\end{array}
    			\right.
    		\end{equation}
    	\end{itemize}
    \end{itemize}
\end{definition}

We have specified $j_{n+1} \in \set{j^{w0}_{n+1}, j^{w1}_{n+1}}$ on only two elements of $\Omega(n+1)$.
It remains to define $j_{n+1}(x)$ for $x < \ynhollow[n+1]$.
We prove that there is a unique way to extend $j$ to all of $\Omega(n+1)$ in \cref{lem:topologies_on_ndimsemisSet_unique_extension} and that it gives a well-defined topology in \cref{thm:topologies_on_ndimsemisSet}.

\begin{replemma}{lem:topologies_on_ndimsemisSet_unique_extension}
	Each $j^w \colon \Omega \to \Omega$ from \cref{def:topologies_on_ndimsemisSet} extends uniquely to a well-defined monotone and idempotent natural transformation as follows.
	The definition is recursive.
	For $n=0$, $j^0,j^1$ are already defined.
	Suppose $j^w$ is defined for $w \in \set{0,1}^{n+1}$.
	In order to define $j \in \set{j^{w0}, j^{w1}}$ on $x \in \Omega(n+1)$ with $x < \ynhollow[n+1]$, we look at the $(n+2)$-tuple $\vec{z} \in \Omega(n)^{n+2}$ defined as follows:
	\begin{equation}
		\replabel{eq:lem_topologies_on_ndimsemisSet_unique_extension_vec_z}
		\vec{z} \defeq (j_n \Omega(d_{n+1})(x), \ldots, j_n \Omega(d_0)(x)).
	\end{equation}
	\[
	\begin{tikzcd}[column sep=huge, ampersand replacement=\&]
		\Omega(n+1) 	\& \Omega(n+1) \\
		\Omega(n)^{n+2} \& \Omega(n)^{n+2}
		\ar[from=1-1, to=1-2, "j_{n+1}", dotted]
		\ar[from=1-1, to=2-1, "{(\Omega(d_{n+1}) , \, \ldots , \, \Omega(d_0))}"']
		\ar[from=1-2, to=2-2, "{(\Omega(d_{n+1}) , \, \ldots , \, \Omega(d_0))}"]
		\ar[from=2-1, to=2-2, "{(j_n)^{n+2}}"']
	\end{tikzcd}
	\]
	\begin{itemize}
    	\item 
    	If $\vec{z} \neq (y(n), \ldots, y(n))$, then by \cref{lem:unique_incidence_except_yn_ynhollow} there exists a unique $x' \in \Omega(n+1)$ with $\vec{z}$ as incidence tuple.
    	We define 
    	\(
    		j_{n+1} (x) \defeq x'.
    	\)
    	
    	\item 
    	If $\vec{z} = (y(n), \ldots, y(n))$, we define 
    	\begin{equation}
    		\replabel{eq:j_n+1_on_x_<_ynhollow}
    		j_{n+1} (x) =
    		\begin{cases}
    			\ynhollow[n+1] &\text{if } j=j^{w0}\\
    			y(n+1) &\text{if } j=j^{w1}
    		\end{cases}
    	\end{equation}
	\end{itemize}
    \customqed
\end{replemma}

\begin{reptheorem}{thm:topologies_on_ndimsemisSet}
    Each $j^w \colon \Omega \to \Omega$ from \cref{def:topologies_on_ndimsemisSet} 
    and extended as in \cref{lem:topologies_on_ndimsemisSet_unique_extension} is a topology on $\ndimsemisSet$.
    Moreover, all topologies on $\ndimsemisSet$ are of this form, implying that there are exactly $2^{n+1}$ of them.
    \customqed
\end{reptheorem}

\subsection{Topologies on \texorpdfstring{$\ndimsSet$}{n-dimensional sSet} }
\label{subsec:leqn}

Fix a natural number $n \in \N$.
In this subsection, we prove that we have an injection from the topologies on $\ndimsSet$ to the topologies $j^w$ on $\ndimsemisSet$, and that only the topologies $j^w$ with $w$ of the form $0^m 1^{n-m+1}$ are in the image of this injection.

From now on, let us specify the category at hand in the super- and subscript of the Yoneda embedding $y$ and of the classifying object $\Omega$, i.e.,
\begin{align*}
    y\leqn\semi\colon \Delta\leqn\semi &\to \ndimsemisSet,
    & \text{and} && \Omega\leqn\semi &= \Sub(y\leqn\semi(n)). \\
    y\leqn \colon \Delta\leqn &\to \ndimsSet,
    & \text{and} && \Omega\leqn &= \Sub(y\leqn(n)).
\end{align*}
To understand topologies on $\ndimsSet$, we first have to better understand what $y\leqn$ and $\Omega\leqn$ are compared to what we already know about $y\leqn\semi$ and $\Omega\leqn\semi$ from \cref{subsec:leqn_semi}.

For instance, $y\leqn(0)$ consists of a vertex $0$, together with the degenerate $k$-simplices $(0,0), (0,0,0), \ldots, (0, \ldots, 0)$ for $k=1, \ldots, n$.
Given $k \leq n$, we observe that in general, $y\leqn(k)$ is obtained from $y\leqn\semi(k)$ by adding all degenerate simplices.

\begin{example}
    \label{ex:Omega_ReflGraph_isom_Omega_Graph}
    For instance, in $\ReflGraph$ given $A' \subseteq A$, then a vertex $v \in A(0)$ is in the subgraph $A'(0)$ if and only if its reflexive loop is in the subgraph too $\refl(v) \in A'(1)$.
	This implies that subgraphs of $y\leqn[1]\semi(1)$ in $\ReflGraph$ are in $1$-$1$ correspondence with subgraphs of $y\leqn[1](1)$ in $\Graph$, and hence there is a Heyting algebra isomorphism between the subobject Heyting algebra in $\ReflGraph$ and the subobject Heyting algebra in $\Graph$:
    \[
	\begin{tikzcd}[
		ampersand replacement=\&,
		/tikz/column 1/.append style={anchor=base east},
		/tikz/column 2/.append style={anchor=base west},
		]
		\Omega\leqn[1]\semi(1)
		\ar[r, draw=none, "=" description]
		\ar[d, draw=none, "\isom" {description, rotate=90}]
		\&
		{%
			\newcommand{\va}{\node (va) at (-2mm,0mm) [vertex] {};}%
			\newcommand{\vb}{\node (vb) at (2mm,0mm) [vertex] {};}%
			\begin{tikzpicture}[node distance=15mm,baseline=-1mm, scale=0.7, every node/.style={scale=0.7}]
				\begin{scope}[local bounding box=a]
					\graphnode[]{empty}{6mm}{2mm}{}{}
					\graphnode[right of=empty,yshift=3mm]{ga}{6mm}{2mm}{}{\va}
					\graphnode[right of=empty,yshift=-3mm]{gb}{6mm}{2mm}{}{\vb}
					\graphnode[right of=ga,yshift=-3mm]{gab}{6mm}{2mm}{}{\va\vb}
					\graphnode[right of=gab]{eab}{6mm}{2mm}{}{\va\vb\edge{va}{vb}}
					\begin{scope}[metaedge]
						\draw (empty) to (ga);
						\draw (empty) to (gb);
						\draw (ga) to (gab);
						\draw (gb) to (gab);
						\draw (gab) to (eab);
					\end{scope}
				\end{scope}
			\end{tikzpicture}%
		}
		\\
		\Omega\leqn[1](1)
		\ar[r, draw=none, "=" description]
		\&
		{%
			\newcommand{\va}{\node (va) at (-2mm,0mm) [vertex] {};}%
			\newcommand{\vb}{\node (vb) at (2mm,0mm) [vertex] {};}%
			\begin{tikzpicture}[node distance=18mm,baseline=-1mm,loop/.style={->,densely dotted,distance=3mm}, scale=0.7, every node/.style={scale=0.7}]%
				\begin{scope}
					\graphnode[]{empty}{6mm}{2mm}{}{}
					\graphnode[right of=empty,xshift=-2mm,yshift=3mm]{ga}{10mm}{2mm}{}{
						\va
						\draw [loop] (va) to[out=180-35,in=180+35] (va);
					}
					\graphnode[right of=empty,xshift=-2mm,yshift=-3mm]{gb}{10mm}{2mm}{}{
						\vb
						\draw [loop] (vb) to[out=-35,in=35] (vb);
					}
					\graphnode[right of=ga,yshift=-3mm]{gab}{10mm}{2mm}{}{
						\va\vb
						\draw [loop] (va) to[out=180-35,in=180+35] (va);
						\draw [loop] (vb) to[out=-35,in=35] (vb);
					}
					\graphnode[right of=gab]{eab}{10mm}{2mm}{}{
						\va\vb\edge{va}{vb}
						\draw [loop] (va) to[out=180-35,in=180+35] (va);
						\draw [loop] (vb) to[out=-35,in=35] (vb);
					}
					
					\begin{scope}[metaedge]
						\draw (empty) to (ga);
						\draw (empty) to (gb);
						\draw (ga) to (gab);
						\draw (gb) to (gab);
						\draw (gab) to (eab);
					\end{scope}
				\end{scope}
			\end{tikzpicture}%
		}
	\end{tikzcd}
	\]
\end{example}

The above example suggests that there is a Heyting algebra isomorphism in general between $\Omega\leqn\semi$ and $\Omega\leqn$.

Because topologies in presheaf categories are bounded meet-semilattice morphisms (cf.~\cref{lem:topology_monotone}), to characterise topologies on $\ndimsSet$ using our characterisation of topologies on $\ndimsemisSet$ (\cref{thm:topologies_on_ndimsemisSet}), it suffices to establish a bounded meet-semilattice isomorphism between $\Omega\leqn\semi$ and $\Omega\leqn$, which we do in \cref{lem:F_k_between_semicase_and_generalcase}.

\begin{remark}
	\label{rem:alternate_proof_Kan_extension}
	We outline here a potential proof strategy for establishing a Heyting algebra isomorphism between $\Omega\leqn\semi$ and $\Omega\leqn$.
	While we did not succeed using this approach due to the complexity of the technical computations, we believe it may still be viable and therefore present it here.
	To proceed, we denote two specific inclusion functors as follows:
	\[
	\begin{array}{rcl}
		\inclusion \colon \Delta\leqn\semi 						&\mono& \Delta\leqn \\
		\opcat{\inclusion} \colon \opcat{(\Delta\leqn\semi)} 	&\mono& \opcat{(\Delta\leqn)}
	\end{array}\]
	Consider the precomposition functor $(\opcat{\inclusion})^* \colon \ndimsSet \to \ndimsemisSet$ induced by $\opcat{\inclusion}$.
	It is a forgetful functor: degenerate simplices are kept, but they are no longer regarded as degenerate,
	see e.g.~\cite[p.~322-323]{Rourke_Sanderson_1971_Delta_sets} or \cite[p.~22]{Fiore_2024_Logical_structure_inverse_functor}.\footnote{In earlier literature, semi-simplicial sets were referred to as \emph{semisimplicial (ss) complexes}  or \emph{$\Delta$-sets}, while simplicial sets were called \emph{complete semisimplicial (css) complexes}.}
	For instance, when $n=1$, the functor has type $\ReflGraph \to \Graph$.
	Given a reflexive graph $A$, it forgets the function $A(\refl)$ from the structure and treats reflexive loops as ordinary loops.
	Consider the left Kan extension $(\opcat{\inclusion})_! \colon \ndimsemisSet \to \ndimsSet$, which exists because $\Set$ is cocomplete.
	It appears to be folklore that $(\opcat{\inclusion})_!$ freely adds degeneracies to a semi-simplicial set to obtain a simplicial set, see e.g.~\cite[Theorem 1.7]{Rourke_Sanderson_1971_Delta_sets} or \cite{Kan_1970_Is_an_ss_complex_a_css_complex}.
	Consider the following diagram:
	\[\begin{tikzcd}[ampersand replacement=\&, row sep=tiny]
		{\Delta\leqn} \& {\ndimsSet} \& {} \\
		\&\& \HeytingAlgebra \\
		{\Delta\leqn\semi} \& {\ndimsemisSet}
		\arrow["{y\leqn}", from=1-1, to=1-2]
		\arrow["{\Sub\leqn}", from=1-2, to=2-3]
		\arrow["\inclusion", tail, from=3-1, to=1-1]
		\arrow["{(*)}"{description}, draw=none, from=3-1, to=1-2]
		\arrow["{y\leqn\semi}"', from=3-1, to=3-2]
		\arrow[""{name=0, anchor=center, inner sep=0}, "{(\opcat{\inclusion})_!}"{description}, from=3-2, to=1-2]
		\arrow["{\Sub\leqn\semi}"', from=3-2, to=2-3]
		\arrow["{(**)}"{description, pos=.7}, draw=none, from=0, to=2-3]
	\end{tikzcd}\]	
	If we assume the folklore result, then $(*)$ commutes up to natural isomorphism, i.e., for all $k \in \Delta\leqn\semi$, then $y\leqn(k) \isom (\opcat{\inclusion})_*(y\leqn\semi(k))$.
	Moreover, we expect $(**)$ to commute, because we observe that a degenerate simplex of the form $s_i(x)$ is in a subpresheaf if and only if $x$ itself is in the subpresheaf.
	We formalise this observation later in \cref{lem:sx_in_subobject_iff_x_in_subobject}.
	The following Heyting algebra isomorphism would then follow, because of preservation of isomorphism by functors:
	\begin{align*}
		\Omega\leqn(k) 
		&=			\Sub\leqn(y\leqn(k)) \\
		&\isom		\Sub\leqn((\opcat{\inclusion})_*(y\leqn\semi(k)) 	\tag*{by $(*)$}\\
		&=		 	\Sub\leqn\semi(y\leqn\semi(k)) 						\tag*{by $(**)$}\\
		&= 			\Omega\leqn\semi(k).
	\end{align*}
	If we assume that the Heyting algebra isomorphism $\Omega\leqn\semi \isom \Omega\leqn$ holds, then \cref{lem:F_k_between_semicase_and_generalcase} becomes redundant.
\end{remark}

We return to our own proof of a bounded meet-semilattice isomorphism between $\Omega\leqn\semi$ and $\Omega\leqn$.

We generalise the observation made in \cref{ex:Omega_ReflGraph_isom_Omega_Graph} that all degenerate simplexes $s^k_i(x)$ of a $k$-simplex $x$ belong to a subpresheaf if and only if $x$ itself belongs to the subpresheaf $A'$.

\begin{replemma}{lem:sx_in_subobject_iff_x_in_subobject}
	For all $0 \leq k < n$, and for each category $\cat{D} \in \set{\ndimsSet, \sSet}$, subpresheaf $A' \subseteq A$, and $k$-simplex $x \in A(k)$:
	\begin{align*}
		x \in A'(k) 
		\iff 
		\forall i \in \set{0, \ldots, k}: A'(s^k_i)(x) \in A'(k+1).
        \tag*{\customqed}
	\end{align*}
\end{replemma}

\begin{definition}
    \label{def:degen_set}
    Given a (semi-)simplicial set $x$ in $\Omega\leqn\semi(k)$ or $\Omega\leqn(k)$, i.e., $x \subseteq y\leqn\semi(k)$ or $x \subseteq y\leqn(k)$, we inductively define for each $l \leq n$ the sets $\degen(x,l)$ of degenerate $l$-simplices:
    \begin{align*}
        \degen(x,0) &= \emptyset, \text{ (there is no degenerate vertices)} \\
        \degen(x,l+1) &= \setvbar{s^l_i(f) \in \opcat{(\Delta\leqn)}(k,l+1)}{f \in x(l) \cup \degen(x,l),\ 0 \leq i \leq l}.
    \end{align*}
\end{definition}

The set $\degen(x,l)$ contains all possible degenerate simplexes of dimension $l$ that the presheaf $x$ can have.
For instance, if $x$ contains only vertices, i.e., $x(0) = \set{v_0, \ldots, v_m}$ and $x(k) = \emptyset$ for $k \geq 1$, then $\degen(x,1) = \set{s^0_0(v_0), \ldots, s^0_0(v_n)}$ contains a reflexive loop for each vertex in $x(0)$.

For each $k \leq n$, we define a function $F_k$ that adds all degeneracies, and a function $F_k\inv$ that remove all degeneracies.

\begin{definition}
	\label{def:F_k}
	For each $k \leq n$, we define two functions 
	\[
	\begin{array}{rcl}
		F_k \colon \Omega\leqn\semi(k) &\rightleftarrows& \Omega\leqn(k) \noloc F_k\inv \\[4mm]
		x\semi \subseteq y\leqn\semi(k) 
		\quad &\rightmapsto&
		x\semi(-) \cup \degen(x_+, -), \\[2mm]
		x \setminus \degen(x,-) &\leftmapsto& x \subseteq y\leqn(k).
	\end{array}
	\]
\end{definition}

\begin{replemma}{lem:F_k_between_semicase_and_generalcase}
    For all $k \leq n$, we have a bounded meet-semilattice isomorphism
	\[
		F_k \colon \Omega\leqn(k) \isom \Omega\leqn\semi(k) \noloc F_k\inv.
	\]
	In other words, the functions $F_k$ and $F_k \inv$ are inverses of each other,
	\begin{itemize}[topsep=4pt, noitemsep]
		\item 
		$F_k$ preserves the top element:
		$F_k(y\leqn\semi(k)) = y\leqn(k)$
		
		\item 
		$F_k$ commutes with $\land$:
		 $F_k(x\semi \land y\semi) = F_k(x\semi) \land F_k(y\semi)$, and
		
		\item 
		$F_k$ commutes with all face maps $d_i$:
		 $\Omega\leqn(d_i) F_{k+1} = F_k \Omega\leqn\semi(d_i)$,
	\end{itemize}
	and the same hold for $F_k \inv$.
    \customqed
\end{replemma}

\cref{lem:F_k_between_semicase_and_generalcase} implies that we have an injection of topologies.

\begin{replemma}{cor:injection_of_topologies}
	The following mapping is injective:
	\[
	\begin{array}{rcl}
		\set{j \mid j \text{ is a topology on }\ndimsSet}
		&\xrightarrow{F\inv(-)F}&
		\set{j^w \mid j^w \text{ is a topology on  }\ndimsemisSet}
		\\
		\big(\Omega\leqn \xrightarrow{j} \Omega\leqn\big)
		&\mapsto&
		\big(\Omega\leqn\semi \xrightarrow{F} \Omega\leqn \xrightarrow{j} \Omega\leqn \xrightarrow{F\inv} \Omega\leqn\semi\big).
	\end{array}
	\]
    \customqed
\end{replemma}

Therefore, characterising the topologies on $\ndimsSet$ reduces to identifying which topologies on $\ndimsemisSet$ are in the image of the injection from \cref{cor:injection_of_topologies}.
The next example provides some intuition.

\begin{example}
    \label{ex:topologies_on_ReflGraph}
    Among the four topologies $j^{00}$, $j^{01}$, $j^{10}$, and $j^{11}$ on $\ndimsemisSet[1]=\Graph$ (see \cref{ex:topologies_on_graphs}), only $j^{00}$, $j^{01}$, and $j^{10}$ commute with the degeneracy map $\Omega\leqn[1](\refl) \colon \Omega\leqn[1](0) \to \Omega\leqn[1](1)$ and are therefore topologies on $\ndimsSet[1] = \ReflGraph$.
    It is easy to check that the desired naturality square commutes for $j^{00}, j^{01}$ and $j^{11}$.
    We show that it does not commute for $j^{10}:$
    \begin{align*}
        \begin{tikzcd}[sep=small, ampersand replacement=\&]
            \Omega(1) \& \Omega(1)  \\
            \Omega(0) \& \Omega(0) 
            \ar[from=1-1, to=1-2, "j_1"]
            \ar[from=2-1, to=1-1, "\refl"]
            \ar[from=2-2, to=1-2, "\refl"']
            \ar[from=2-1, to=2-2, "j_0"']
        \end{tikzcd}
        \quad
        \text{where}
        \quad
        \refl \colon \Omega(0) \to \Omega(1) \colon
        \begin{cases}
            \emptyset 
            &\mapsto 
            \emptyset \\
            \begin{tikzcd}[ampersand replacement=\&] \cdot \arrow[refl] \end{tikzcd} 
            &\mapsto 
            \begin{tikzcd}[sep=small, ampersand replacement=\&]
                s \ar[loop, refl] \ar[r] \& t \ar[loop, refl]
            \end{tikzcd}
        \end{cases}
        \\
        j^{10}_1 \refl (\emptyset)
        = j^{10}_1 (\emptyset)
        = (\strefl)
        \neq \stotrefl
        = \refl(\idrefl)
        = \refl j^{10}_0 (\emptyset).
    \end{align*}
\end{example}

The example suggests that once a topology adds elements of a dimension $m$, i.e., $w_m=1$, then all elements of higher dimension must also added, i.e., $w_l = 1$ for all $l \geq m$. 
This observation leads to the following characterisation of topologies on $\ndimsSet$.
Recall from \cref{def:topologies_on_ndimsemisSet,lem:topologies_on_ndimsemisSet_unique_extension} the mappings of the form $j^w$, which we showed in \cref{thm:topologies_on_ndimsemisSet} to be the topologies on $\ndimsemisSet$.

\begin{reptheorem}{thm:topologies_on_ndimsSet}
    A topology $j^w$ on $\ndimsemisSet$ is in the image of $F\inv(-)F$ from \cref{cor:injection_of_topologies} if and only if $w = 0^m 1^{n+1-m}$ for some $0 \leq m \leq n$.
	In other words, we have a bijection
	\begin{align*}
		\set{j \mid j\text{ is a topology on }\ndimsSet}
		\quad\isom\quad
		\{j^w \mid &~j^w \text{ is a topology on }\ndimsemisSet \\
		&\text{and $w=0^m 1^{n+1-m}$ for}\\
		&\text{some $0 \leq m \leq n+1$}\}.
        \tag*{\customqed}
	\end{align*}
\end{reptheorem}

In the remainder of this paper, we identify each $j$ and $F\inv j F = j^w$.

\subsection{Topologies on \texorpdfstring{$\semisSet$}{semi-sSet} and \texorpdfstring{$\sSet$}{sSet}}
\label{subsec:semi_and_general}

In this section, we consider the categories $\semisSet$ and $\sSet$, meaning that we no longer limit ourselves to only objects of dimension $0$ until $n$, but instead allow elements of \textit{all} dimensions.
In both $\semisSet$ and $\sSet$, we extend what is known in the $n$-dimensional case and infer the general case from it.

\begin{reptheorem}{thm:topologies_on_semisSet_and_sSet}
    Let $\cat{D} \in \set{\semisSet, \sSet}$.
	Take $w \in \set{0,1}^\omega$.
	If $\cat{D} = \sSet$, suppose additionally that $w$ is of the form $0^\omega$ or $0^m 1^\omega$.
	Let $j^w\colon \Omega \to \Omega$ be defined by taking $j^w_n \defeq j^{w_0 \ldots w_n}_n$, where $j^{w_0 \ldots w_n}$ is the topology on $\cat{D}\leqn$ (cf.~\cref{thm:topologies_on_ndimsemisSet,thm:topologies_on_ndimsSet}).  
	Then $j^w$ is a topology on $\cat{D}$.
	Moreover, all topologies on $\cat{D}$ are of this form.
    \customqed
\end{reptheorem}

\subsection{Separated elements and sheaves}
\label{subsec:separated_elements_and_sheaves}

In this subsection, let $\cat{D} \in \set{\ndimsemisSet, \ndimsSet, \semisSet, \sSet}$ denote any of the four categories of simplicial sets that we have considered so far.
We have characterised all topologies on $\cat{D}$ in respectively \cref{thm:topologies_on_ndimsemisSet,thm:topologies_on_ndimsSet,thm:topologies_on_semisSet_and_sSet}.
We now describe the corresponding closure operator, the dense elements, the separated elements, and the sheaves.

The closure associated with a topology $j^w$  (see \cref{lem:equivalence_LT-topologies_and_closure_operators}) adds to a subobject at each dimension $k \geq 1$, recursively, the $k$-simplices for which the faces are in the closure of dimension $k-1$.
This generalises \eqref{eq:closure_double_negation_topology_on_Graph} of the topology $j^{01}$ on $\Graph$ given in \cref{ex:topologies_on_graphs}.

\begin{replemma}{lem:closure_of_topologies_on_sSet}
    For all topologies $j^w$ on $\cat{D} \in \set{\ndimsemisSet, \ndimsSet, \semisSet, \sSet}$ and all subobjects $A' \subseteq A$ in $\cat{D}$, the closure of $A'$ associated with $j^w$, which we denote $\topologybar{A'}$, is given by
	\begin{itemize}[leftmargin=1.5cm]
		\item[$(k=0)$] 
		If $w_0 = 0$, then $\topologybar{A'}(0) = A'(0)$. \\
		If $w_0 = 1$, then $\topologybar{A'}(0) = A(0)$.
		
		\item[$(k > 0)$] 
		If $w_k = 0$, then $\topologybar{A'}(k) = A'(k)$. \\
		If $w_k = 1$, then 
		\begin{equation}
			\replabel{eq:IH_closure_j^w}
			\topologybar{A'}(k) =
			\setvbar{x \in A(k)}{\forall i=0, \ldots, k \colon A(d^k_i)(x) \in \topologybar{A'}(k-1)}.
		\end{equation}
        \customqed
	\end{itemize}
\end{replemma}

With the description in \cref{lem:closure_of_topologies_on_sSet} of the closure operator associated with each topology $j^w$, we can now describe the classes of $j^w$-dense subobjects (\cref{def:dense_subobject}).
For a subobject $A' \subseteq A$ to be $j^w$-dense, $A'$ must already have all elements of each dimensions where nothing is added, i.e., each dimension $k$ with $w_k = 0$.

\begin{replemma}{lem:j^w_dense_subpresheaf}
	Let $j^w$ be a topology on $\cat{D} \in \set{\ndimsemisSet, \ndimsSet, \semisSet, \sSet}$ and $A' \subseteq A$ be a subpresheaf.
	Take an arbitrary $n \in \N$.
	The following are equivalent:
	\begin{enumerate}[noitemsep]
		\item 
		$A'$ is $j^w$-dense.
		In other words, for all $k \leq n$: $\topologybar{A'}(k) = A(k)$.
		
		\item 
		For all $k \leq n$: if $w_k = 0$, then $A'(k) = A(k)$.
        \customqed
	\end{enumerate}
\end{replemma}

\begin{notation}
	\label{not:parallel_k_faces}
    For all simplicial sets $B \in \cat{D}$, and $\vec{x}=(x_k, \ldots, x_0) \in B(k-1)^{k+1}$, we introduce the following notation to denote the set of elements in $B(k)$ that have $\vec{x}$ as incidence tuple:
	\[
	d\inv(\vec{x}) \defeq \setvbar{x \in B(k)}{(B(d_k), \ldots, B(d_0))(x) = \vec{x}}.
	\]
\end{notation}

Thus, two distinct elements that belong to the same set $d \inv (\vec{x})$ generalise the concept of two edges being \emph{parallel} in the category $\Graph$.

In order to describe separated elements and sheaves for every topology on $\cat{D}$, we introduce the next definition, generalising the notion of a graph being \emph{simple}.

\begin{definition}
	\label{def:k-simple_k-exact}
	Take $B \in \cat{D}$.
	We say that $B$ is
	\begin{itemize}[topsep=3pt, noitemsep]
		\item 
		\makebox[2cm][l]{\myemph{$0$-simple}} if $B(0)$ has at most one element,
		
		\item 
		\makebox[2cm][l]{\myemph{$0$-complete}} if $B(0)$ has at least one element,
		
		\item
		\makebox[2cm][l]{\myemph{$0$-exact}} if $B(0) = \set{\cdot}$ is exactly one element. 
	\end{itemize}
	Given $k \geq 1$, we say that $B$ is
	\begin{itemize}[topsep=3pt, noitemsep]
		\item 
		\makebox[2cm][l]{\myemph{$k$-simple}} if for all $\vec{x} \in B(k-1)^{k+1}$, $\Card{d\inv(\vec{x})} \leq 1$,
		
		\item 
		\makebox[2cm][l]{\myemph{$k$-complete}} if for all $\vec{x} \in B(k-1)^{k+1}$, $\Card{d\inv(\vec{x})} \geq 1$,
		
		\item 
		\makebox[2cm][l]{\myemph{$k$-exact}} if for all $\vec{x} \in B(k-1)^{k+1}$, $\Card{d\inv(\vec{x})} = 1$,
	\end{itemize}
\end{definition}

\begin{example}
	In $\Graph$, the notions from \cref{def:k-simple_k-exact} mean the following:
	\begin{itemize}[topsep=3pt, noitemsep]
		\item 
		\makebox[20mm][l]{$0$-simple} $=$ having zero or one vertex.
		
		\item 
		\makebox[20mm][l]{$0$-complete} $=$ having at least one vertex, i.e., being a nonempty graph.
		
		\item 
		\makebox[20mm][l]{$0$-exact} $=$ having one vertex.
		
		\item 
		\makebox[20mm][l]{$1$-simple} $=$ no parallel edges, i.e., simple graph in the usual sense.
		
		\item 
		\makebox[20mm][l]{$1$-complete} $=$ all pairs of vertices $(v_0,v_1)$ admit at least one edge $v_0 \to v_1$.
		
		\item 
		\makebox[20mm][l]{$1$-exact} $=$ all pairs of vertices $(v_0,v_1)$ admit exactly one edge $v_0 \to v_1$, i.e., complete graph in the usual sense.
	\end{itemize} 
\end{example}

We can now describe separated objects, complete objects, and sheaves (\cref{def:separated_elements_complete_elements_sheaves}) for every topology $j^w$ on $\cat{D}$.

\begin{reptheorem}{thm:separated_elements_sheaves_ndimsemisSet}
    Let $j^w$ be a topology on $\cat{D} \in \set{\ndimsemisSet, \ndimsSet, \semisSet, \sSet}$ (cf \cref{thm:topologies_on_ndimsemisSet,thm:topologies_on_ndimsSet,thm:topologies_on_semisSet_and_sSet}) and let $B$ be an object in $\cat{D}$.
	\begin{enumerate}[topsep=3pt, noitemsep]
		\item 
		\makebox[35mm][l]{$B$ is $j^w$-separated}  $\iff$ $B$ is $k$-simple for all $k$ with $w_k = 1$.
		
		\item 
		\makebox[35mm][l]{$B$ is $j^w$-complete} $\iff$ $B$ is $k$-complete for all $k$ with $w_k = 1$.
		
		\item 
		\makebox[35mm][l]{$B$ is a $j^w$-sheaf}  $\iff$ $B$ is $k$-exact for all $k$ with $w_k = 1$.
	\end{enumerate}
	If $D \in \set{\semisSet, \sSet}$, then $k$ ranges over $\N$.
	Otherwise, $k$ ranges over $\set{0, \ldots, n}$.
    \customqed
\end{reptheorem}

For the discrete topology, all three right-hand side conditions in \cref{thm:separated_elements_sheaves_ndimsemisSet} are vacuously true, because $w_k = 0$ for all $k \in \N$.
Hence, the following corollary holds.

\begin{corollary}
	If $j^w$ is the discrete topology on $\cat{D}$, i.e., $w_k = 0$ for all $k$, then all objects in $\cat{D}$ are $j^w$-separated, $j^w$-complete, and $j^w$-sheaves.
\end{corollary}

Thanks to \cref{lem:separated_elements_and_sheaves_form_quasitoposes}, we immediately have the next result.

\begin{corollary}
    For each topology $j^w$ on $\cat{D} \in \set{\ndimsemisSet, \ndimsSet, \semisSet, \sSet}$, the category $\Separated_{j^w}$ of $j^w$-separated elements and the category $\Sheaf_{j^w}$ of $j^w$-sheaves are quasitoposes.
\end{corollary}

If we want to give a name to the $j^w$-separated elements, we might refer to them  as \emph{partially simple simplicial sets}.
This name suggests a form of \qmarks{simplicity} (generalising the concept of simple graphs), i.e., no distinct simplices with same incidence tuples, as defined in Notation~\ref{not:parallel_k_faces}, at each dimension $k$ where $w_k = 1$.

\section{Topologies on bicolored graphs}
\label{sec:bicolour}

We define \emph{bicoloured graphs} as graphs in which the edges are partitioned into two sets.
Morphisms of bicoloured graphs are graph homomorphisms that respect the edge partition.
Bicoloured graphs are equivalently defined as a presheaf category.
We visualise the edge partition by using two colours: blue and red.

\begin{definition}
    \label{def:bicolored_graph}
    \smash{$\BiColGraph \defeq \PresheafDefault$} for
    \(
        \opcat{\indexcategory} \defeq
        \smash{
        \begin{tikzcd}[ampersand replacement=\&]
            {\color{fcolor} E}
                \ar[r, "\color{fcolor}s"{description}, shift left=3pt]
                \ar[r, "\color{fcolor}t"{description}, shift right=3pt]
            \& V
            \& {\color{scolor} E'}
                \ar[l, "{\color{scolor}t\smash{'}}"{description}, shift left=3pt]
                \ar[l, "{\color{scolor}s\smash{'}}"{description}, shift right=3pt]
        \end{tikzcd}
        }
    \).
\end{definition}

The image of the Yoneda embedding $y$ and the Hasse diagrams of the classifying object $\Omega$ of $\BiColGraph$ are similar to the ones in $\Graph$ (\cref{ex:classifying_object_Omega_in_Set_and_Graph,subsec:topologies}):
\[
    \begin{tikzcd}[ampersand replacement=\&, column sep=large, scale cd=.9]
        \ 0 \ 
        \ar[loop, distance=1em, in=90+25, out=90-25, color=fcolor, "0"']
        \ar[loop, distance=1em, in=270+25, out=270-25, color=scolor, "0'"'] 
        \ar[r, "t"{description}, color=fcolor, bend left=15]
        \ar[r, "t'"{description}, color=scolor, bend right=15, shift right=8pt] 
        \& \ 1 \ 
        \ar[loop, distance=1em, in=90+20, out=90-20, color=fcolor, "{(s,t)}"']
        \ar[loop, distance=4em, in=90+40, out=90-40, color=fcolor, "s \to t"']
        \ar[loop, distance=1em, in=270+20, out=270-20, color=scolor, "{(s',t')}"'] 
        \ar[loop, distance=4em, in=270+40, out=270-40, color=scolor, "s' \to t'"'] 
        \ar[l, "s"{description}, color=fcolor, bend right=15, shift right=8pt]
        \ar[l, "s'"{description}, color=scolor, bend left=15] 
    \end{tikzcd}
    \begin{aligned}
        y(V) &= \cdot 
        & \Omega(V) &= 
        \begin{tikzpicture}[baseline={(0v.base)}]
            \node (0v) at (-5,0) {$\emptyset$};
            \node (1v) at (-4,0) {$\cdot$};
            \draw[-, amber] (0v) to (1v);
        \end{tikzpicture}%
        \\
        y(E) &= \cdot  {\color{fcolor} \to} \cdot
        & \Omega(E) &= 
            \begin{tikzpicture}[baseline={(0.base)}]
                \node (0) at (0,0) {$\emptyset$};
                \node (s) at (1,.2) {$\cdot \phantom{\to \cdot}$};
                \node (t) at (1,-.2) {$\phantom{\cdot \to} \cdot$};
                \node (st) at (2.5,0) {$\cdot \phantom{\to} \cdot$};
                \node (stot) at (4,0) {$\cdot {\color{fcolor} \to} \cdot$};
                \begin{scope}[every path/.append style={amber}]
                \draw[-] (0) to (s);
                \draw[-] (0) to (t);
                \draw[-] (s) to (st);
                \draw[-] (t) to (st);
                \draw[-] (st) to (stot);
                \end{scope}
            \end{tikzpicture}%
        \\
        y(E') &= \cdot {\color{scolor} \to} \cdot
        & \Omega(E') &=
            \begin{tikzpicture}[baseline={(0.base)}]
                \node (0) at (0,0) {$\emptyset$};
                \node (s) at (1,.2) {$\cdot \phantom{\to \cdot}$};
                \node (t) at (1,-.2) {$\phantom{\cdot \to} \cdot$};
                \node (st) at (2.5,0) {$\cdot \phantom{\to} \cdot$};
                \node (stot) at (4,0) {$\cdot {\color{scolor} \to} \cdot$};
                \begin{scope}[every path/.append style={amber}]
                \draw[-] (0) to (s);
                \draw[-] (0) to (t);
                \draw[-] (s) to (st);
                \draw[-] (t) to (st);
                \draw[-] (st) to (stot);
                \end{scope}
            \end{tikzpicture}%
    \end{aligned}
\]

While $\Graph$ has four topologies, it turns out that there are eight in $\BiColGraph$.
We use an extended version of the superscript notation from \cref{def:topologies_on_ndimsemisSet,thm:topologies_on_ndimsemisSet} to exhibit what the closure of each topology adds.

\begin{replemma}{lem:8_topologies_in_BiColGraph}
    There are eight topologies in $\BiColGraph$, whose closure are:
	\begin{multicols}{2}
		\begin{enumerate}[noitemsep]
			\item 
			$j^{00}$ adds nothing (discrete),
			
			\item
			$j^{01}$ adds ${\color{fcolor}E}$,
			
			\item
			$j^{02}$ adds ${\color{scolor}E'}$,
			
			\item
			$j^{03}$ adds ${\color{fcolor}E},{\color{scolor}E'}$,
			
			\item
			$j^{10}$ adds $V$,
			
			\item
			$j^{11}$ adds $V,{\color{fcolor}E}$,
			
			\item
			$j^{12}$ adds $V,{\color{scolor}E'}$,
			
			\item
			$j^{13}$ adds everything (trivial).
            \customqed
		\end{enumerate}
	\end{multicols}
\end{replemma}

With the same reasoning as in $\Graph$ (see, e.g.,~\cite[Theorems~3 \& 4]{Vigna_2003_Simplegraph_quasitopos}),
the $j^{01}$-separated elements are graphs without parallel red edges, and the $j^{02}$-separated elements are graphs without parallel blue edges. 
Thanks to \cref{lem:separated_elements_and_sheaves_form_quasitoposes}, we have the next result.

\begin{corollary}
    Partially simple bicolored graphs, i.e., bicolored graphs where no parallel edges are allowed for only one edge colour, form a quasitopos.
\end{corollary}

For future work, one could imagine defining \emph{bicolored simplicial sets} by considering the following index category:
\[
    \opcat{(\mathsf{BiCol}\Delta)}
    \defeq
    \begin{tikzcd}
	{\color{fcolor} \cdots} & {\color{fcolor} 2} & {\color{fcolor} 1} & 0 & {\color{scolor} {1'}} & {\color{scolor} {2'}} & {\color{scolor} \cdots}
	\arrow["\cdots"{description}, draw=none, from=1-1, to=1-2]
	\arrow["\cdots"{description}, shift right=2, draw=none, from=1-1, to=1-2]
	\arrow["\cdots"{description}, shift left=2, draw=none, from=1-1, to=1-2]
	\arrow[shift right=4, from=1-2, to=1-3]
	\arrow[shift left=4, from=1-2, to=1-3]
	\arrow[from=1-2, to=1-3]
	\arrow[shift right=2, from=1-3, to=1-2]
	\arrow[shift left=2, from=1-3, to=1-2]
	\arrow[shift right=2, from=1-3, to=1-4, "\color{fcolor}t"']
	\arrow[shift left=2, from=1-3, to=1-4, "\color{fcolor}s"]
	\arrow[from=1-4, to=1-3]
	\arrow[from=1-4, to=1-5]
	\arrow[shift right=2, from=1-5, to=1-4, "\color{scolor}s'"']
	\arrow[shift left=2, from=1-5, to=1-4, "\color{scolor}t'"]
	\arrow[shift right=2, from=1-5, to=1-6]
	\arrow[shift left=2, from=1-5, to=1-6]
	\arrow[shift right=4, from=1-6, to=1-5]
	\arrow[shift left=4, from=1-6, to=1-5]
	\arrow[from=1-6, to=1-5]
	\arrow["\cdots"{description}, draw=none, from=1-7, to=1-6]
	\arrow["\cdots"{description}, shift right=2, draw=none, from=1-7, to=1-6]
	\arrow["\cdots"{description}, shift left=2, draw=none, from=1-7, to=1-6]
\end{tikzcd}\]

With this definition, we obtain uniformly colored triangles, tetrahedra, etc., i.e., with all faces having the same colour.
To allow for a triangle to have edges of different colours, we could consider one set of triangle for each possible combination.
In other words, have an object $2_u$ for each word $u=u_0 u_1 u_2 \in \set{0,1}^3$ where $u_i=0$, respectively $u_i=1$, means the $i$\tss{th} edge of the triangle is blue, respectively red.
The object $2_{001}$, for instance, would therefore have face maps $d_2,d_1 \colon 2_{001} \to \color{fcolor}1$ and $d_0 \colon 2_{001} \to \color{scolor}1'$.

\section{Topologies on Fuzzy sets}
\label{sec:fuzzy_case}
\label{sec:topologies_on_FuzzySet}

In the previous sections, we studied specific categories of graphs that are presheaf toposes and identified certain subcategories as quasitoposes through Lawvere-Tierney topologies.
In this section we extend these results to categories of sets endowed with a fuzzy structure.
Assigning a membership value from a poset $(\labels, \leq)$ to each element of a set defines a \emph{fuzzy set} (\cref{def:fuzzy_set}).
Similarly, assigning membership values to the vertices and edges of a graph defines a \emph{fuzzy graph}.
In \cite{Rosset_Overbeek_Endrullis_2023_Fuzzy_presheaves_are_quasitoposes}, we introduced the concept of \emph{fuzzy presheaves}, generalising both fuzzy sets and fuzzy graphs.
Our interest in fuzzy graphs stems from the fact that graph rewriting with \emph{relabelling} can be easily achieved when graphs are labelled from a Heyting algebra—making labelled graphs synonymous with fuzzy graphs—using the \pbpostrong{} graph rewriting formalism.
Moreover, we showed that fuzzy sets, fuzzy graphs, and, more generally fuzzy presheaves form quasitoposes when the membership values are drawn from Heyting algebras~\cite{Rosset_Overbeek_Endrullis_2023_Fuzzy_presheaves_are_quasitoposes}.

The notion of LT-topology, as defined in \cref{def:topology_on_topos} applies only when the underlying category is a topos.
In a topos, LT-topologies correspond to closure operators (\cref{def:closure_operator,lem:equivalence_LT-topologies_and_closure_operators}).
In a quasitopos, the notion of closure operator still applies.

In this section, we characterise all topologies on the category of fuzzy sets, as well as their separated elements and sheaves.

We first recall the definition of fuzzy sets.

\medskip
\noindent\begin{minipage}{.8\linewidth}
\begin{definition}[{\cite{Zadeh_1965_Fuzzy_sets}}]
    \label{def:fuzzy_set}
    Given a poset $(\labels,\leq)$, an \myemph{$\labels$-fuzzy set} is a pair $(A,\alpha)$ consisting of a set $A$ and a \myemph{membership function} $\alpha\colon A \to \labels$.
    A morphism $f\colon(A,\alpha) \to (B,\beta)$ of fuzzy sets is a function $f\colon A \to B$ such that $\alpha \leq \beta f$, i.e., $\alpha(a) \leq \beta f (a)$ for all $a \in A$.
    They form the category $\FuzzySetDefault$.
\end{definition}
\end{minipage}
\hfill
\begin{minipage}{.2\linewidth}
\begin{center}
\begin{tikzcd}[column sep=2mm, row sep=5mm,ampersand replacement=\&]
    A \& {} \& B \\
    \& {\labels} 
    \arrow["\alpha"', from=1-1, to=2-2] 
    \arrow["\beta", from=1-3, to=2-2] 
    \arrow["f", from=1-1, to=1-3]
    \arrow["\leq"{description}, color=gray, draw=none, from=1-2, to=2-2,pos=0.4]
\end{tikzcd}
\end{center}
\end{minipage}
\medskip

\begin{definition}[{\cite[41.1]{Wyler_1991}}]
    \label{def:topology_on_quasitopos}
    A \myemph{topology on a quasitopos} $\cat{E}$ is a closure operator as defined in \cref{def:closure_operator}.
\end{definition}

With the definition of a topology $\topology$ as a closure operator $\topology$ on a quasitopos (\cref{def:topology_on_quasitopos}), the following notions also apply: $\tau$-dense subobjects (\cref{def:dense_subobject}), $\topology$-separated elements, $\topology$-complete elements, and $\topology$-sheaves (\cref{def:separated_elements_complete_elements_sheaves}), see \cite[Definition 41.4 and 42.1]{Wyler_1991}.
Moreover, the analogue of \cref{lem:separated_elements_and_sheaves_form_quasitoposes} holds in this setting:

\begin{lemma}[{\cite[Theorem 43.6]{Wyler_1991}}]
	\label{lem:separated_elements_and_sheaves_in_quasitopos_form_quasitoposes}
	Let $\cat{E}$ be a quasitopos and $\topology$ be a topology on $\cat{E}$.
	Then $\Separated_\topology$ and $\Sheaf_\topology$ are quasitoposes.
\end{lemma}

We fix a Heyting algebra $(\labels, \leq)$ for the remainder of this section.
Therefore, $\FuzzySetDefault$ is a quasitopos~\cite{Stout_1993}, which means that we can consider topologies.

\begin{remark}
    Condition \ref{def:topology_on_quasitopos_5_strong} in the definition of closure operator (\cref{def:closure_operator}) concerns \emph{strong} subobjects. 
    In $\FuzzySetDefault$, a subset $(A_0,\alpha_0) \subseteq (A,\alpha)$ is \emph{strong} if $\alpha_0 = \alpha\restrict{A_0}$, see e.g.~\cite[Ex.~19]{Rosset_Overbeek_Endrullis_2023_Fuzzy_presheaves_are_quasitoposes}.
\end{remark}

We will see in \cref{lem:topologies_on_FuzzySet} that there is a bijection between the topologies on $\FuzzySetDefault$ that do not add elements and the set of \emph{nuclei} on $\labels$.
Thus, we begin by recalling the definition of a nucleus.

\begin{definition}[{\cite[Section 3]{Bezhanishvili_Ghilardi_2007_Nuclei_Heyting_algebra}}]
	\label{def:nucleus}
	In a Heyting algebra $(\labels, \leq)$, a \myemph{nucleus} (plural: \emph{nuclei}) on $\labels$ is a function $\phi \colon \labels \to \labels$ that satisfies for all $a,b \in \labels$:
		\begin{enumerate}[topsep=2pt, noitemsep, series=def_nucleus, label={{\normalfont(\Alph*)}}]
			\item 
			\replabel{it:nucleus_1_meet_preserving}
			\makebox[5cm][l]{$\phi(a \meet b) = \phi(a) \meet \phi(b)$,}  ($\meet$-preservation)
			
			\item 
			\replabel{it:nucleus_2_increasing}
			\makebox[5cm][l]{$a \leq \phi(a)$, and}  (increasing)
			
			\item 
			\replabel{it:nucleus_3_idempotent_one_direction}
			\makebox[5cm][l]{$\phi(\phi(a)) \leq \phi(a)$.} 
		\end{enumerate}
\end{definition}

\begin{replemma}{lem:alternate_def_nucleus}
	A nucleus $\phi \colon \labels \to \labels$ satisfies for all $a,b \in \labels$:
	\begin{enumerate}[topsep=2pt, noitemsep]
		\item[\normalfont(D)]
		\makebox[5cm][l]{$\phi(\top) = \top$,} {\normalfont($\top$-preservation)} 
		
		\item[\normalfont(E)]
		\makebox[5cm][l]{$\phi(a) \leq \phi(b)$ if $a \leq b$,} {\normalfont(monotonicity)} 
		
		\item[\normalfont(F)]
		\makebox[5cm][l]{$\phi(\phi(a)) = \phi(a)$, and} {\normalfont(idempotence)} 
		
		\item[\normalfont(G)]
		\makebox[5cm][l]{$\phi(a) \meet b \leq \phi(a \meet b)$.} 
	\end{enumerate}
	Moreover, a function $\phi \colon \labels \to \labels$ satisfying {\normalfont(B)}, {\normalfont(C)}, {\normalfont(E)}, and {\normalfont(G)} is a nucleus, meaning it satisfies {\normalfont(A)}.
    \customqed
\end{replemma}

\begin{remark}
	Nuclei are usually defined on locale, see e.g.,~\cite[IX.4, p.483]{MacLane_Moerdijk_1994} or \cite[Section 1.5]{Borceux_1994_vol3}.
	A \emph{locale} is the dual structure (i.e.~opposite category) of a \emph{frame}, which is the same structure as a Heyting algebra.
	The key difference lies in the morphisms: frame morphisms and Heyting algebra morphisms are not required to preserve the same operations.
	When viewing the Heyting algebra $\labels$ as a category, a nucleus on $\labels$ is a meet-preserving monad, or in other words a \emph{(commutative) strong} monad (see, e.g.,~\cite[Definition 5.2.9]{Jacobs_2017_book} for the definition).
\end{remark}

Our next result supports MacLane's observation of a \qmarks{formal similarity between the definition of a nucleus on a locale and that of a Lawvere-Tierney topology on a topos}~\cite[p.~484]{MacLane_Moerdijk_1994} and Johnstone's observation that nuclei on the complete Heyting algebra $\Omega$ correspond to Lawvere-Tierney topologies on a topos~\cite[Section C1.1 p.481]{Johnstone_elephant_2002}.

\begin{replemma}{lem:topologies_on_FuzzySet}
    Here is an enumeration of all topologies on $\FuzzySetDefault$:
	\begin{enumerate}[topsep=1pt, itemsep=0pt]
		\item 
		The trivial topology adds all elements and replaces membership values by the ones of the main fuzzy set:
		\[
			\topologybar{(A',\alpha') \subseteq (A,\alpha)}
			\defeq
			(A, \alpha).
		\]
		
		\item 
		The other topologies do not add elements and are in bijection with nuclei $\phi \colon \labels \to \labels$, with the following inverse constructions.
		Given a topology, we define $\phi(x \in \labels)$ as the increase of membership in 
		\begin{align*}
				\topologybar{\set{ \, \cdot^x } \subseteq \set{\cdot^\top}} 
				&\defeq
				\big( \set{ \, \cdot^{\phi(x)}} \subseteq \set{\cdot^\top} \big).
				\\
		\intertext{Given a nucleus $\phi$, define a topology by}
				\topologybar{(A',\alpha') \subseteq (A,\alpha)}
				&\defeq 
				(A', \phi \alpha' \meet \alpha).
                \tag*{\customqed}
		\end{align*}
	\end{enumerate}
\end{replemma}

\begin{example}
	We consider different functions $\phi : \labels \to \labels$.
	\begin{itemize}[topsep=0pt, itemsep=1pt]
		\item 
		For $\phi = \id_{\labels}$, the labels are unchanged, hence giving the \emph{discrete topology}.
		
		\item
		For $\phi = \top$, the constant top function, the corresponding topology discards the membership of the subset and replaces it with membership of the main set, i.e., $\topologybar{(A', \alpha') \subseteq (A,\alpha)} = \big( (A',\alpha\restrict{A'}) \subseteq (A,\alpha) \big)$.
		It is the \emph{double negation topology} as seen in \cite[Lemma 51]{Rosset_Overbeek_Endrullis_2023_Fuzzy_presheaves_are_quasitoposes}. 
		
		\item 
		In a De Morgan Heyting Algebra \textemdash i.e., a Heyting algebra where \textit{both} De Morgan Laws hold \textemdash there exists a topology corresponding to the double negation function $\phi_{\lnot \lnot} \colon \labels \to \labels$.
		We denote it as $\phi_{\lnot \lnot}$ to avoid ambiguity with the double negation topology.  
		We prove in \cref{lem:lnotlnot_monad_plus_inequality_in_DMHeytAlg} that this function is indeed a nucleus on $\labels$.
		Note that the corresponding topology is not the double negation topology. 
	\end{itemize}  
\end{example}

\begin{replemma}{lem:lnotlnot_monad_plus_inequality_in_DMHeytAlg}
    In any Heyting Algebra $(\labels,\leq)$, $\phi_{\lnot \lnot} \colon \labels \to \labels$, defined as $\phi_{\lnot \lnot}(x) = \lnot \lnot x$, is a monad.
	Moreover, if $(\labels,\leq)$ is a De Morgan Heyting algebra, then 
	$\phi_{\lnot \lnot}$ is a nucleus.
	Hence, $\phi_{\lnot \lnot}$ corresponds to a topology on $\FuzzySetDefault$ according to \cref{lem:topologies_on_FuzzySet}.
    \customqed
\end{replemma}

The next lemma gives a necessary and sufficient condition for a fuzzy set to be a sheaf for the topologies in \cref{lem:topologies_on_FuzzySet}.

\begin{replemma}{lem:sheaves_in_FuzzySet}
    Consider a topology on $\FuzzySetDefault$ that corresponds to a nucleus ${\phi \colon \labels \to \labels}$ (cf \cref{lem:topologies_on_FuzzySet}). 
	Every fuzzy set $(B,\beta)$ is separated and $(B,\beta)$ is a sheaf if and only if $\ima(\beta) \subseteq \ima(\phi)$.
    \customqed
\end{replemma}

\noindent \begin{minipage}{.85\linewidth}
\begin{example}
    \label{ex:sheaves_in_fuzzyset_with_unit_interval}
    Take the unit interval $\labels = [0,1]$ and the topology corresponding to $\phi(x) = \max \set{x,0.5} = x \join 0.5$, which is indeed a nucleus as it is increasing, monotone, idempotent, and satisfies $\phi(x) \meet y \leq \phi(x \meet y)$.
	Then $(B,\beta)$ is a sheaf if and only if it has all memberships within $[0.5,1]$.
\end{example}
\end{minipage}
\hfill
\begin{minipage}{.1\linewidth}
\begin{center}
\begin{tikzpicture}[scale=.4]
    \draw[pattern={south west lines}]
        (0,2) rectangle +(1,2);
    \draw (0,0) rectangle + (1,2);
    \node at (1,0) [anchor=west] {$0$};
    \node at (1,2) [anchor=west] {$0.5$};
    \node at (1,4) [anchor=west] {$1$};
    \draw [line width=1pt, -latex] (.5,0) -- (.5,2);
\end{tikzpicture}
\end{center}
\end{minipage}
\medskip

\section{Conclusion}

In this paper, we characterised all Lawvere-Tierney topologies on the categories of simplicial sets, bicoloured graphs, and fuzzy sets.
We also detailed what separated elements and sheaves are for each of those topologies, resulting in the identification of new quasitoposes: partially simple simplicial sets and partially simple graphs.

\subsection{Related Work}
\label{subsec:simplicial_related_work}

Lawvere-Tierney topologies topologies generalise \emph{Grothendieck topologies}, and characterisations of Grothendieck topologies have been made in the literature for several toposes.
For instance, Lindenhovius~\cite{Lindenhovius_2014_Grothendieck_topologies_on_Artinian_poset,Ochs_2021_Missing_formula_Grothendieck_topologies_on_Artinian_poset} showed that for an \emph{Artinian poset}\footnote{A poset is called \emph{Artinian} if every nonempty subset contains a least element} $P$, there exist bijections between the set of subsets of $P$, the set of Grothendieck topologies on the presheaf topos $\cat{E} \defeq \Presheaf{P}$, and the set of nuclei (\cref{def:nucleus}) on the Heyting algebra $\Sub(1_E)$.
There is some similarity with the correspondence we established between LT-topologies on the category of $\labels$-fuzzy set (for a complete Heyting algebra $\labels$) and nuclei on $\labels$ (\cref{lem:topologies_on_FuzzySet}).
Note that Lindenhovius' result does not apply to simplicial sets, as the simplex category $\Delta$ is not a poset due to the presence of parallel morphisms.

Classifications of subtoposes also exist for $\sSet$, the category of simplicial sets.
Kennett et al.~\cite{Kennett_Riehl_Roy_Zaks_2011_Levels_in_topos_simplicial_sets} observed that essential subtoposes of $\sSet$ correspond to the full subcategories of $\Delta$ that are stable under idempotent splitting.
This can be combined with the observation by Kelly and Lawvere~\cite[2, Remark 4.9]{Kelly_Lavwere_1989_Complete_lattice_essential_localizations} that all subtoposes of $\sSet$ are essential.
Note that the subquasitoposes we identified in this paper are not always toposes.
For instance, simple graphs, which arise through separation via the double negation topology on $\Graph = \sSet\leqn[1]\semi$, form a quasitopos but not a topos.

There has also been work on classifying subtoposes of the toposes of sheaves on finite categories, or, more generally, on categories whose slices are essentially finite, such as the category of semi-simplicial sets, see e.g.~\cite[Proposition 4.10]{Kelly_Lavwere_1989_Complete_lattice_essential_localizations} or \cite[Lemma C2.2.21]{Johnstone_elephant_2002}.

\subsection{Future Work}

There are several directions for future work.
Extending our results on simplicial sets to include \emph{symmetric} simplicial sets, which generalise undirected graphs \cite{Grandis_2001_Symmetric_simplicial_sets}, would be a natural next step.
Another direction would be to generalise our results to general presheaf categories.
It would be interesting to see if one can define a notion of \emph{partially simple presheaves}, generalising the notion of \emph{partially simple graphs} and \emph{partially simple simplicial sets} that we have introduced, and to prove that those are the separated elements of LT-topologies on presheaf categories.
An extension to non-presheaf toposes might also be interesting.

Additionally, we could characterise topologies on bicoloured simplicial sets, as mentioned at the end of \cref{sec:bicolour}, or on fuzzy presheaves~\cite{Rosset_Overbeek_Endrullis_2023_Fuzzy_presheaves_are_quasitoposes}, thereby generalising our results on fuzzy sets (\cref{sec:topologies_on_FuzzySet}).

In \cref{sec:top-simplicial,sec:bicolour}, we employed a separation process based on a topology to obtain novel quasitoposes from toposes. 
In \cref{sec:topologies_on_FuzzySet}, we employed the same process, but starting from quasitoposes.
In this respect, an interesting question is whether iterating this processnmight yield interesting new structures.

A complete characterisation of all subquasitoposes for the (quasi)toposes studied in this paper could be another direction for future work.
This would extend existing classifications of subtoposes of $\sSet$ and other toposes, and may reveal connections between the subtopos classifications by Kennett et al.~\cite{Kennett_Riehl_Roy_Zaks_2011_Levels_in_topos_simplicial_sets} and our characterisation of subquasitoposes (see \cref{subsec:simplicial_related_work}).
However, identifying the precise nature of the full subcategories of $\Delta$ that are stable under idempotent splitting remains unclear to us and their identification is a direction for future work.

Finally, one could investigate whether the internal logic of the quasitopos obtained via separation can be deduced from the internal logic of the initial (quasi)topos.

\subsubsection*{Acknowledgments}
Alo\"is Rosset and J\"org Endrullis received funding from the Netherlands Organization for Scientific Research (NWO) under the Innovational Research Incentives Scheme Vidi (project.\ No.\ VI.Vidi.192.004).

{
\bibliographystyle{meta/splncs04}
\sloppy
\bibliography{main}
}

\newpage
\section{Appendix}

\subsection{Proofs of \texorpdfstring{\cref{subsec:leqn_semi}}{Section 3.2}}

The next lemma is a well-known fact that we will need in the proof of \cref{lem:isom_Omega(n)_Omega(n+1)_leq_delta_i}.

\begin{lemma}
    \label{lem:postcomposition_by_monic_is_monic}
    Given a category $\cat{D}$, and a monomorphism $f\colon x \to y$, then postcomposition with $f$ is a monomorphism in $\PresheafDefault$, i.e., for all $c \in \cat{C}$, $f \circ - \colon \cat{C}(c,x) \to \cat{C}(c,y)$ is injective.
\end{lemma}

\repeatlemma{lem:isom_Omega(n)_Omega(n+1)_leq_delta_i}

\begin{proof}[Proof of \cref{lem:isom_Omega(n)_Omega(n+1)_leq_delta_i}]
	Consider the following simplicial set morphism, which postcomposes with $\opcat{(d^{k+1}_i)}$ in $\Delta\leqn\semi$ (see \cref{def:yoneda_embedding}), or equivalently that precomposes with $d^{k+1}_i$ in $\opcat{(\Delta\leqn\semi)}$.
	\[
		y(d^{k+1}_i) = \big(\opcat{(d^{k+1}_i)} \cdot -\big) \colon y(k) \to y(k+1).
	\]
	By definition of $\Delta\leqn\semi$, the morphism $\opcat{(d^{k+1}_i)}$ is monic (see e.g.~\cite{Friedman_2012_simplicial_sets,Riehl_2011_leisurely_intro_simplicial_sets}).
	It is well-known that postcomposition by a monomorphism is itself a monic operation (see \cref{lem:postcomposition_by_monic_is_monic} for a proof).
	Hence, $y(d^{k+1}_i)$ is monic.
	There is therefore a presheaf isomorphism between $y(k)$ and its image via $y(d^{k+1}_i)$.
	Call this image $x \subseteq y(k+1)$.
	Notice that $x$ is equal to $\hat{i}_{k+1}$ as detailed in \eqref{eq:i_hat}:
	\[
	\left\{
	\begin{array}{lcl}
		y(k)(k+1)		&=& \emptyset  \\
		y(k)(k)			&=& \set{\id_k} \\
		y(k)(k-1)		&=& \setvbar{ d^{k}_{i_{k}} }{i_{k}} \\
		\ldots \\
		y(k)(0)			&=& \setvbar{ d^{1}_{i_{1}} \ldots d^{k}_{i_{k}} }{\ldots}
	\end{array}
	\right.
	\mapsto
	\left\{
	\begin{array}{lcl}
		x(k+1) 	&=& \emptyset \\
		x(k)	&=& \set{d^{k+1}_i} \\
		x(k-1)	&=& \setvbar{ d^{k}_{i_{k}} d^{k+1}_i }{i_{k}} \\
		\ldots \\
		x(0)	&=& \setvbar{ d^{1}_{i_{1}} \ldots d^{k}_{i_{k}} d^{k+1}_i }{\ldots}
	\end{array}
	\right.
	\]
	Hence, there is an isomorphism $y(k) \isom \hat{i}_{k+1}$ as objects of $\ndimsemisSet$.
	Isomorphisms are preserved by functors, hence, the functor $\Sub \colon \ndimsemisSet \to \cat{HeytAlg}$ 
    gives an isomorphism in $\HeytingAlgebra$:
	\[
		\Omega(k) = \Sub(y(k)) \isom \Sub(\ithface_{k+1} \subseteq y(k+1)) = \Omega(k+1)_{\leq \ithface_{k+1}}.
		\qedhere
	\]
\end{proof}

\noindent\begin{minipage}{.52\linewidth}
    \repeatlemma{cor:incidence_Omega(n+1)_becomes_intersection}
\end{minipage}
\hfill
\begin{minipage}{.45\linewidth}
    \[
	\begin{tikzcd}[scale cd=.9, row sep=tiny, column sep=tiny]
	& \Omega(k+1) & \\
	\Omega(k)
	& & \Omega(k+1)_{\leq \ithface_{k+1}}
	\ar[from=1-2, to=2-1, "{\Omega(d^{k+1}_i)}"']
	\ar[from=2-1, to=2-3, leftrightarrow, "\isom"' {name=0}]
	\ar[from=1-2, to=2-3, "{-\meet\ithface_{k+1}}"]
	\ar[from=1-2, to=0, phantom, "\circlearrowleft" description]
	\end{tikzcd}
	\]
\end{minipage}

\begin{proof}[Proof of \cref{cor:incidence_Omega(n+1)_becomes_intersection}]
	Take $x \subseteq y(k+1)$.
	Recall that $\Omega(d^{k+1}_i)$ is defined as follows:
	\[
		\Omega(d^{k+1}_i) =
		\quad
		\begin{tikzcd}[ampersand replacement=\&]
			{\text{output}} \& {\text{input}} \\
			y(k)			\& y(k+1)
			\arrow[dotted, tail, from=1-1, to=1-2]
			\arrow[dotted, tail, from=1-1, to=2-1]
			\arrow["\lrcorner"{anchor=center, pos=0.125}, draw=none, from=1-1, to=2-2]
			\arrow[tail, from=1-2, to=2-2]
			\arrow[from=2-1, to=2-2, tail, "y(d^{k+1}_i)"']
		\end{tikzcd}
	\]
	As noted in the proof of \cref{lem:isom_Omega(n)_Omega(n+1)_leq_delta_i}, the bottom morphism $y(d^{k+1}_i)$ is monic; this allows us to replace $y(k)$ by its image under $y(d^{k+1}_i)$, which is $\hat{i}_{k+1}$.
	Take $x$ as input: 
	\[
	\begin{tikzcd}[ampersand replacement=\&, column sep=huge]
		\mathclap{\Omega(d^{k+1}_i)(x)} 	\& \mathclap{x} \\
		\mathclap{\hat{i}_{k+1}}			\& \mathclap{y(k+1)}
		\arrow[draw=none, "\subseteq" description, from=1-1, to=1-2]
		\arrow[draw=none, "\subseteq" description, sloped, from=1-1, to=2-1]
		\arrow["\lrcorner"{anchor=center, pos=0.25}, draw=none, from=1-1, to=2-2]
		\arrow[draw=none, "\subseteq" description, sloped, from=1-2, to=2-2]
		\arrow[draw=none, "\subseteq" description, from=2-1, to=2-2]
	\end{tikzcd}
	\]
	All objects are now viewed as elements of $\Omega(k+1)$.
	Pullbacks in the Heyting algebra $\Omega(k+1)$ correspond to the meet operator, i.e.,
	\[
		\Omega(d^{k+1}_i)(x) = \hat{i}_{k+1} \meet x.
		\qedhere
	\]
\end{proof}

\repeatlemma{lem:unique_incidence_except_yn_ynhollow}

\begin{proof}[Proof of \cref{lem:unique_incidence_except_yn_ynhollow}]
    To prove surjectivity, take $x_i \in \Omega(k)$, i.e., $x_i \subseteq y(k)$ for $i=0, \ldots, k+1$.
    By \cref{lem:isom_Omega(n)_Omega(n+1)_leq_delta_i}, each $x_i$ can be seen as a subobject of the $i$\tss{th} face $\ithface$ of $y(k+1)$, and therefore as an element of $\Omega(k+1)$.
    Now that we have only element in $\Omega(k+1)$, let $x \defeq x_{k+1} \join \ldots \join x_0$.
    By \cref{cor:incidence_Omega(n+1)_becomes_intersection}, the element $x$ has incidence tuple $(x_{k+1}, \ldots, x_0)$ as desired.

    Regarding injectivity, suppose $x,x' \in \Omega(k+1)$ have the same incidence tuple, i.e., the same image through $(d_{k+1}, \ldots, d_0)$.
    By \cref{cor:incidence_Omega(n+1)_becomes_intersection}, it follows that they have the same intersection with each face $\ithface$ of $y(k+1)$.
    Hence,
    \begin{align*}
        x \meet \ynhollow[k+1]
        &= x \meet \big( \textstyle \bigjoin_{i=0, \ldots, k+1} \ithface \big) 
        \\
        &= \textstyle\bigjoin_{i=0, \ldots, k+1} (x \meet \ithface)
        \tag*{HA distributivity} \\ 
        &= \textstyle\bigjoin_{i=0, \ldots, k+1} (x' \meet \ithface)
        \tag*{assumption}\\
        &= x' \meet \big( \textstyle \bigjoin_{i=0, \ldots, k+1} \ithface \big) 
        \tag*{HA distributivity} \\ 
        &= x' \meet \ynhollow[k+1].
    \end{align*}
    If $x \meet \ynhollow[k+1] = x' \meet \ynhollow[k+1]$ in the Heyting algebra $\Omega(k+1)$, then $x, x' \in \set{\ynhollow[k+1], y(k+1)}$, which proves the lemma.
\end{proof}

\repeatlemma{lem:topologies_on_ndimsemisSet_unique_extension}

\begin{proof}[Proof of \cref{lem:topologies_on_ndimsemisSet_unique_extension}]
	We prove by induction on $n$ that $j^w$ is a well-defined monotone natural transformation.
	In the case $n=0$, there is nothing to prove about $j^0$ and $j^1$.
	Suppose the lemma true for $n$ and let us prove it for $n+1$.
	We write $j \in \set{j^{w0}, j^{w1}}$ to simplify notations. 
	For $j$ to be a natural transformation, the following diagram must commute for each $i=0, \ldots, n+1$:
	\[
	\begin{tikzcd}[ampersand replacement=\&]
	\Omega(n+1) \& \Omega(n+1) \\
	\Omega(n)   \& \Omega(n)
	\ar[from=1-1, to=1-2, "j_{n+1}", dotted]
	\ar[from=1-1, to=2-1, "\Omega(d_i)"']
	\ar[from=1-2, to=2-2, "\Omega(d_i)"]
	\ar[from=2-1, to=2-2, "j_n"']
	\end{tikzcd}
	\]
	By IH, $j_n$ is already well-defined, and thus the incidence tuple of $j_{n+1}(x)$ must be equal to $\vec{z}$ as defined in \eqref{eq:lem_topologies_on_ndimsemisSet_unique_extension_vec_z}.
	By \cref{lem:unique_incidence_except_yn_ynhollow} there are two possible scenarios.
	\begin{itemize}
	\item 
	In the first scenario, suppose $\vec{z} \neq (y(n), \ldots, y(n))$.
	By \cref{lem:unique_incidence_except_yn_ynhollow}, there is only one element $x' \in \Omega(n+1)$ with this incidence tuple, forcing $j_{n+1}(x) = x'$.
	
	\item 
	The second scenario is when $\vec{z} = (y(n), \ldots, y(n))$, in which case $j_{n+1}(x)$ can be either $\ynhollow[n+1]$ or $y(n+1)$ by \cref{lem:unique_incidence_except_yn_ynhollow}.
	Thus, 
	$\ynhollow[n+1] \leq j_{n+1}(x) \leq y(n+1)$.
	To have monotonicity, we must have:
	\[
	x < \ynhollow[n+1] \implies j_{n+1} (x) \leq j_{n+1} (\ynhollow[n+1]).
	\]
	\begin{itemize}
		\item 
		If $j=j^{w0}$, then by definition $j^{w0}_{n+1} (\ynhollow[n+1]) = \ynhollow[n+1]$, which forces 
		\[
			j^{w0}_{n+1} (x) = \ynhollow[n+1].
		\]
		
		\item 
		If $j=j^{w1}$, then by definition $j^{w1}_{n+1} (\ynhollow[n+1]) = y(n+1)$.
		In order to have idempotence, one cannot have $j^{w1}_{n+1}(x) = \ynhollow[n+1]$.
		Therefore, 
		\[
			j^{w1}_{n+1}(x) = y(n+1).
		\]
	\end{itemize}
	\end{itemize}
	In each scenario, there was only one candidate for the image $j_{n+1}(x)$, which proves that the extension is unique.
\end{proof}

\repeattheorem{thm:topologies_on_ndimsemisSet}

\begin{proof}[Proof of \cref{thm:topologies_on_ndimsemisSet}]
	We prove by induction on $n$ that $j^w$ with $w \in \set{0,1}^{n+1}$ is a topology on $\ndimsemisSet$, i.e., that it satisfies \ref{it:topology_true}, \ref{it:topology_idempotent}, and \ref{it:topology_meet} from \cref{def:topology_on_topos}.
	For the case $n=0$, we already know that $j^0$ and $j^1$ are topologies on $\ndimsemisSet[0] = \Set$, and that these are the only topologies on $\Set$ (cf.~\cref{ex:topologies_on_set}).
	
	Suppose the theorem holds for $n \in \N$.
	We prove it for $n+1$.
	For $j \in \set{j^{w0}, j^{w1}}$, we have by IH that  $j_0, \ldots, j_n$ satisfy the three axioms.
	We check that $j_{n+1}$ satisfies the three axioms too.

	The morphism $j_{n+1}$ has been defined on $y(n+1)$ above in \eqref{eq:topologies_on_ndimsSet_on_yn} in a way that guarantees the first topology axiom \ref{it:topology_true} $j(y(n+1)) = y(n+1))$.
	
	We check axiom \ref{it:topology_idempotent}, i.e., that $j_{n+1}$ is idempotent. 
	Take $x \in \Omega(n+1)$.
	\begin{itemize}[topsep=0pt]
		\item 
		If $x=y(n+1)$, then 
		\[
			j_{n+1}(j_{n+1}(y(n+1)) \stackrel{\text{axiom \ref{it:topology_true}}}{=} y(n+1) \stackrel{\text{axiom \ref{it:topology_true}}}{=} j_{n+1}(y(n+1).
		\]
		
		\item
		If $x=\ynhollow[n+1]$, then:
		\[\begin{array}{rclcl}
			j^{w0}_{n+1} j^{w0}_{n+1} \big( \ynhollow[n+1] \big)
			&\stackrel{\text{def.~\eqref{eq:topologies_on_ndimsSet_on_yn}}}{=} 
			&\ynhollow[n+1]
			&\stackrel{\text{def.~\eqref{eq:topologies_on_ndimsSet_on_yn}}}{=} 
			&j^{w0}_{n+1} \big( \ynhollow[n+1] \big). \\
			j^{w1}_{n+1} j^{w1}_{n+1} \big( \ynhollow[n+1] \big)
			&\stackrel{\text{def.~\eqref{eq:topologies_on_ndimsSet_on_yn}}}{=} 
			&y(n+1)
			&\stackrel{\text{def.~\eqref{eq:topologies_on_ndimsSet_on_yn}}}{=} 
			&j^{w1}_{n+1} \big( \ynhollow[n+1] \big).
		\end{array}\]
		
		\item 
		For $x < \ynhollow[n+1]$, we distinguish cases based on the incidence tuple $\vec{z}$ of $j_{n+1}(x)$.
		
		In the case $\vec{z} \neq (y(n), \ldots, y(n))$, $j_{n+1}(x)$ is uniquely determined by $\vec{z}$.
		Notice that $j_{n+1}(j_{n+1}(x))$ and $j_{n+1}$ have the same incidence tuple: for $0 \leq i \leq n+1$,
		\begin{align*}
			(\Omega(d_i) j_{n+1} j_{n+1}) (x)
			&= (j_n j_n \Omega(d_i))(x)
			\tag*{naturality $jj$} \\
			&= (j_n \Omega(d_i)) (x)
			\tag*{IH: $j_n$ idempotent} \\
			&= (\Omega(d_i)j_n)
			\tag*{naturality $j$}
		\end{align*}
		\[\begin{tikzcd}[ampersand replacement=\&]
			{\Omega(n+1)} \& {\Omega(n+1)} \& {\Omega(n+1)} \\
			{\Omega(n)} \& {\Omega(n)} \& {\Omega(n)}
			\arrow["{j_{n+1}}", from=1-1, to=1-2]
			\arrow["{\Omega(d_i)}"', from=1-1, to=2-1]
			\arrow["{j_{n+1}}", from=1-2, to=1-3]
			\arrow["{\Omega(d_i)}"{description}, from=1-2, to=2-2]
			\arrow["{\Omega(d_i)}", from=1-3, to=2-3]
			\arrow["{j_n}"', from=2-1, to=2-2]
			\arrow["{j_n}"', from=2-2, to=2-3]
		\end{tikzcd}\]
		Therefore, by \cref{lem:unique_incidence_except_yn_ynhollow} we have $j_{n+1} j_{n+1} (x) = j_{n+1}$ .

		In the case $\vec{z} = (y(n), \ldots, y(n))$, then $j_{n+1}(x)$ is either $\ynhollow[n+1]$ or $y(n+1)$.
		In each respective case, we have:
		\[\begin{array}{rclclcl}
			j^{w0}_{n+1} j^{w0}_{n+1} (x)
			&= 
			& j^{w0}_{n+1} (\ynhollow[n+1])
			&= 
			&\ynhollow[n+1]
			&= 
			&j^{w0}_{n+1} (x). 
			\\[1mm]
			j^{w1}_{n+1} j^{w1}_{n+1} (x)
			&= 
			&j^{w1}_{n+1}(y(n+1))
			&= 
			&y(n+1)
			&=
			&j^{w1}_{n+1}(x).
		\end{array}\]
	\end{itemize}
	
	Finally, we check axiom \ref{it:topology_meet}, i.e., that $j$ commutes with $\meet$.
	Take $x,x' \in \Omega(n+1)$.
	We want to prove that $j_{n+1}(x \meet x') = j_{n+1}(x) \meet j_{n+1}(x')$.
	We distinguish cases based on $x$ and $x'$:
	\begin{itemize}
		\item
		If one (or both) of $x$ and $x'$ is $y(n+1)$ which is the greatest element of $\Omega(n+1)$, the desired equality trivially holds.
		
		\item 
		If $x=x'=\ynhollow[n+1]$, then the desired equality also trivially holds.
		
		\item 
		If $x < \ynhollow[n+1]$ and $x'=\ynhollow[n+1]$, then $j_{n+1}(x \meet x') = j_{n+1}(x)$, and
		\[
		j_{n+1}(x) \meet j_{n+1}(x') =
		\begin{cases}
			j_{n+1}(x) \meet \ynhollow[n+1]
			\stackrel{(\ast)}{=} j_{n+1}(x)
			&\text{if } j=j^{w0}, \\
			j_{n+1}(x) \meet y(n+1) = j_{n+1}(x)
			&\text{if } j=j^{w1}.
		\end{cases}
		\]
		Where $(\ast)$ holds because by definition of $j^{w0}$, we have $j^{w0}_{n+1}(x) \leq \ynhollow[n+1]$.
		
		\item 
		If $x,x' < \ynhollow[n+1]$, then we have to look at the incidence tuples.
		We notice that both $j_{n+1}(x \meet x')$ and $j_{n+1}(x) \meet j_{n+1}(x')$ have the same incidence tuple:
		for $0 \leq i \leq n+1$, 
		\begin{equation}
			\label{eq:j(x_meet_z)_and_j(x)_and_j(x)_meet_j(z)_same_incidence}
			\left.
			\begin{array}{rcl}
				\Omega(d_i) (j_{n+1} (x \meet x'))
				& \stackrel{j \text{ nat.}}{=}
				& j_n (\Omega(d_i) (x \meet x'))
				\\
				& \stackrel{\text{Cor.~}\ref{lem:face_map_di_commutes_with_meet}}{=}
				& j_n (\Omega(d_i)(x) \meet \Omega(d_i)(x'))
				\\
				& \stackrel{\text{IH}}{=}
				& j_n \Omega(d_i) (x) \meet j_n \Omega(d_i) (x')
				\\
				& \stackrel{j \text{ nat.}}{=}
				& \Omega(d_i) j_{n+1} (x) \meet \Omega(d_i) j_{n+1} (x')
				\\
				& \stackrel{\text{Cor.~}\ref{lem:face_map_di_commutes_with_meet}}{=}
				& \Omega(d_i) (j_{n+1}(x) \meet j_{n+1}(x'))
			\end{array}
			\right\}
		\end{equation}
		To finish the proof, we distinguish cases based on the incidence tuples of $j_{n+1}(x)$ and $j_{n+1}(x')$.
		\begin{itemize}
			\item 
			Suppose one of them, say w.l.o.g.~$j_{n+1}(x)$, has an incidence tuple that is not $(y(n), \ldots, y(n))$.
			It means there exists an index $i$ such that $\Omega(d_i) j_{n+1}(x) < y(n)$.
			Hence,
			\[
			\begin{array}{rcl}
				\Omega(d_i) (j_{n+1}(x) \meet j_{n+1}(x'))
				&\stackrel{\text{Cor.~}\ref{lem:face_map_di_commutes_with_meet}}{=}&
				\Omega(d_i) j_{n+1} (x) \meet \Omega(d_i) j_{n+1} (x') \\
				&\leq& \Omega(d_i) j_{n+1} (x) \\
				&<& y(n).
			\end{array}
			\]
			Therefore, by \eqref{eq:j(x_meet_z)_and_j(x)_and_j(x)_meet_j(z)_same_incidence} both $j_{n+1}(x \meet x')$ and $j_{n+1}(x) \meet j_{n+1}(x')$ have incidence tuple not equal to $(y(n), \ldots, y(n))$.
			By \cref{lem:unique_incidence_except_yn_ynhollow}, both $j_{n+1}(x \meet x')$ and $j_{n+1}(x) \meet j_{n+1}(x')$ are uniquely determined by their incidence tuple, and because both have the same incidence tuple,
			$j_{n+1}(x \meet x') = j_{n+1}(x) \meet j_{n+1}(x')$.
			
			\item 
			Suppose both $j_{n+1}(x)$ and $j_{n+1}(x')$ have their incidence tuple equal to $(y(n), \ldots, y(n))$, which means they both must be either $\ynhollow[n+1]$ or $y(n+1)$.
			Hence,
			\[
			j_{n+1}(x) \meet j_{n+1}(x') =
			\begin{cases}
				\ynhollow[n+1] \meet \ynhollow[n+1] &\text{if } j=j^{w0} \\
				y(n+1) \meet y(n+1)                 &\text{if } j=j^{w1}
			\end{cases}
			\]
			We now consider $j_{n+1}(x \meet x')$.
			Notice that $x \meet x' \leq x < \ynhollow[n+1]$, thus its image through $j_{n+1}$ is also determined by its incidence tuple.
			Moreover, for $0 \leq i \leq n+1$
			\begin{align*}
				\Omega(d_i)(j_{n+1}(x \meet x'))
				&=
				\Omega(d_i)j_{n+1}(x) \meet \Omega(d_i) j_{n+1}(x)
				\tag*{by \eqref{eq:j(x_meet_z)_and_j(x)_and_j(x)_meet_j(z)_same_incidence}} \\
				&= y(n) \meet y(n)
				\tag*{by assumption} \\
				&= y(n)
			\end{align*}
			In other words, the incidence tuple of $j_{n+1}(x \meet x')$ is also $(y(n), \ldots, y(n))$.
			Hence, 
			\[
			j_{n+1}(x \meet x') =
			\begin{cases}
				\ynhollow[n+1] &\text{if } j=j^{w0} \\
				y(n+1)         &\text{if } j=j^{w1}
			\end{cases}
			\]
			Therefore, $j_{n+1}(x \meet x') = j_{n+1}(x) \meet j_{n+1}(x')$.
		\end{itemize}
	\end{itemize}
	We have proven all three axioms, and $j^{w0}$ and $j^{w1}$ are therefore topologies.
	
	To complete the induction step, it remains to prove that all topologies $j'$ on $\ndimsemisSet[n+1]$ are of the form $j^w$ for some bit string $w \in \set{0,1}^{n+2}$.
	Let $j'$ be a topology on $\ndimsemisSet[n+1]$.
	Then, $(j'_k)_{0 \leq k \leq n}$ is a topology on $\ndimsemisSet$.
	Hence, by IH, $(j'_k)_{0 \leq k \leq n}$ must be of the form $j^w$ for some bit string $w \in \set{0,1}^{n+1}$.
	We now show that $j'$ is equal to either $j^{w0}$ or $j^{w1}$ by showing that $j'_{n+1} = j^{w0}_{n+1}$ or $j'_{n+1} = j^{w1}_{n+1}$.
	We look at the image $j'_{n+1}(\ynhollow[n+1])$.
	By naturality of $j'$, the following square commutes:
	\[\begin{tikzcd}[ampersand replacement=\&, sep=tiny]
		{\Omega(n+1)} \&\&\& {\Omega(n+1)} \\
		\& {\ynhollow[n+1 ]} \& {j'_{n+1}(\ynhollow[n+1])} \\
		\& {(y(n), \ldots, y(n))} \& {(y(n), \ldots, y(n))} \\
		{\Omega(n)^{n+1}} \&\&\& {\Omega(n)^{n+1}}
		\arrow["{j'_{n+1}}", from=1-1, to=1-4]
		\arrow["{(\Omega(d_n), \ldots, \Omega(d_0))}"', from=1-1, to=4-1]
		\arrow["{(\Omega(d_n), \ldots, \Omega(d_0))}", from=1-4, to=4-4]
		\arrow[maps to, from=2-2, to=2-3]
		\arrow[maps to, from=2-2, to=3-2]
		\arrow[maps to, from=2-3, to=3-3]
		\arrow[maps to, from=3-2, to=3-3]
		\arrow["{(j'_n)^{n+1} = (j^w_n)^{n+1}}"', from=4-1, to=4-4]
	\end{tikzcd}\]
	In other words, the incidence tuple of $j'_{n+1}(\ynhollow[n+1])$ is $(y(n), \ldots, y(n))$.
	By \cref{lem:unique_incidence_except_yn_ynhollow}, only two elements of $\Omega(n+1)$ have this incidence tuple: $\ynhollow[n+1]$ and $y(n+1)$.
	Hence, $j'_{n+1}(\ynhollow[n+1])$ can only be either $\ynhollow[n+1]$ or $y(n+1)$, which gives the definitions of $j_{n+1}^{w0}$ and $j_{n+1}^{w1}$, respectively (\cref{def:topologies_on_ndimsemisSet}).
	We moreover showed in \cref{lem:topologies_on_ndimsemisSet_unique_extension} that there is a unique way of extending the definition of $j^{w0}$ and $j^{w1}$ into a monotone idempotent natural transformation.
	This concludes the proof.
\end{proof}

\subsection{Proofs of \texorpdfstring{\cref{subsec:leqn}}{Section 3.3}}

\repeatlemma{lem:sx_in_subobject_iff_x_in_subobject}

\begin{proof}[Proof of \cref{lem:sx_in_subobject_iff_x_in_subobject}]
	This is simply a naturality square for the inclusion presheaf morphism:
	\[
	\begin{tikzcd}
		A'(k+1) & A(k+1) \\
		A'(k)   & A(k)
		\ar[from=1-1, to=1-2, tail, "\inclusion"]
		\ar[from=2-1, to=2-2, tail, "\inclusion"']
		\ar[from=2-1, to=1-1, "A'(s^k_i)"]
		\ar[from=2-2, to=1-2, "A(s^k_i)"']
	\end{tikzcd}
	\qedhere
	\]
\end{proof}

Next is a technical lemma needed to prove \cref{thm:topologies_on_ndimsSet}, the main result of \cref{subsec:leqn}.

\begin{lemma}
	\label{lem:s_i_on_true_and_on_emptyset}
	For all $k \leq n$, and $0 \leq i \leq k$, the function 
	\[
		\Omega\leqn(s_i) : \Omega\leqn(k) \to \Omega\leqn(k+1)
	\]
	\begin{itemize}[topsep=1pt, noitemsep]
		\item preserves the top element: $\Omega\leqn(s_i)(y\leqn(k)) = y\leqn(k+1)$, and
		\item preserves the bottom element: $\Omega\leqn(s_i)(\emptyset) = \emptyset$.
	\end{itemize}
\end{lemma}

\begin{proof}
	Recall 
    that $\Omega\leqn(s_i)$ is defined as taking the pullback along $y(\opcat{s_i})$ (in the category $\ndimsSet$): given some subobject $m \colon x \mono y\leqn(k)$,
	\[
		\begin{tikzcd}
			\cdot 			& x \\
			y\leqn(k+1)		& y\leqn(k)
			\ar[from=1-1, to=1-2, dotted]
			\ar[from=1-1, to=2-1, dotted, tail, "\Omega\leqn(s_i)(m)"']
			\ar[from=1-2, to=2-2, tail, "m"]
			\ar[from=2-1, to=2-2, "y\leqn(\opcat{s_i})"']
			\ar[from=1-1, to=2-2, "\lrcorner"{anchor=center, pos=0.125}, draw=none]
		\end{tikzcd}
	\]
	The pullback of an identity morphism $\id: y\leqn(k) \mono y\leqn(k)$ is also an identity morphism $\id: y\leqn(k+1) \mono y\leqn(k+1)$.
	Moreover, the pullback of $\emptyset \mono y\leqn(k)$ can be computed pointwise in $\Set$, and gives $\emptyset \mono y\leqn(k+1)$.
\end{proof}

The next lemma says that thanks to the simplicial identities \eqref{eq:first_simplicial_identity} and \eqref{eq:simplicial_identities}, any morphism $f \in \opcat{(\Delta\leqn)}$ can be reordered by swapping the face maps at the front of the composition and the degeneracies at the back.

\begin{lemma}[{\cite[Lemma p.177, VII.5]{MacLane_1971}}]
	\label{lem:Delta_morphism_unique_representation}
	In $\opcat{(\Delta\leqn)}$, any morphism $f : k \to l$ has a unique representation
	\begin{equation}
		\label{eq:Delta_morphism_unique_representation}
		f = s^{l-1}_{i_1} 
		\,\mycirc\, 
		\ldots 
		\,\mycirc\, 
		s^{l-m}_{i_m} 
		\,\mycirc\, 
		d^{(l-m)+1}_{i'_1} 
		\,\mycirc\,
		\ldots 
		\,\mycirc\,
		d^{(l-m)+m'=k}_{i'_{m'}}.
	\end{equation}
	where the strings of subscripts $i$ and $i'$ satisfy
	$0 \leq i'_1 < \ldots < i'_{m'} \leq k$ and 
	$l-1 \geq i_1 > \ldots > i_m \geq 0$.
	When $m=0$, $f$ is represented by face maps only and is therefore a $\opcat{(\Delta\leqn\semi)}$-morphism.
\end{lemma}

Similarly to \cref{def:ith_face}, we have the following definition.

\begin{definition}
	For all $k, i \in \N$ such that $0 < k$ and $0 \leq i \leq k$, the \myemph{$i$\tss{th} face} of $y\leqn(k)$ is the least subpresheaf $x$ of $y\leqn(k)$ such that
	\[
		d^k_i \in x(k-1)
		\subseteq
		y\leqn(k)(k-1)
		\opcat{(\Delta\leqn)}(k,k-1)
	\] 
\end{definition}

\begin{remark}
	\label{rem:faces_yleqnsemi}
	Note that in $\ndimsemisSet$, the $i$\tss{th} face of $y\leqn\semi(k)$ (\cref{def:ith_face}) consisted of compositions of face maps.
	However, in $\ndimsSet$, the faces of $y\leqn(k)$ also contain degeneracy maps.
	In particular, the face $\hat{i}_k$ of $y\leqn(k)$ is equal to the $i$\tss{th} face of $y\leqn\semi(k)$ with all degeneracies added.
	For example:
	\begin{itemize}[topsep=4pt, noitemsep]
		\item 
		\makebox[2cm][l]{In $\ndimsemisSet$:}
		\(
		\hat{1}_1 =
		{%
			\newcommand{\va}{\node (va) at (-2mm,0mm) [vertex] {};}%
			\newcommand{\vb}{\node (vb) at (2mm,0mm) [vertex] {};}%
			\begin{tikzpicture}[node distance=15mm,baseline=-1mm,loop/.style={->,densely dotted,distance=3mm}]
				\begin{scope}[local bounding box=a]
					\graphnode[]{eab}{6mm}{2mm}{}{\va}
				\end{scope}
			\end{tikzpicture}%
		} 
		\subseteq
		{%
			\newcommand{\va}{\node (va) at (-2mm,0mm) [vertex] {};}%
			\newcommand{\vb}{\node (vb) at (2mm,0mm) [vertex] {};}%
			\begin{tikzpicture}[node distance=15mm,baseline=-1mm,loop/.style={->,densely dotted,distance=3mm}]
				\begin{scope}[local bounding box=a]
					\graphnode[]{eab}{6mm}{2mm}{}{\va\vb\edge[->]{va}{vb}}
				\end{scope}
			\end{tikzpicture}%
		}
		\)
		
		\item 
		\makebox[2cm][l]{In $\ndimsSet$:}
		\(
		\hat{1}_1 =
		{%
			\newcommand{\vx}{\node (vx) at (-2mm,0mm) [vertex] {};}%
			\begin{tikzpicture}[node distance=15mm,baseline=-1mm,loop/.style={->,densely dotted,distance=3mm}]
				\graphnode[]{gx}{10mm}{2mm}{}{\vx \draw [loop] (vx) to[out=180-35,in=180+35] (vx);}
			\end{tikzpicture}%
		}
		\subseteq
		{%
			\newcommand{\va}{\node (va) at (-2mm,0mm) [vertex] {};}%
			\newcommand{\vb}{\node (vb) at (2mm,0mm) [vertex] {};}%
			\begin{tikzpicture}[node distance=15mm,baseline=-1mm,loop/.style={->,densely dotted,distance=3mm}]
				\begin{scope}[local bounding box=a]
					\graphnode[]{eab}{10mm}{2mm}{}{\va\vb\edge[->]{va}{vb}
						\draw [loop] (va) to[out=180-35,in=180+35] (va);
						\draw [loop] (vb) to[out=  0-35,in=  0+35] (vb);}
				\end{scope}
			\end{tikzpicture}%
		}
		\)
	\end{itemize}
	Therefore, a morphism $f \colon k \to l$ in $\opcat{(\Delta\leqn)}$ is in $\hat{i}_k(l) \subseteq y\leqn(k)(l)$ if and only if its unique representation given in \eqref{eq:Delta_morphism_unique_representation} satisfies $i'_m = i$.
\end{remark}

\repeatlemma{lem:F_k_between_semicase_and_generalcase}

\begin{proof}[Proof of \cref{lem:F_k_between_semicase_and_generalcase}]
    We show the claims one by one for $F_k$. The proofs for $F_k \inv$ are similar.
	\begin{itemize}
		\item
		For $F_k(x\semi)$ to be well-defined, it means $x\semi \cup \degen(x\semi,-) \subseteq y\leqn(k)$, which holds by \cref{lem:sx_in_subobject_iff_x_in_subobject}.
		
		\item
		We show $F_k$ injective.
		Take $x\semi$ and $x'\semi$ in $\Omega\leqn\semi(k)$ with 
		\begin{equation}
			\label{eq:F_k_injective_hypothesis}
			x\semi \cup \degen(x\semi,-) = x'\semi \cup \degen(x'\semi,-).
		\end{equation}
		We prove that $x\semi(l) = x'\semi(l)$ by induction on $l \leq n$:
		\begin{itemize}[topsep=0pt]
			\item For $l=0$, there are no degenerate vertices, so
			\[
			x\semi(0) 
			\stackrel{\text{Def.~\ref{def:degen_set}}}{=} x\semi(0) \cup \degen(x\semi,0)
			\stackrel{\text{\eqref{eq:F_k_injective_hypothesis}}}{=} x'\semi(0) \cup \degen(x'\semi,0)
			\stackrel{\text{Def.~\ref{def:degen_set}}}{=} x'\semi(0).
			\]
			
			\item 
			Suppose true for $l$.
			Given $z \in y\leqn(k)(l+1)$, we have
			\begin{align}
				\begin{split}
					\label{eq:F_k_injective_intermediary_equation}
					z \in &~\degen(x\semi,l+1) \\
					&\stackrel{\mathclap{\text{Def.~\ref{def:degen_set}}}}{\iff}
					z = s^l_i \cdot f, \text{ some } f \in x\semi(l)\cup\degen(x\semi,l) \\
					&\stackrel{\mathclap{\text{IH \eqref{eq:F_k_injective_hypothesis}}}}{\iff}
					z = s^l_i \cdot f, \text{ some } f \in x'\semi(l)\cup\degen(x'\semi,l) \\
					&\stackrel{\mathclap{\text{Def.~\ref{def:degen_set}}}}{\iff}
					z \in \degen(x'\semi,l+1)
				\end{split}
			\end{align}
			Thus:
			\begin{align*}
				&x\semi(l+1) \\
				&= \big( x\semi(l+1) \cup \degen(x\semi,l+1) \big) \setminus \degen(x\semi, l+1)
				\tag*{$F,F\inv$ inv.}\\
				&= \big( x'\semi(l+1) \cup \degen(x'\semi,l+1) \big) \setminus \degen(x\semi, l+1)
				\tag*{by \eqref{eq:F_k_injective_hypothesis}} \\
				&= \big( x'\semi(l+1) \cup \degen(x'\semi,l+1) \big) \setminus \degen(x'\semi, l+1)
				\tag*{by \eqref{eq:F_k_injective_intermediary_equation}} \\
				&= x'\semi(l+1) 
			\end{align*}    
		\end{itemize}
		
		\item
		To see that $F_k$ is surjective, simply observe that given $x \in \Omega\leqn(k)$, i.e., $x \subseteq y\leqn(k)$, then $x = F_k(x \setminus \degen(x,-))$, which is the definition of $F\inv$.
		
		\item
		We show that $F_k$ preserves truth by proving by induction on $l$ that
		\[
		y\leqn\semi(k)(l) \cup \degen(y\leqn\semi(k),l) = y\leqn(k)(l).
		\]
		\begin{itemize}
			\item 
			For $l=0$, there are no degenerate vertices $\degen(y\leqn\semi(k),0) = \emptyset$, hence
			\begin{align*}
				y\leqn\semi(k)(0) 
				&= \opcat{(\Delta\leqn\semi)}(k,0) \\
				&\stackrel{(*)}{=} \opcat{(\Delta\leqn)}(k,0) \\
				&= y\leqn(k)(0).
			\end{align*}
			To see why $(*)$ holds, we show both inclusions.\\
			$(\subseteq)$ Every $\opcat{(\Delta\leqn\semi)}$-morphism is also a $\opcat{(\Delta\leqn)}$-morphism. \\
			$(\supseteq)$ All $\opcat{(\Delta\leqn)}$-morphisms with codomain $0$, when written in their unique representation of \cref{lem:Delta_morphism_unique_representation}, cannot be made of degeneracy maps.
			
			\item 
			Suppose true for $l$.
			We show the desired equality for $l+1$ by double inclusion.\\
			$(\subseteq)$
			Take $f \in y\leqn\semi(k)(l+1)$, i.e., $f \colon k \to l+1$.
			Every $\opcat{(\Delta\leqn\semi)}$-morphism is also a $\opcat{(\Delta\leqn)}$-morphism.
			Hence $f \in y\leqn(k)(l+1)$.\\
			Take $f \in \degen(y\leqn\semi,l+1)$, i.e., $f = s^l_i \cdot f'$ for some $f' \in y\leqn\semi(k)(l) \cup \degen(y\leqn\semi(k),l)$.
			By IH, $f' \in y\leqn(k)(l)$.
			Hence, $f = s^l_i \cdot f' \in y\leqn(k)(l+1)$.
			\\
			$(\supseteq)$
			Take $f \in y\leqn(k)(l+1)$, i.e., $f : k \to l+1$ in $\opcat{(\Delta\leqn)}$.
			By \cref{lem:Delta_morphism_unique_representation}, write $f$ in its unique representation 
			\[
			f = s^l_{i_1} \ldots s^{l-m+1}_{i_{m}} d^{l-m+2}_{i'_1} \ldots d^{l-m+m'-1=k}_{i'_{m'}}.
			\]
			In the case where $m=0$, there is no degeneracy in the composition above, hence $f \in y\leqn\semi(k)(l+1)$.
			Otherwise, $m \geq 1$ and $f=s^l_{i_1} (g)$ for
			\[
			g \defeq s^{l-1}_{i_2} \ldots s^{l-m+1}_{i_{m}} d^{l-m+2}_{i'_1} \ldots d^{l-m+m'-1=k}_{i'_{m'}}.
			\]
			Hence, $g \in y\leqn(k)(l) \stackrel{\text{IH}}{=} y\leqn\semi(k)(l) \cup \degen(y\leqn\semi(k),l)$.
			By \cref{def:degen_set}, this effectively means $f \in \degen(y\leqn\semi(k),l+1)$.
		\end{itemize}
		
		\item 
		We show that $F_k$ commutes with $\land$.
		Given $x\semi$ and $x'\semi$ in $\Omega\leqn\semi(k)$, we have
		\[
		\begin{array}{rcl}
			F(x\semi \cap x'\semi)
			&=& (x\semi \cap x'\semi) \cup \degen(x\semi \cap x'\semi, -) \\
			F(x\semi) \cap F(x'\semi)
			&=& \big(x\semi \cup \degen(x\semi,-)\big) \cap \big(x'\semi \cup \degen(x'\semi,-)\big) \\
			&=& (x\semi \cap x'\semi) 
			\cup \big(\degen(x\semi,-) \cap \degen(x'\semi,-)\big)
			\\
			&& \cup \underbrace{\big(x'\semi \cap \degen(x\semi,-)\big)}_{=\emptyset} 
			\cup \underbrace{\big(x\semi \cap \degen(x'\semi,-)\big)}_{=\emptyset}
			\\
			&=& (x\semi \cap x'\semi) \cup \big(\degen(x\semi,-) \cap \degen(x'\semi,-)\big)
		\end{array}
		\]
		It suffices to prove by induction that for all $l \in \N$: 
		\[
			\degen(x\semi \cap x'\semi, l) = \degen(x\semi,l) \cap \degen(x'\semi,l).
		\]
		\begin{itemize}
			\item
			For $l=0$, we indeed have $\emptyset = \emptyset \cap \emptyset$.
			
			\item 
			Suppose true for $l$.
			For $l+1$, we show equality via a double-inclusion:
			\begin{itemize}
				\item[$(\subseteq)$]
				Take $s^l_i \cdot f \in \degen(x\semi \cap x'\semi,l+1)$, for some $0 \leq i \leq l+1$ and 
				\[
				f \in \big( x\semi(l) \cap x'\semi(l) \big) \cup \degen(x\semi \cap x'\semi,l).
				\]
				If $f \in x\semi(l) \cap x'\semi(l)$, then $s^l_i \cdot f \in \degen(x\semi,l+1) \cap \degen(x'\semi,l+1)$.
				Otherwise, $f \in \degen(x\semi \cap x'\semi,l) \stackrel{\text{IH}}{=} \degen(x\semi,l) \cap \degen(x'\semi,l)$, which also implies that $s^l_i \cdot ) \in \degen(x\semi,l+1) \cap \degen(x'\semi,l+1)$.
				
				\item[$(\supseteq)$]
				Take $f \in \degen(x\semi,l+1) \cap \degen(x'\semi,l+1)$, 
				and consider the unique representation (\cref{lem:Delta_morphism_unique_representation}) of $f$:
				\[
				f = s_{i_1} \underbrace{s_{i_2} \ldots s_{i_m} d_{i'_1} \ldots d_{i'_{m'}}}_{\eqdef g}.
				\]
				By hypothesis, 
				$g  \in x\semi(l) \cup \degen(x\semi,l)$ and $g \in x'\semi(l) \cup \degen(x'\semi,l)$

				When $g \in x\semi(l) \cap x'\semi(l)$, then $f \in \degen(x\semi \cap x'\semi, l+1)$.
				
				The case where $g \in x\semi(l) \cap \degen(x', l)$, or vice-versa, is impossible, because $g$ cannot be both non-degenerate and degenerate.
				
				Lastly, when $g \in \degen(x\semi,l) \cap \degen(x'\semi,l) \smash{\stackrel{\text{IH}}{=}} \degen(x\semi \cap x'\semi,l)$, then $f \in \degen(x\semi \cap x'\semi, l+1)$.
			\end{itemize}
		\end{itemize}
		
		\item 
		Lastly, we show $F$ commutes with face maps. 
		Using \cref{cor:incidence_Omega(n+1)_becomes_intersection}, we reformulate the desired equality as follows, where $\hat{i}_k$ is the $i$\tss{th} face of $y\leqn\semi(k)$: for each $l \in \N$
		\[
		(x\semi(l) \cap \hat{i}_k(l)) \cup \degen(x\semi \cap \hat{i}_k,l)
		= 
		\big(x\semi(l) \cup \degen(x\semi,l)\big) \meet \hat{i}_k(l).
		\]
		For a $\opcat{(\Delta\leqn\semi)}$-morphism $f \colon k \to l$ to be in $\hat{i}_k(l)$, we must have the unique representation (\cref{lem:Delta_morphism_unique_representation}) of $f$ is of the form $f = g \cdot d^{k}_i$ (cf.~\cref{rem:faces_yleqnsemi}).
		We prove the desired equality by induction on $l$.
		
		For $l=0$, we indeed have $x\semi(0) \cap \hat{i}_k(0) = x\semi(0) \cap \hat{i}_k(0)$.
		
		Suppose true for $l$ and let us prove the desired equality with $l+1$.
		First note that
		\begin{align}
				&\degen(x\semi \meet \hat{i}_k, l+1) \nonumber \\
				&= \setvbar{s^l_{i'} \cdot h}{h \in \big(x\semi(l) \meet \hat{i}_k(l) \big) \cup \degen(x\semi \meet \hat{i}_k, l),\ i' \leq l+1} \nonumber \\
				&\stackrel{\text{IH}}{=}  \setvbar{s^l_{i'} \cdot h}{h \in \big(x\semi(l) \cup \degen(x\semi, l) \big) \meet \hat{i}_k(l),\ i' \leq l+1} 
				\label{eq:F_k_commutes_di_IH}
		\end{align}
		We prove the desired equality by double inclusion.
		Take $f \colon k \to l+1$ in $\opcat{(\Delta\leqn)}$.
		\begin{itemize}
			\item[$(\subseteq)$]
			Assume $f \in (x\semi \meet \hat{i}_k)(l+1) \cup \degen(x\semi \meet \hat{i}_k, l+1)$.
			If $f \in (x\semi \meet \hat{i}_k)(l+1) $, then $f \in \big(x\semi(l+1) \cup \degen(x\semi,l+1)\big) \cap \hat{i}_k(l+1)$ directly.
			Otherwise, suppose $f \in \degen(x\semi \meet \hat{i}_k, l+1)$.
			By \eqref{eq:F_k_commutes_di_IH}, we can write $f = s^l_{i'} \cdot h$ for some $h \in \big(x\semi(l) \cup \degen(x\semi, l) \big) \cap \hat{i}_k(l)$.
			The fact that $h \in \hat{i}_k(l)$ means it can be written as a composition in $\opcat{(\Delta\leqn\semi)}$ of the form $g \cdot d^{k}_i$ for some $g \colon k-1 \to l$.
			Therefore, $f = s^l_{i'} \cdot g \cdot d^{k}_i$	and thus belongs to $\hat{i}_k(l+1)$.
			Moreover, having $f=s^l_{i'} \cdot h$ with $h \in \big(x\semi(l) \cup \degen(x\semi, l) \big)$ implies that $f \in \degen(x\semi, l+1)$.
			
			\item[$(\supseteq)$]
			Assume $f \in x\semi(l+1) \cup \degen(x\semi,l+1)$ and $f \in \hat{i}_k(l+1)$.
			If $f \in x\semi(l+1)$, then $f \in x\semi(l+1)\meet \hat{i}_k(l+1)$, as desired.
			Otherwise, assume $f \in \degen(x\semi,l+1)$, which means that $f= s^l_{i'} \cdot h$ for some $h \in x\semi(l) \cup \degen(x\semi,l)$.
			Since $f=s^l_{i'} \cdot h \in \hat{i}_k(l+1)$, the morphism $h$ can be written as a composition of the form $g \cdot d^{k}_i$ for some $g \colon k-1 \to l$, and thus $h \in \hat{i}_k(l)$.
			Hence, by \eqref{eq:F_k_commutes_di_IH}, $f \in \degen(x\semi \meet \hat{i}_k, l+1)$.
			\qedhere
		\end{itemize}
	\end{itemize}
\end{proof}

\repeatlemma{cor:injection_of_topologies}

\begin{proof}[Proof of \cref{cor:injection_of_topologies}]
	Recall that topologies in presheaf categories are bounded meet-semilattice morphisms (cf.~\cref{lem:topology_monotone}).
	Therefore, by the bounded meet-semilattice isomorphism in \cref{lem:F_k_between_semicase_and_generalcase}, we may, without loss of generality, identify $\Omega\leqn$ and $\Omega\leqn\semi$.
	The only distinction between a topology $j$ on $\ndimsSet$ and a topology $j^w$ on $\ndimsemisSet$ is that $j$ must satisfy additional naturality squares for the degeneracy maps.
	Specifically, for each $0 \leq k < k+1 \leq n$, and $0 \leq i \leq k$, the following diagram must commute:
	\[\begin{tikzcd}[sep=scriptsize]
		\Omega\leqn(k+1) 	& \Omega\leqn(k+1) \\
		\Omega\leqn(k) 		& \Omega\leqn(k)
		\ar[from=1-1, to=1-2, "j^w_{k+1}"]
		\ar[from=2-1, to=1-1, "\Omega\leqn(s_i)"]
		\ar[from=2-2, to=1-2, "\Omega\leqn(s_i)"']
		\ar[from=2-1, to=2-2, "j^w_k"']
	\end{tikzcd}\]
	Therefore, each topology on $\ndimsSet$ correspond to a topology on $\ndimsemisSet$ but the converse does not necessarily hold, which gives the desired injection between the set of topologies.
\end{proof}

\repeattheorem{thm:topologies_on_ndimsSet}

\begin{proof}[Proof of \cref{thm:topologies_on_ndimsSet}]
    In this proof, we simplify notations as follows: we identify each $j$ and $F\inv j F = j^w$, we denote $y\leqn$ by $y$, $\Omega\leqn$ by $\Omega$, $\Omega\leqn(s^k_i)$ by $s_i$, and $\Omega\leqn(d^k_i)$ by $d_i$.
	
	Note that a string $w \in \set{0,1}^{n+1}$ is of the form $w=0^m1^{n+1-m}$ for some $0 \leq m \leq n$ if and only if 10 is not a substring of $w$.
	As mentioned in the proof of \cref{cor:injection_of_topologies},
	for a topology $j^w$ on $\ndimsemisSet$ to be a topology on $\ndimsSet$, it must fulfil an extra requirement: the naturality square \eqref{eq:naturality_square_degeneracy_maps} must commute for each degeneracy map $s_i$.
	\begin{equation}
		\label{eq:naturality_square_degeneracy_maps}
		\begin{tikzcd}[sep=scriptsize]
			\Omega(k+1) 	& \Omega(k+1) \\
			\Omega(k) 		& \Omega(k)
			\ar[from=1-1, to=1-2, "j^w_{k+1}"]
			\ar[from=2-1, to=1-1, "s_i"]
			\ar[from=2-2, to=1-2, "s_i"']
			\ar[from=2-1, to=2-2, "j^w_k"']
		\end{tikzcd}
	\end{equation}
	
	To prove the theorem, it suffices to prove the following: for all topologies $j^w$ on $\ndimsemisSet$ with $w = w_0 \ldots w_n \in \set{0,1}^{n+1}$, the following are equivalent:
	\begin{enumerate}[(a), noitemsep]
		\item 
		\label{it:proof_topologies_on_ndimsSet_1}
		For all $k \in \N$ such that $0 \leq k < k+1 \leq n$: $(j^w_k)_{0 \leq k \leq n}$ is a topology on $\ndimsSet$, i.e., \eqref{eq:naturality_square_degeneracy_maps} commutes.
		
		\item 
		\label{it:proof_topologies_on_ndimsSet_2}
		For all $k \in \N$ such that $0 \leq k < k+1 \leq n$: $w_k w_{k+1} \in \set{00, 01, 11}$.
	\end{enumerate}
	We prove that \ref{it:proof_topologies_on_ndimsSet_1} $\Leftrightarrow$ \ref{it:proof_topologies_on_ndimsSet_2} by induction on $n$.
	
	For $n=0$, we have $\ndimsemisSet[0] = \ndimsSet[0] = \Set$ (cf.~\cref{ex:simplicial_sets}).
	The two topologies $j^0$ and $j^1$ on $\Set$ (\cref{ex:topologies_on_set}) vacuously satisfy \ref{it:proof_topologies_on_ndimsSet_1} and \ref{it:proof_topologies_on_ndimsSet_2}, as there exists no $k \in \N$ such that $0 \leq k < k+1 \leq 0$.
	This concludes the base case.
	
	For the induction step, suppose \ref{it:proof_topologies_on_ndimsSet_1} $\Leftrightarrow$ \ref{it:proof_topologies_on_ndimsSet_2} holds for $n$ and let us prove this equivalence for $n+1$.
	Take a topology $j^w$ on $\ndimsemisSet[n+1]$ with $w = w_0 \ldots w_{n+1} \in \set{0,1}^{n+2}$.

	\ref{it:proof_topologies_on_ndimsSet_1} $\Leftarrow$ \ref{it:proof_topologies_on_ndimsSet_2}.
	Suppose $w_k w_{k+1} \in \set{00, 01, 11}$ for all $k \in \N$ such that $0 \leq k < k+1 \leq n+1$.
	In particular, this holds for all $k \in \N$ such that $0 \leq k < k+1 \leq n$, hence by IH, $(j^w_k)_{0 \leq k \leq n}$ is a topology on $\ndimsSet$ and \eqref{eq:naturality_square_degeneracy_maps} commutes for $k = 0, \ldots, n-1$.
	We prove \eqref{eq:naturality_square_degeneracy_maps} commutes for $k=n$.
	It trivially commutes on input $y(n) \in \Omega(n)$:
	\[\begin{tikzcd}[ampersand replacement=\&, column sep=7em]
		{y(n+1)} 	\& {y(n+1)} \\
		{y(n)} 		\& {y(n)}
		\arrow[maps to, from=1-1, to=1-2, "j^w_{n+1}", "\text{by \cref{def:topology_on_topos} \ref{it:topology_true}}"']
		\arrow[maps to, from=2-1, to=1-1, "s_i"', "\text{by \cref{lem:s_i_on_true_and_on_emptyset}}"]
		\arrow[maps to, from=2-1, to=2-2, "j^w_{n}"', "\text{by \cref{def:topology_on_topos} \ref{it:topology_true}}"]
		\arrow[maps to, from=2-2, to=1-2, "s_i", "\text{by \cref{lem:s_i_on_true_and_on_emptyset}}"']
	\end{tikzcd}\]
	Let us prove \eqref{eq:naturality_square_degeneracy_maps} commutes for $x \in \Omega(n)$ with $x \leq \ynhollow$.
	For an arbitrary $0 \leq i \leq n$, consider the element $s_i(x) \in \Omega(n+1)$.
	By definition of $j^w$ (\cref{def:topologies_on_ndimsemisSet,lem:topologies_on_ndimsemisSet_unique_extension}), there are two possible values for $j^w_{n+1} (s_i(x))$, depending on its incidence tuple 
	\begin{align*}
		\vec{z} 
		&\defeq (d_{n+1}, \ldots, d_0)(j^w_{n+1}(s_i(x))) \\
		&= ((j^w_n d_{n+1} s_i) (x), \ldots, (j^w_n d_0 s_i) (x)).
		\tag*{$j^w$ commutes with $d$'s}
	\end{align*}
	\begin{itemize}
		\item 
		Suppose first that $\vec{z} \neq (y(n), \ldots, y(n))$.
		Then, there is a unique element in $\Omega(n+1)$ with $\vec{z}$ as incidence tuple, and $j^w_{n+1}(s_i(x))$ is defined as this unique element.
		To prove that $j^w_{n+1}(s_i(x)) = s_i(j^w_n(x))$, we show that the right-hand side $s_i(j^w_n(x))$ also has $\vec{z}$ as incidence tuple: for $0 \leq i' \leq n+1$,
		\begin{align*}
			&(d_{i'} s_i j^w_n)(x) \\
			&= 
			\begin{cases}
				(s_i d_{i'-1} j^w_n)(x) 	&\text{if } i' > i+1 \\
				j^w_n(x)					&\text{if } i' \in \set{i+1, i} \\
				(s_{i-1} d_{i'} j^w_n)(x)	&\text{if } i' < i
			\end{cases}
			\tag*{simplicial identities \eqref{eq:simplicial_identities}} \\
			&= 
			\begin{cases}
				(s_i j^w_n d_{i'-1})(x) 	&\text{if } i' > i+1 \\
				j^w_n(x)					&\text{if } i' \in \set{i+1, i} \\
				(s_{i-1} j^w_n d_{i'})(x)	&\text{if } i' < i
			\end{cases}
			\tag*{$j^w$ commutes with $d$'s} \\
			&= 
			\begin{cases}
				(j^w_n s_i d_{i'-1})(x) 	&\text{if } i' > i+1 \\
				j^w_n(x)					&\text{if } i' \in \set{i+1, i} \\
				(j^w_n s_{i-1} d_{i'})(x)	&\text{if } i' < i
			\end{cases}
			\tag*{IH: $j^w_n$ commutes with $s$'s} \\
			&= 
			(j^w_n d_{i'} s_i)(x)
			\tag*{simplicial identities \eqref{eq:simplicial_identities}} 
		\end{align*}
		Hence, $(d_{n+1}, \ldots, d_0)(s_i(j^w_n(x))) = \vec{z}$.
		By uniqueness of the element in $\Omega(n+1)$ having $\vec{z}$ as incidence tuple, we must have $j^w_{n+1}(s_i(x)) = s_i(j^w_n(x))$, which concludes the case.
		
		\item 
		Suppose now that $\vec{z} = (y(n), \ldots, y(n))$.
		By definition of $j^w$ (\cref{def:topologies_on_ndimsemisSet,lem:topologies_on_ndimsemisSet_unique_extension}), this can only happen if $w_n = 1$.
		Because we assumed $w_n w_{n+1} \in \set{00, 01, 11}$, we are necessarily in the case $w_n w_{n+1} = 11$.
		Notice that
		\begin{equation}
			\label{eq:j^w_x(x)=y(n)}
			\begin{array}{rcccl}
				j^w_n(x) 
				&\stackrel{\text{identity \eqref{eq:simplicial_identities}}}{=}&
				(j^w_n d_i s_i) (x)
				&\stackrel{\text{Assumption}}{=}&
				y(n).
			\end{array}
		\end{equation}
		Moreover, by definition of $j^w$, the fact that $j^w_{n+1} (s_i (x))$ has incidence tuple $(y(n), \ldots, y(n))$ and the fact $w_{n+1} = 1$ together imply that 
		\begin{equation}
			\label{eq:j^w_n+1(s_i(x))=y(n+1)}
			j^w_{n+1} (s_i(x)) = y(n+1).
		\end{equation}
		Together with the fact that $s_i(y(n)) = y(n+1)$ by \cref{lem:s_i_on_true_and_on_emptyset}, the diagram \eqref{eq:naturality_square_degeneracy_maps} thus commutes:
		\[\begin{tikzcd}[ampersand replacement=\&]
			{s_i(x)} \& {y(n+1)} \\
			x \& {y(n)}
			\arrow[maps to, from=1-1, to=1-2, "j^w_{n+1}", "\text{by \eqref{eq:j^w_n+1(s_i(x))=y(n+1)}}"']
			\arrow[maps to, from=2-1, to=1-1, "s_i"']
			\arrow[maps to, from=2-1, to=2-2, "j^w_n", "\text{by \eqref{eq:j^w_x(x)=y(n)}}"']
			\arrow[maps to, from=2-2, to=1-2, "s_i", "\text{by \cref{lem:s_i_on_true_and_on_emptyset}}"']
		\end{tikzcd}\]
	\end{itemize}
	
	\ref{it:proof_topologies_on_ndimsSet_1} $\Rightarrow$ \ref{it:proof_topologies_on_ndimsSet_2}.
	We show this direction by contraposition.
	Suppose there exists $k \in \N$ such that $0 \leq k < k+1 \leq n+1$ and $w_k w_{k+1} = 10$.
	We show that there exists an element $x \in \Omega(k)$ for which \eqref{eq:naturality_square_degeneracy_maps} does \textbf{not} commute, similarly to what we showed in \cref{ex:topologies_on_ReflGraph}.
	Take $x \defeq \ynhollow[k]$.
	To determine $j^w_{k+1} (s_i (\ynhollow[k]))$, we look at its incidence tuple: for $0 \leq i' \leq k$
	\begin{align*}
		&(j^w_k d_{i'} s_i)(\ynhollow) \\
		&=
		\begin{cases}
			(j^w_k s_i d_{i'-1})(\ynhollow[k]) 	&\text{if } i' > i+1 \\
			j^w_k(\ynhollow[k])					&\text{if } i' \in \set{i+1, i} \\
			(j^w_k s_{i-1} d_{i'})(\ynhollow[k])	&\text{if } i' < i
		\end{cases}
		\tag*{simplicial identities \eqref{eq:simplicial_identities}} \\
		&=
		\begin{cases}
			(j^w_k s_i (y(k-1)) 	&\text{if } i' > i+1 \\
			y(k)					&\text{if } i' \in \set{i+1, i} \\
			(j^w_k s_{i-1} (y(k-1))	&\text{if } i' < i
		\end{cases}
		\tag*{\makecell{by \cref{lem:unique_incidence_except_yn_ynhollow} and\\def.~$j^w$ with $w_k=1$}} \\
		&=
		\begin{cases}
			(j^w_k (y(k-1)) 	&\text{if } i' > i+1 \\
			y(k)				&\text{if } i' \in \set{i+1, i} \\
			(j^w_k (y(k-1))		&\text{if } i' < i
		\end{cases}
		\tag*{by \cref{lem:s_i_on_true_and_on_emptyset}} \\
		&= y(k) 
		\tag*{by definition of $j^w$}
	\end{align*}
	Thus, by definition of $j^w$ and since $w_{k+1} = 0$, we have 
	\begin{equation}
		\label{eq:j^w_k+1(s_i(ynhollow_k))=ynhollow_k+1}
		j^w_{k+1} (s_i (\ynhollow[k])) = \ynhollow[k+1].
	\end{equation}
	Hence:
	\[\begin{tikzcd}[ampersand replacement=\&, column sep=huge]
		{s_i(\ynhollow[k])} \& {\ynhollow[k+1] \neq y(k+1)} \\
		\ynhollow[k] \& {y(k)}
		\arrow[maps to, from=1-1, to=1-2, "j^w_{k+1}", "\text{by \eqref{eq:j^w_k+1(s_i(ynhollow_k))=ynhollow_k+1}}"']
		\arrow[maps to, from=2-1, to=1-1, "s_i"']
		\arrow[maps to, from=2-1, to=2-2, "j^w_k", "\text{def.~$j^w$ with $w_k=1$}"']
		\arrow[maps to, from=2-2, to=1-2, "s_i", "\text{by \cref{lem:s_i_on_true_and_on_emptyset}}"']
	\end{tikzcd}
	\qedhere
	\]
\end{proof}

\subsection{Proofs of \texorpdfstring{\cref{subsec:semi_and_general}}{Section 3.4}}

\repeattheorem{thm:topologies_on_semisSet_and_sSet}

\begin{proof}[Proofs of \cref{thm:topologies_on_semisSet_and_sSet}]
		Let $w \in \set{0,1}^\omega$ be an infinite bit string.
		We proceed with the proof for $D = \semisSet$; the case $D = \sSet$ is analogous with the additional assumption that we take $w$ of the form $0^\omega$ or $0^m1^\omega$.
		
		We define a family of mappings $j^w = (j^w_n)_{n \in \N}$ as follows: on component $n \in \N$, let $j^w_n \defeq j^{w_0 \ldots w_n}_n$. 
		By \cref{thm:topologies_on_ndimsemisSet}, $(j^w_k)_{0 \leq k \leq n} = (j^{w_0 \ldots w_n}_k)_{0 \leq k \leq n}$ is a topology on $\ndimsemisSet$, i.e., the following hold for all $n \in \N$:
		\begin{enumerate}[(a), noitemsep]
			\item 
			\label{it:proof_sSet_1}
			For all $0 \leq k < k+1 \leq n$, the naturality squares for all $\Omega\leqn\semi(d_i)$ 
			commute:
			\[
			\begin{tikzcd}
				\Omega\leqn\semi(k+1) 	& \Omega\leqn\semi(k+1) \\
				\Omega\leqn\semi(k)   	& \Omega\leqn\semi(k)
				\ar[from=1-1, to=1-2, "j^w_{k+1}"]
				\ar[from=1-1, to=2-1, "\Omega\leqn\semi(d_i)"']
				\ar[from=1-2, to=2-2, "\Omega\leqn\semi(d_i)"]
				\ar[from=2-1, to=2-2, "j^w_k"']
			\end{tikzcd}
			\]
			
			\item 
			\label{it:proof_sSet_2}
			For all $0 \leq k \leq n$, $j^w_k \cdot \true_k = \true_k$.
			
			\item 
			\label{it:proof_sSet_3}
			For all $0 \leq k \leq n$, $j^w_k \cdot j^w_k = j^w_k$.
			
			\item 
			\label{it:proof_sSet_4}
			For all $0 \leq k \leq n$, $j^w_k \cdot \land = \land \cdot (j^w_k \times j^w_k)$.
		\end{enumerate}
		By definition, $j^w$ is a topology on $\sSet$ if and only if \ref{it:proof_sSet_1}-\ref{it:proof_sSet_4} hold for all $k \in \N$.
		This is the case, because \ref{it:proof_sSet_1}-\ref{it:proof_sSet_4} hold for arbitrary $n \in \N$.

		We now prove that all topologies on $\semisSet$ are of the form $j^w$ for some infinite bit string $w \in \set{0,1}^\omega$.
		Take a topology $j = (j_k)_{k \in \N}$ on $\semisSet$.
		We inductively define $w \in \set{0,1}^\omega$, such that, for all $n \in \N$, the topologies $(j_k)_{0 \leq k \leq n}= j^{w_0 \ldots w_n}$ are equal.
		\begin{itemize}
			\item 
			For $n=0$, $j_0$ is a topology on $\cat{D}\leqn[0] = \Set$ (\cref{ex:topologies_on_set}).
			If $j_0 = j^0$, the discrete topology on $\Set$, let $w_0 \defeq 0$, and if $j_0 = j^1$, the trivial topology on $\Set$, let $w_0 \defeq 1$.
			
			\item 
			Suppose true for $n$, i.e., $w_0 \ldots w_n$ has been defined and $(j_k)_{0 \leq k \leq n}=j^{w_0 \ldots w_n}$.
			Notice that $(j_k)_{0 \leq k \leq n+1}$ is a topology on $\ndimsemisSet[n+1]$ that extends the topology $(j_k)_{0 \leq k \leq n}$ on $\ndimsemisSet$.
			By \cref{thm:topologies_on_ndimsemisSet}, the topology $(j_k)_{k \leq n+1}$ must be of the form $j^{w_0 \ldots w_n w_{n+1}}$ for some bit $w_{n+1} \in \set{0,1}$. 
		\end{itemize}
		The natural transformations $j$ and $j^w$ agree on each component and are therefore equal, which concludes the proof.
\end{proof}

\subsection{Proofs of \texorpdfstring{\cref{subsec:separated_elements_and_sheaves}}{Section 3.5}}

\repeatlemma{lem:closure_of_topologies_on_sSet}

\begin{proof}[Proof of \cref{lem:closure_of_topologies_on_sSet}]
    Recall the definition of characteristic functions (\cref{def:subobject_classifier}).
	In particular, for all $k \in \N$:
	\begin{equation}
		\label{eq:charac_function_recall}
		a \in A'(k) \iff \charac{A'}_k(a) = \true_k(\cdot) = y(k)
		\qquad\qquad
		\begin{tikzcd}[column sep=5mm, row sep=5mm,ampersand replacement=\&]
			A' \& 1 \\
			A \& \Omega
			\ar[from=1-1, to=1-2, "!"]
			\ar[from=1-1, to=2-1, tail]
			\ar[from=1-2, to=2-2, "\true"]
			\ar[from=2-1, to=2-2, "\charac{A'}"']
			\ar[from=1-1, to=2-2, "\lrcorner"{anchor=center, pos=0.125}, draw=none]
		\end{tikzcd}
	\end{equation}
	Recall also that $\charac{\topologybar{A'}} = j^w \mycirc \charac{A'}$ by \cref{lem:equivalence_LT-topologies_and_closure_operators}, i.e., the pullback squares of the subobjects $A' \subseteq A$ and $\topologybar{A'} \subseteq$ can be visualised as follows~\cite[p.~220]{MacLane_Moerdijk_1994}:
	\begin{equation}
		\label{eq:charac_function_closure_recall}
		\begin{tikzcd}[ampersand replacement=\&]
		{\topologybar{A'}} \&\&\& 1 \\
		\& {A'} \& 1 \\
		A 
		\arrow[drrru, to path= { 
			-- ([yshift=-2ex]\tikztostart.south) 
			|- ([yshift=-2ex]\tikztotarget.south) node[near end, below, font=\scriptsize] {${\charac{\topologybar{A'}}}$}
			-- (\tikztotarget)}]
		\& A \& \Omega \& \Omega
		\arrow[from=1-1, to=1-4]
		\arrow[tail, from=1-1, to=3-1]
		\arrow["\lrcorner"{anchor=center, pos=0.125}, draw=none, from=1-1, to=3-4]
		\arrow["\true", from=1-4, to=3-4]
		\arrow[from=2-2, to=2-3]
		\arrow[tail, from=2-2, to=3-2]
		\arrow["\lrcorner"{anchor=center, pos=0.125}, draw=none, from=2-2, to=3-3]
		\arrow[from=2-3, to=1-4]
		\arrow["\true", from=2-3, to=3-3]
		\arrow[equal, from=3-1, to=3-2]
		\arrow["{\charac{A'}}", from=3-2, to=3-3]
		\arrow["j", from=3-3, to=3-4]
	\end{tikzcd}
	\end{equation}
	We prove the lemma by a (strong) induction on $k$.
	For $k=0$, we have two cases
	\begin{itemize}
		\item 
		Suppose $w_0 = 0$, which means that $j^w$ extends the topology $j^0$ on $\Set$, i.e,. $j^w_0 = \id_{\Omega(0)}\colon \Omega(0) \to \Omega(0)$.
		Recall that $\Omega(0) = \set{\emptyset, y(0)=\id_0}$.
		We show that $\topologybar{A'}(0) = A'(0)$: given a vertex $v \in A(0)$, then
		\begin{align*}
			v \in \topologybar{A'}(0)
			&\iff j^w_0 \charac{A'}_0 (v) = \id_0 = y(0) \in \Omega(0)
			\tag*{by \eqref{eq:charac_function_recall}, \eqref{eq:charac_function_closure_recall}} \\
			&\iff \charac{A'}_0 (v) = \id_0
			\tag*{def.~$j^w_0$} \\
			&\iff v \in A'(0) 
			\tag*{by \eqref{eq:charac_function_recall}}
		\end{align*}
		
		\item 
		Suppose $w_0 = 1$, which means that $j^w$ extends the topology $j^1$ on $\Set$, i.e,. $j^w_0 \colon \Omega(0) \to \Omega(0) \colon \emptyset, \id_0 \mapsto \id_0$.
		We show that $\topologybar{A'}(0) = A(0)$: given a vertex $v \in A(0)$, then
		\begin{align*}
			v \in \topologybar{A'}(0)
			&\iff j^w_0 \charac{A'}_0 (v) = \id_0 = y(0) \in \Omega(0)
			\tag*{by \eqref{eq:charac_function_recall}, \eqref{eq:charac_function_closure_recall}} \\
			&\iff \charac{A'}_0 (v) \in \set{\emptyset, \id_0}
			\tag*{def.~$j^w_0$} \\
			&\iff v \in A(0) 
			\tag*{by \eqref{eq:charac_function_recall}}
		\end{align*}
	\end{itemize}

	For the induction step, take $k \geq 1$.
	Suppose true for all $m < k$, and let us prove the statement of the lemma for $k$.
	We again distinguish cases:
	\begin{itemize}
		\item 
		Suppose $w_k = 0$.
		Then, by \cref{def:topologies_on_ndimsemisSet} $j^w$ sends only $y(k)$ to $y(k)$.
		Thus for $x \in A(k)$, we have
		\begin{align*}
			x \in \topologybar{A'}(k)
			&\iff j^w_k \charac{A'}_k (x) = \true_k(\cdot) = y(k) \in \Omega(k)
			\tag*{by \eqref{eq:charac_function_recall}, \eqref{eq:charac_function_closure_recall}} \\
			&\iff \charac{A'}_k (x) = y(k)
			\tag*{def.~$j^w_k$ with $w_k=0$} \\
			&\iff x \in A'(k) 
			\tag*{by \eqref{eq:charac_function_recall}}
		\end{align*}
		
		\item 
		Suppose $w_k = 1$.
		There are two scenarios that we consider, depending on whether all bits $w_m$ with $m < k$ are also $1$ or not.
		
		The first scenario is when $w \in \set{1}^{k+1}$.
		We want to show that $j^w$ is the trivial topology, i.e., that the closure associated with $j^w$ replaces every subobject by its main object.
		In other words, we make the intermediate claim that \eqref{eq:proof_closure_topologies_sSet_intermediary_claim} holds for all $m$ such that $0 \leq m < k$, and we prove it by induction on $m$:
		\begin{equation}
			\label{eq:proof_closure_topologies_sSet_intermediary_claim}
			\topologybar{A'}(m) = A(m)
		\end{equation}
		\begin{itemize}[leftmargin=8mm]
			\item
			For $m=0$: $\topologybar{A'}(0) = A(0)$ was proven in the base case of the main induction.
			
			\item 
			Suppose \eqref{eq:proof_closure_topologies_sSet_intermediary_claim} holds for $m$.
			We prove it for $m+1 < k$.
			Note that because $m+1 < k$, the main IH \eqref{eq:IH_closure_j^w} applies to $m+1$:
			\begin{align*}
				&\topologybar{A'}(m+1) \hspace*{120mm} \\
				&= \setvbar{x \in A(m+1)}{\forall 0 \leq i \leq m+1 \colon A(d_i^{m+1})(x) \in \topologybar{A'}(m)} 
				\tag*{main IH \eqref{eq:IH_closure_j^w}} \\
				&= \setvbar{x \in A(m+1)}{\smash{\underbrace{\forall 0 \leq i \leq m+1 \colon A(d_i^{m+1})(x) \in A(m)}_{\mathclap{\text{vacuous statement}}}}}
				\tag*{IH claim \eqref{eq:proof_closure_topologies_sSet_intermediary_claim}} \\
				&= A(m+1).
			\end{align*}
		\end{itemize}
		We have proven the intermediate claim \eqref{eq:proof_closure_topologies_sSet_intermediary_claim} and can go back to proving \eqref{eq:IH_closure_j^w} for $k$:
		\begin{align*}
			\topologybar{A'}(k)
			&= \set{x \in \topologybar{A'}(k)} 
			\\
			&= \setvbar{x \in \topologybar{A'}(k)}{\smash{
				\overbrace{\forall 0 \leq i \leq k \colon A(d_i)(x) \in A(k-1)}^{\mathclap{\text{vacuous statement}}}
			}}
			\\
			&= \setvbar{x \in \topologybar{A'}(k)}{\forall 0 \leq i \leq k \colon A(d_i)(x) \in \topologybar{A'}(k-1)}
			\tag*{by \eqref{eq:proof_closure_topologies_sSet_intermediary_claim}}
		\end{align*}
		as desired.
		
		In the second scenario to consider, $j^w$ is not the trivial topology and there exists $m \geq 1$ such that $w_{k-m} = 0$.
		Suppose w.l.o.g.~that $m \geq 1$ is the smallest natural number such that $w_{k-m} = 0$; in other words, $w_l = 1$ for all $k-m+1 \leq l \leq k$.
		In the next part of the proof, we frequently consider the morphism $d_{i_{k-m+1}, \ldots, i_k}$ in the index category, defined below for arbitrary indices $0 \leq i_k \leq k$, $\ldots$, $0\leq i_{k-m+1} \leq k-m+1$.
		This morphism is defined as a composition of $m$ face maps:
		\[
			d_{i_{k-m+1},\ldots,i_k}
			\defeq
			\qquad
			\big(
			\begin{tikzcd}[ampersand replacement=\&]
				{k} \& {k-1} \& \cdots \& {k-m}
				\arrow["{d^k_{i_k}}", from=1-1, to=1-2]
				\arrow[from=1-2, to=1-3]
				\arrow["{d^{k-m+1}_{i_{k-m+1}}}", from=1-3, to=1-4]
			\end{tikzcd}
			\big)
		\]
		We now prove \eqref{eq:IH_closure_j^w} for $k$.
		Take $x \in A(k)$.
		We use the composition $d_{i_{k-m+1},\ldots,i_k}$ of face maps to go to the dimension $k-m$ in order to use the IH:
		\begin{align*}
			&x \in \topologybar{A'}(k) \\
			&\iff j^w_k \charac{A'}_k (x) = \true_k(\cdot) = y(k) \in \Omega(k)
			\hspace*{15mm}
			\tag*{by \eqref{eq:charac_function_recall}, \eqref{eq:charac_function_closure_recall}} \\
			&\iff
			\forall i_k \colon (j^w_{k-1} \Omega(d^k_{i_k}) \charac{A'}_k) (x) = y(k-1)
			\tag*{$(*)$}
		\intertext{
		The reason for $(*)$ is due to the definition of $j^w_k$ with $w_k = 1$, as described in \eqref{eq:j_n+1_on_x_<_ynhollow} of \cref{lem:topologies_on_ndimsemisSet_unique_extension}.
		It says that an element of $\Omega(k)$ is sent by $j^w_k$ to $y(k)$ if and only if all its faces are sent by $j^w_{k-1}$ to $y(k-1)$. 
		We repeat the same reasoning $m-1$ more times and continue the equivalence:}
			&\iff
			\forall i_{k-m+1} \ldots i_k \colon (j^w_{k-m} \Omega(d_{i_{k-m+1},\ldots,i_k}) \charac{A'}_k) (x) = y(k-m)
			\tag*{$(*)$}
			\\
			&\iff \forall \dittotikz \colon (\Omega(d_{i_{k-m+1},\ldots,i_k}) \charac{A'}_k) (x) = y(k-m)
			\tag*{\eqref{eq:topologies_on_ndimsSet_on_yn}, $w_{k-m}=0$}
			\\
			&\iff \forall \dittotikz \colon (\charac{A'}_{k-m} A(d_{i_{k-m+1},\ldots,i_k})) (x) = y(k-m)
			\tag*{$\charac{A'}$ naturality} \\
			&\iff \forall \dittotikz \colon  A(d_{i_{k-m+1},\ldots,i_k}) (x) \in A'(k-m)
			\tag*{by \eqref{eq:charac_function_recall}} \\
			&\iff \forall \dittotikz \colon  A(d_{i_{k-m+1},\ldots,i_k}) (x) \in \topologybar{A'}(k-m)
			\tag*{IH $w_{k-m}=0$ \eqref{eq:IH_closure_j^w}} \\
			&\iff \forall i_k \colon A(d_{i_k}) (x) \in \topologybar{A'}(k-1).
			\tag*{IH \eqref{eq:IH_closure_j^w} $(m-1)$ times}
		\end{align*}
      This concludes the induction step and thus the proof.
      \qedhere
	\end{itemize}
\end{proof}

\repeatlemma{lem:j^w_dense_subpresheaf}

\begin{proof}[Proof of \cref{lem:j^w_dense_subpresheaf}]
    We prove the lemma by induction on $n \in \N$.
	\begin{itemize}
		\item 
		For $n=0$,
		\[
		\begin{array}{rcl}
			\topologybar{A'}(0) = A(0)
			&\stackrel{\text{\cref{lem:closure_of_topologies_on_sSet}}}{\iff}&
			\begin{cases}
				\text{If } w_0=0:~ A'(0) = A(0) \\
				\text{If } w_0=1:~ A (0) = A(0)
			\end{cases}
			\\
			&\iff&
			\text{If } w_0=0: A'(0) = A(0).
		\end{array}
		\]
		
		\item 
		Suppose the lemma holds for $n$ and let us prove it for $n+1$.
		\begin{itemize}
			\item[($1 \Rightarrow 2$)]
			Suppose $\topologybar{A'}(n+1) = A(n+1)$ and $w_{n+1} = 0$.
			By \cref{lem:closure_of_topologies_on_sSet}, $\topologybar{A'}(n+1) = A'(n+1)$.
			Hence, $A'(n+1) = A(n+1)$, as desired.
			
			\item[(2 $\Rightarrow$ 1)]
			Suppose that for all $k \leq n+1$, if $w_k = 0$, then $A'(k) = A(k)$.
			We want to prove that for all $k \leq n+1$: $\topologybar{A'}(k) = A(k)$.
			By IH, this holds already for all $k \leq n$, hence what needs to be proven is $\topologybar{A'}(n+1) = A(n+1)$.
			
			Consider the case where $w_{n+1} = 0$.
			We have the desired equality:
			\[
			\begin{array}{rcl}
				\topologybar{A'}(n+1)
				\stackrel{\text{\cref{lem:closure_of_topologies_on_sSet}}}{=} 
				A'(n+1)
				\stackrel{\text{Assumption}}{=}
				A(n+1).
			\end{array}
			\]
			
			Consider the case where $w_{n+1} = 1$.
			Take $x \in A(n+1)$.
			\[
			\begin{array}{rcl}
				&\topologybar{A'}(n+1) \\
				&\stackrel{\text{\cref{lem:closure_of_topologies_on_sSet}}}{=}&
				\setvbar{x \in A(n+1)}{\forall i=0, \ldots, n+1 \colon A(d_i)(x) \in \topologybar{A'}(n)}
				\\
				&\stackrel{\text{IH}}{=}&
				\setvbar{x \in A(n+1)}{\smash{\underbrace{\forall i=0, \ldots, n+1 \colon A(d_i)(x) \in A(n)}_{\text{vacuous statement}}}}
				\\
				&=&
				A(n+1) 
			\end{array}
			\]
		\end{itemize}
	\end{itemize}
\end{proof}

\repeattheorem{thm:separated_elements_sheaves_ndimsemisSet}

\begin{proof}[Proof of \cref{thm:separated_elements_sheaves_ndimsemisSet}]
    Observe that Items 1 and 2 imply Item 3
	Take $B \in \cat{D}$.
	We first prove Item 1.
	Recall from \cref{lem:j^w_dense_subpresheaf} that a subpresheaf $A' \subseteq A$ is $j^w$-dense if $A'(k) = A(k)$ for every $k$ such that $w_k = 0$.

	\begin{itemize}
		\item[($\Rightarrow$)]
		Suppose $B$ is $j^w$-separated.
		Take $k \in \N$ with $w_k = 1$.
		We prove that $B$ is $k$-simple.
		Consider the subpresheaf $\ynhollow[k] \subseteq y(k)$.
		By \cref{lem:j^w_dense_subpresheaf}, $\ynhollow[k]$ is $j^w$-dense, because it is different from $y(k)$ only on component $k$, $y(k)(k) \setminus \ynhollow[k](k) = \set{\id_k}$, and $w_k=1$.
		We distinguish cases based on $k$.
		
		\begin{itemize}
			\item 
			Suppose $k=0$.
			Take two vertices $x, x' \in B(0)$.
			To prove that $B$ is $0$-simple, we show that $x$ and $x'$ must necessarily be equal.
			For $k=0$, the subpresheaf $\ynhollow[0] \subseteq y(0)$ is the inclusion of the empty graph into the graph with one vertex: $\emptyset \subseteq \set{\cdot}$.
			Both $x,x' \in B(0)$ define a factorisation of $\emptyset \colon \emptyset \to B$ through $\emptyset \subseteq \set{\cdot}$.
			Because $B$ is $j^w$-separated, there is at most one such factorisation, hence $x=x'$, i.e., $B$ is $0$-simple.
			\[
				\begin{tikzcd}[row sep=scriptsize]
					\emptyset 
					\ar[dr, "\emptyset"]
					\ar[d, "\subseteq" {description, rotate=270}, phantom] \\
					\set{\cdot}
					\ar[r, dotted, shift left=1, "x"]
					\ar[r, dotted, shift right=1, "x'"']
					& B
				\end{tikzcd}
			\]
			
			\item 
			Suppose $1 \leq k$.
			Take two parallel $k$-faces $x,x' \in B(k)$, i.e.~$x,x' \in d\inv(\vec{z})$ for some $z \in B(k-1)^{k+1}$ (cf.~Notation~\ref{not:parallel_k_faces}).
			To prove that $B$ is $k$-simple, we show that $x$ and $x'$ must necessarily be equal.
			We define three $D$-morphisms $f \colon \ynhollow[k] \to B$, $g \colon y(k) \to B$, and $g' \colon y(k) \to B$ as follows.
			For $f$, it suffices to define its component $f_{k-1}$, and for $g$ and $g'$, their components $g_k$ and $g'_k$, as the remaining components are then uniquely determined in order for the naturality squares with face maps to commute.
			For $0 \leq i \leq k$, let
			\[
			\begin{array}{rclrcl}
				f_{k-1} \colon \ynhollow[k](k-1)	&\to& B(k-1) \colon
				& \opcat{(d^k_i)}					&\mapsto& z_i. \\
				g_k \colon y(k)(k)	&\to& B(k) \colon
				& \id_k				&\mapsto& x. \\
				g'_k \colon y(k)(k)	&\to& B(k) \colon
				& \id_k				&\mapsto& x'. 
			\end{array}
			\]
			\[
			\begin{tikzcd}[row sep=scriptsize]
				\ynhollow[k]
				\ar[dr, "f"]
				\ar[d, "\subseteq" {description, rotate=270}, phantom] \\
				y(k)
				\ar[r, dotted, shift left=1, "g"]
				\ar[r, dotted, shift right=1, "g'"']
				& B
			\end{tikzcd}
			\]
			Both $g$ and $g'$ are factorisations of $f$ through $\ynhollow[k] \subseteq y(k)$: for $0 \leq i \leq k$ and $\opcat{(d^k_i)} \in \ynhollow[k](k-1)$,
			\begin{align*}
				g_{k-1}(\opcat{(d^k_i)}) 
				&= (g_{k-1} \cdot y(k)(d_i)) (\id_k)
				\tag*{def.~$y$ \cref{def:yoneda_embedding}} \\
				&= (B(d_i) \cdot g_k) (\id_k)
				\tag*{naturality $g$} \\
				&= B(d_i) (x) 
				\tag*{definition $g$} \\
				&= z_i
				\tag*{definition $x \in d\inv(\vec{z})$} \\
				&= f_{k-1} (\opcat{(d^k_i)}),
				\tag*{definition $f$}
			\end{align*}
			Because $B$ is $j^w$-separated, the factorisations are equal: $g=g'$.
			Hence, $x=x'$ in $B(k)$, i.e., $B$ is $k$-simple.
		\end{itemize}

		\item[$(\Leftarrow)$]
		Suppose $B$ is $k$-simple whenever $w_k = 1$.
		Take a $j^w$-dense subobject $A' \subseteq A$, a morphism $f \colon A' \to B$, and two factorisations $g,g' \colon A \to B$ of $f$ through $A$.
		To demonstrate $g=g'$, we prove $g_k = g'_k$ by induction on $k$.
		\begin{itemize}
			\item
			For $k=0$, let us distinguish cases based on $w_0$.
			\begin{itemize}
				\item 
				If $w_0 = 0$, then $A'(0) = A(0)$, and thus $g_0 = f_0 = g'_0$.
				
				\item 
				If $w_0 = 1$, then we know $B$ is $0$-simple.
				In case $A(0) = \emptyset$, then $g_0 = \emptyset = g'_0 \colon \emptyset \to B(0)$.
				Otherwise, take $x \in A(0)$.
				By $0$-simplicity, $g_0 (x) = g'_0(x)$ in $B(0)$.
			\end{itemize}
			
			\item 
			Suppose true until $k$ and let us prove $g_{k+1} = g'_{k+1}$.
			We distinguish cases based on $w_{k+1}$.
			\begin{itemize}
				\item 
				If $w_{k+1} = 0$, then $A'(k+1) = A(k+1)$, and thus $g_{k+1} = f_{k+1} = g'_{k+1}$.
				
				\item 
				If $w_{k+1} = 1$, then we know $B$ is $(k+1)$-simple.
				If $A(k+1) = \emptyset$, then $g_{k+1} = \emptyset = g'_{k+1} \colon \emptyset \to B(0)$.
				Otherwise, take $x \in A(k+1)$.
				Observe that $g_{k+1}(x), g'_{k+1}(x) \in B(k+1)$ share the same faces: for $0 \leq i \leq k+1$,
				\begin{align*}
					B(d^{k+1}_i) (g_{k+1} (x)) 
					&= g_k (A(d^{k+1}_i) (x))
					\tag*{$g$ naturality} \\
					&= g'_k (A(d^{k+1}_i) (x))
					\tag*{IH} \\
					&= B(d^{k+1}_i) (g'_{k+1} (x)).
					\tag*{$g$ naturality}
				\end{align*}
				Hence, if we define $z_i \defeq  B(d^{k+1}_i) (g_{k+1} (x))$ for $0 \leq i \leq k+1$, then
				\[
					g_{k+1}(x), g'_{k+1}(x) \in d\inv(\vec{z}).
				\]
				Because $B$ is $(k+1)$-simple, there is at most one element in the set $d\inv(\vec{z})$, hence $g_{k+1}(x) = g'_{k+1}(x)$.
			\end{itemize}
		\end{itemize}
	\end{itemize}
	We now prove Item 2.
	\begin{itemize}
		\item[($\Rightarrow$)]
		Suppose $B$ is $j^w$-complete.
		Take $k \in \N$ with $w_k = 1$.
		We prove that $B$ is $k$-complete.
		Consider the subpresheaf $\ynhollow[k] \subseteq y(k)$.
		By \cref{lem:j^w_dense_subpresheaf}, $\ynhollow[k]$ is $j^w$-dense, because it is different from $y(k)$ only on component $k$, $y(k)(k) \setminus \ynhollow[k](k) = \set{\id_k}$, and $w_k=1$.
		We distinguish cases based on $k$.
		
		\smallskip
		\begin{minipage}{.74\linewidth}
			\begin{itemize}
				\item 
				Suppose $k=0$.
				To prove that $B$ is $0$-complete, we show that $B(0)$ has at least one element.
				For $k=0$, the subpresheaf $\ynhollow[0] \subseteq y(0)$ is the inclusion of the empty graph into the graph with one vertex: $\emptyset \subseteq \set{\cdot}$.
				Because $B$ is $j^w$-complete, there exists a factorisation $g$ of $\emptyset \colon \emptyset \to B$ through $\emptyset \subseteq \set{\cdot}$.
				Hence, $g_0(\cdot) \in B(0)$ and $B(0)$ is nonempty, as desired.
			\end{itemize}
		\end{minipage}
		\hfill
		\begin{minipage}{.2\linewidth}
			\vspace*{-3mm}
			\begin{tikzcd}[row sep=scriptsize]
				\emptyset 
				\ar[dr, "\emptyset"]
				\ar[d, "\subseteq" {description, rotate=270}, phantom] \\
				\set{\cdot}
				\ar[r, dotted, "g"'] 
				& B
			\end{tikzcd}
		\end{minipage}   
		\begin{itemize}
			\item 
			Suppose $1 \leq k$.
			To prove that $B$ is $k$-complete, we assume the existence of a $k+1$-tuple $\vec{x} = (x_0, \ldots, x_k) \in B(k-1)^{k+1}$ and show that the set $d\inv(\vec{x})$ is nonempty.
			We define a $\cat{D}$-morphism $f\colon\ynhollow[k] \to B$ as follows.
			It suffices to define its component $f_{k-1}$, as the remaining components are then uniquely determined in order for the naturality squares with face maps to commute.
			\[
			\begin{array}{rcl}
				f_{k-1} \colon \ynhollow[k](k-1) &\to& B(k-1) \\
				\text{for all $0 \leq i \leq k$:}\quad \opcat{(d^k_i)} &\mapsto& x_i.
			\end{array}
			\]
			Because $B$ is $j^w$-complete, there exists at least one factorisation $g$ of $f$ through $\ynhollow[k] \subseteq y(k)$.
			\begin{equation}
				\label{eq:g_factorisation_of_f}
			\begin{tikzcd}[row sep=scriptsize]
				\ynhollow[k]
				\ar[dr, "f"]
				\ar[d, "\subseteq" {description, rotate=270}, phantom] \\
				y(k)
				\ar[r, dotted, "\exists g"']
				& B
			\end{tikzcd}
			\end{equation}
			We check that $g_k(\id_k) \in B(k)$ belongs to the set $d\inv(\vec{x})$:
			for $0 \leq i \leq k$,
			\begin{align*}
				B(d^k_i) (g_k(\id_k))
				&= g_{k-1} (y(k)(d^k_i)(\id_k))
				\tag*{$g$ naturality} \\
				&= g_{k-1} (\opcat{(d^k_i)})
				\tag*{def.~$y(k)(d^k_i)$ \cref{def:yoneda_embedding}} \\
				&= f_{k-1} (\opcat{(d^k_i)})
				\tag*{by \eqref{eq:g_factorisation_of_f}} \\
				&= x_i
				\tag*{definition $f$}
			\end{align*}
			Therefore, the set $d\inv(\vec{x})$ is nonempty, which proves that $B$ is $k$-complete.
		\end{itemize}

		\item[$(\Leftarrow)$]
		Suppose $B$ is $k$-complete whenever $w_k = 1$.
		Take a $j^w$-dense subobject $A' \subseteq A$, a $\cat{D}$-morphism $f \colon A' \to B$.
		We want to prove existence of a factorisation $g \colon A \to B$ of $f$ through $A' \subseteq A$.
		We demonstrate that such a $g$ exists, by constructing inductively $g_k$ for every $k \in \N$.
		\begin{itemize}
			\item
			For $k=0$, let us distinguish cases based on $w_0$.
			\begin{itemize}
				\item 
				If $w_0 = 0$, then $A'(0) = A(0)$ by \cref{lem:j^w_dense_subpresheaf}.
				Hence, define
				\[
					g_0 \defeq f_0 \colon A(0) = A'(0) \xrightarrow{} B(0).
				\]
				
				\item 
				If $w_0 = 1$, then by assumption $B$ is $0$-complete, i.e., has at least one vertex.
				In case $A(0) = \emptyset$, let $g_0 \defeq \emptyset \colon \emptyset \to B(0)$.
				Otherwise, given $x \in A(0)$, define $g_0 (x)$ to be any vertex in $B(0)$.
			\end{itemize}
			
			\item 
			Suppose $g_0, \ldots, g_k$ are constructed so that the naturality square \eqref{eq:proof_j^w_complete_1} commutes when $D \in \set{\ndimsemisSet,\semisSet}$ and both \eqref{eq:proof_j^w_complete_1} and \eqref{eq:proof_j^w_complete_2} commute when $D \in \set{\ndimsSet, \sSet}$, for all $0 \leq l < l+1 \leq k$, $0 \leq i \leq l+1$ and $0 \leq i' \leq l$.
			\\ 
			\noindent\begin{tabularx}{.8\textwidth}{@{}XX@{}}
				\begin{equation}
					\label{eq:proof_j^w_complete_1}
					\begin{tikzcd}[ampersand replacement=\&, scale cd=.9]
						{A(l+1)} \& {B(l+1)} \\
						{A(l)} \& {B(l)}
						\arrow["{g_{l+1}}", from=1-1, to=1-2]
						\arrow["{A(d^{l+1}_i)}", from=1-1, to=2-1]
						\arrow["{B(d^{l+1}_i)}"', from=1-2, to=2-2]
						\arrow["{g_l}"', from=2-1, to=2-2]
					\end{tikzcd}
				\end{equation} &
				\begin{equation}
					\label{eq:proof_j^w_complete_2}
					\begin{tikzcd}[ampersand replacement=\&, scale cd=.9]
						{A(l+1)} \& {B(l+1)} \\
						{A(l)} \& {B(l)}
						\arrow["{g_{l+1}}", from=1-1, to=1-2]
						\arrow["{A(s^l_{i'})}"', from=2-1, to=1-1]
						\arrow["{g_l}"', from=2-1, to=2-2]
						\arrow["{B(s^l_{i'})}", from=2-2, to=1-2]
					\end{tikzcd}
				\end{equation}
			\end{tabularx}\\
			We construct $g_{k+1} \colon A(k+1) \to B(k+1)$ so that naturality still holds, i.e., for $l=k$ either \eqref{eq:proof_j^w_complete_1} commutes or both \eqref{eq:proof_j^w_complete_1} and \eqref{eq:proof_j^w_complete_2} commute, depending on the category $\cat{D}$.
			We distinguish cases based on $w_{k+1}$.
			\begin{itemize}
				\item 
				If $w_{k+1} = 0$, then $A'(k+1) = A(k+1)$ by \cref{lem:j^w_dense_subpresheaf}.
				It suffices to take $g_{k+1} \defeq f_{k+1}$, which satisfies all required naturality squares because $f$ is a $\cat{D}$-morphism.
				
				\item 
				If $w_{k+1} = 1$, then $B$ is $(k+1)$-complete by assumption.
				In case $A(k+1) = \emptyset$, then $g_{k+1} \defeq \emptyset \colon \emptyset \to B(0)$.
				Otherwise, take an arbitrary $x \in A(k+1)$.
				Take all its faces $A(d^{k+1}_{k+1})(x), \ldots, A(d^{k+1}_{0})(x) \in A(k)$.
				By applying $g_k$, we obtain a $(k+2)$-tuple 
				\[
					\vec{z} \defeq (g_k( A(d^{k+1}_{k+1})(x), \ldots, g_k(A(d^{k+1}_{0})(x)) \in B(k)^{k+2}.
				\] 
				By $k$-completeness,
				the set $d\inv(\vec{z})$ is nonempty.
				We distinguish cases based on $\cat{D}$.
				
				Suppose we are in the case where $\cat{D} \in \set{\ndimsemisSet, \semisSet}$, take an arbitrary $x' \in d\inv(\vec{z})$.
				Let $g_{k+1}(x) \defeq x'$.
				This defines a function $g_{k+1} \colon A(k+1) \to B(k+1)$ which, by construction, makes \eqref{eq:proof_j^w_complete_1} commute for $l=k$.
				
				Now suppose we are in the case where $\cat{D} \in \set{\ndimsSet, \sSet}$.
				If $x$ is not a degenerate $k$-simplex, proceed as described just above, by taking an arbitrary $x' \in d\inv(\vec{z})$.
				In the case where $x = A(s^k_{i'})(x'')$ is a degenerate $k$-simplex, for some $0 \leq i' \leq k$ and $x'' \in A(k-1)$, we claim that $B(s^k_{i'})(g_k(x''))$ is in the set $d\inv(\vec{z})$.
				If that holds, then we can define $g_{k+1}(x) \defeq B(s^k_{i'})(g_k(x''))$ and obtain a function $g_{k+1} \colon A(k+1) \to B(k+1)$ which, by construction, makes \eqref{eq:proof_j^w_complete_1} and \eqref{eq:proof_j^w_complete_2} commute for $l=k$.
				We check our claim: for $0 \leq i \leq k+1$,
				\begin{align*}
					&B(d^{k+1}_i (B(s^k_{i'})  g_k(x'') )) 
					\hspace*{80mm}
					\\
					&=
					\begin{cases}
						(B(s^{k-1}_{i'} d^k_{i-1}) g_k) (x'')
						& \text{ if } i > i'+1\\
						g_k (x'')
						& \text{ if } i \in \set{i'+1, i'} \\
						(B(s^{k-1}_{i'-1} d^k_{i}) g_k) (x'')
						& \text{ if } i > i'+1
					\end{cases}
					\tag*{\makecell[r]{simplicial\\identities \eqref{eq:simplicial_identities}}} \\
					&=
					\begin{cases}
						(B(s^{k-1}_{i'}) g_{k-1} A(d^k_{i-1})) (x'')
						& \text{ if } i > i'+1\\
						g_k (x'')
						& \text{ if } i \in \set{i'+1, i'} \\
						(B(s^{k-1}_{i'-1}) g_{k-1} A(d^k_{i})) (x'')
						& \text{ if } i > i'+1
					\end{cases}
					\tag*{\makecell[r]{IH: $g_k, g_{k-1}$\\commute\\with $d$'s}} \\
					&=
					\begin{cases}
						(g_{k} A(s^{k-1}_{i'} d^k_{i-1})) (x'')
						& \text{ if } i > i'+1\\
						g_k (x'')
						& \text{ if } i \in \set{i'+1, i'} \\
						(g_{k} A(s^{k-1}_{i'-1} d^k_{i})) (x'')
						& \text{ if } i > i'+1
					\end{cases}
					\tag*{\makecell[r]{IH: $g_k, g_{k-1}$\\commute with $s$'s}} \\
					&= (g_k A( d^{k+1}_i s^k_{i'} )) (x'')
					\tag*{simplicial identities \eqref{eq:simplicial_identities}}
				\end{align*}
				which is the $(k+1-i)$\tss{th} component of $\vec{z}$, as desired.
				This concludes the proof.
				\qedhere
			\end{itemize}
		\end{itemize}
	\end{itemize}
\end{proof}

\subsection{Proofs of \texorpdfstring{\cref{sec:bicolour}}{Section 4}}

\repeatlemma{lem:8_topologies_in_BiColGraph}

\begin{proof}[Proof of \cref{lem:8_topologies_in_BiColGraph}]
    In $\BiColGraph$, the terminal object $1$ is 
    \smash{\(
        \begin{tikzcd}
            \cdot 
                \ar[loop, distance=.8em, color=fcolor, in=180+15, out=180-15]
                \ar[loop, distance=.8em, color=scolor, dashed , in=0+15, out=0-15]
        \end{tikzcd}
    \)},
    and the image of $\true \colon 1 \to \Omega$ is 
    \smash{\(
        \begin{tikzcd}
            1
            \ar[loop, "{\stot}"', distance=.8em, color=fcolor, in=180+15, out=180-15]
            \ar[loop, "{\stotp}"', distance=.8em, color=scolor, dashed , in=0+15, out=0-15]
        \end{tikzcd}
    \)}.
    Because of \cref{def:topology_on_topos} \ref{it:topology_true}, a topology $j \colon \Omega \to \Omega$ on $\BiColGraph$ must leave the image of $\true$ untouched, i.e., $j_V$ sends the vertex $1$ to itself, $j_E$ sends the edge $s \to t$ to itself, and $j_{E'}$ sends the edge $s' \to t'$ to itself.
	For the vertex $0 \in \Omega(V)$, there are two choices.
	
	The first choice is $j_V(0) = 0$.
	As in \cref{lem:unique_incidence_except_yn_ynhollow}, only $j_E((s,t))$ and $j_{E'}((s', t'))$ have two possible choices: respectively either $(s',t')$ or $s \to t$, and either $(s', t')$ or $s' \to t'$.
	That gives us the following four topologies:
	\begin{table}[ht]
		\centering
		\begin{tabular}{c|cc}
			$j$ 		& $j_E((s,t))$	& $j_{E'}((s',t'))$	\\ \hline
			$j^{00}$  	& $(s,t)$  		& $(s',t')$         \\
			$j^{01}$  	& $s \to t$     & $(s',t')$         \\
			$j^{02}$  	& $(s,t)$       & $s' \to t'$       \\
			$j^{03}$	& $s \to t$		& $s' \to t'$
		\end{tabular}
	\end{table}

	The second choice is $j_V(0) = 1$.
	To respect source and target, $j_E(0)$ must be a loop on the vertex $1$, i.e., it can be either $(s,t)$ or $s \to t$ in $\Omega(E)$.
	We observe that $j_E(0)$ determines the mapping of all edges in $\Omega(E)$:
	\begin{itemize}
		\item 
		Suppose $j_E(0) = (s,t)$.
		By idempotence \ref{it:topology_idempotent} of topologies: $j_E((s,t)) = (s,t)$.
		By monotonicity \ref{it:topology_monotone} of topologies and the fact that $0 \leq s, t \leq (s,t)$ in $\Omega(E)$: $j_E(s) = j_E(t) = (s,t)$.
		
		\item 
		Suppose $j_E(0) = s \to t$, which is the greatest element.
		By monotonicity, all other edges are also mapped onto the greatest element: $j_E(s) = j_E(t) = j_E((s,t)) = s \to t$. 
	\end{itemize}
	Similarly, $j_{E'}(0')$ can be either $(s',t')$ or $s' \to t'$, and this choice determines the mapping of all edges in $\Omega(E')$.
	That gives us the remaining four topologies:\\
	\begin{table}[H]
		\centering
		\begin{tabular}{c|cc}
			$j$ 		& $j_E(0)$		& $j_{E'}(0')$	\\ \hline
			$j^{10}$  	& $(s,t)$  		& $(s',t')$         \\
			$j^{11}$  	& $s \to t$     & $(s',t')$         \\
			$j^{12}$  	& $(s,t)$       & $s' \to t'$       \\
			$j^{13}$	& $s \to t$		& $s' \to t'$
		\end{tabular}
	\end{table}
    This concludes the proof. \qedhere
\end{proof}

\subsection{Proofs of \texorpdfstring{\cref{sec:fuzzy_case}}{Section 5}}

\repeatlemma{lem:alternate_def_nucleus}

\begin{proof}[Proof of \cref{lem:alternate_def_nucleus}]
	For the first claim, we prove (D)-(G):
	\begin{enumerate}[topsep=4pt, noitemsep]
		\item[(D)]
		Every element is smaller than the greatest element, thus $\phi(\top) \leq \top$.
		Moreover, because $\phi$ is increasing {\normalfont(B)}, we have $\top \leq \phi(\top)$. 
		Hence, equality follows: $\top = \phi(\top)$.
		
		\item[(E)]
		By $\meet$-preservation {\normalfont(A)}, we have $\phi(a) = \phi(a \meet b) = \phi(a) \meet \phi(b) \leq \phi(b)$ for any $a \leq b$. 
		
		\item[(F)]
		The first inequality $\phi(\phi(a)) \leq \phi(a)$ is given by {\normalfont(C)}.
		The second inequality $\phi(a) \leq \phi(\phi(a))$ follows from $\phi$ being increasing {\normalfont(B)}.
		
		\item[(G)]
		We have $\phi(a) \meet b \stackrel{{\normalfont(B)}}{\leq} \phi(a) \meet \phi(b) \stackrel{{\normalfont(A)}}{=} \phi(a \meet b)$.
	\end{enumerate}
	For the second claim, we suppose that a function $\phi \colon \labels \to \labels$ satisfies (B), (C), (E), and (G).
	We prove that $\phi$ then also satisfies (A), i.e., that $\phi(a \meet b) = \phi(a) \meet \phi(b)$ for all $a,b \in \labels$.
	\begin{enumerate}[topsep=4pt, noitemsep]
		\item[$(\leq)$]
		One direction follows from monotonicity~(E):
		\[
			\left.
			\begin{aligned}
				a \meet b \leq a &\stackrel{\text{(E)}}{\implies} \phi(a \meet b) \leq \phi(a) \\
				a \meet b \leq b &\stackrel{\text{(E)}}{\implies} \phi(a \meet b) \leq \phi(b)
			\end{aligned}
			\right\}
			\implies
			\phi(a \meet b) \leq \phi(a) \meet \phi(b).
		\]
		
		\item[$(\geq)$]
		The other direction follows from (G), (E), and (C):
		\[
			\phi(a) \meet \phi(b)
			\stackrel{\text{(G)}}{\leq}
			\phi(a \meet \phi(b))
			\stackrel{\text{(G), (E)}}{\leq}
			\phi(\phi(a \meet b))
			\stackrel{\text{(C)}}{\leq}
			\phi(a \meet b).
			\qedhere
		\]
	\end{enumerate}
\end{proof}

\repeatlemma{lem:topologies_on_FuzzySet}

\begin{proof}[Proof of \cref{lem:topologies_on_FuzzySet}]
    For simplicity, the memberships will sometimes be denoted in superscripts instead of using membership functions.
	Recall that $0$ denotes the empty (fuzzy) set and $1 = \set{ \, \cdot^\top }$ denotes the terminal fuzzy set.
	Take a topology $\topology$ on $\FuzzySetDefault$.
	Because $0 \subseteq 1$ is (vacuously) a \textit{strong} fuzzy subset, by \ref{def:topology_on_quasitopos_5_strong} its closure must also be a strong subset of $1$, i.e., $0$ or $1$.
	
	The first option is $\topologybar{0 \subseteq 1} = (1 \subseteq 1)$.
	Let $A$ be an arbitrary fuzzy set.
	As shown in the diagram below, by stability under pullbacks \ref{def:topology_on_quasitopos_4_stable_under_PB} the closure $\topologybar{0 \subseteq A}$ is the pullback of $A \to 1$ and of $(1 \subseteq 1)$ and hence is necessarily the whole of $A$.
	By monotonicity \ref{def:topology_on_quasitopos_3_monotone}, this forces every closure to be $\topologybar{A' \subseteq A} = (A \subseteq A)$, i.e., we have the \emph{trivial} topology.
	\[
	\begin{tikzcd}[column sep=5mm, row sep=5mm,ampersand replacement=\&]
		0 \& 0 \\
		A \& 1
		\ar[from=1-1, to=1-2, equal]
		\ar[from=1-1, to=2-1, tail]
		\ar[from=1-2, to=2-2, tail]
		\ar[from=2-1, to=2-2, "!"']
		\ar[from=1-1, to=2-2, "\lrcorner"{anchor=center, pos=0.125}, draw=none]
	\end{tikzcd}
	\qquad 
	\stackrel{\text{\ref{def:topology_on_quasitopos_4_stable_under_PB}}}{\implies} \qquad
	\begin{tikzcd}[column sep=5mm, row sep=5mm,ampersand replacement=\&]
		\topologybar{0 \subseteq A} \& 1 \\
		A \& 1
		\ar[from=1-1, to=1-2]
		\ar[from=1-1, to=2-1, tail]
		\ar[from=1-2, to=2-2, tail, "{\topologybar{0 \subseteq 1} = (1 \subseteq 1)}"]
		\ar[from=2-1, to=2-2, "!"']
		\ar[from=1-1, to=2-2, "\lrcorner"{anchor=center, pos=0.125}, draw=none]
	\end{tikzcd}
	\]
	
	The second option is $\topologybar{0 \subseteq 1} = (0 \subseteq 1)$.
	Take an arbitrary fuzzy set $A$.
	With the same reasoning as just above, we have that $\topologybar{0 \subseteq A}$ is the pullback of $A \to 1$ and $\topologybar{0 \subseteq 1} = (0 \subseteq 1)$, hence it is $(0 \subseteq A)$.
	For an arbitrary fuzzy subset $A' \subseteq A$, the fuzzy subsets $A' \subseteq A$ and $A \setminus A' \subseteq A$ are complements and have thus an empty intersection, i.e., an empty pullback, as depicted in the diagram below on the left.
	Therefore, by \ref{def:topology_on_quasitopos_4_stable_under_PB}, the pullback of $\topologybar{A' \subseteq A}$ and $A \setminus A' \subseteq A$ must be the empty set, meaning that the subsets have an empty intersection. 
	The complements $A'$ and $A \setminus A'$ are already maximal in order to have an empty intersection.
	Hence the closure $\topologybar{A' \subseteq A}$ has added no element to $A'$.
	\[
	\begin{tikzcd}[column sep=5mm, row sep=5mm,ampersand replacement=\&]
		0 \& A' \\
		A \setminus A' \& A
		\ar[from=1-1, to=1-2, equal]
		\ar[from=1-1, to=2-1, tail]
		\ar[from=1-2, to=2-2, tail]
		\ar[from=2-1, to=2-2, tail]
		\ar[from=1-1, to=2-2, "\lrcorner"{anchor=center, pos=0.125}, draw=none]
	\end{tikzcd}
	\qquad
	\stackrel{\text{\ref{def:topology_on_quasitopos_4_stable_under_PB}}}{\implies}
	\qquad
	\begin{tikzcd}[column sep=5mm, row sep=5mm,ampersand replacement=\&]
		0 \& \topologybar{A' \subseteq A} \\
		A \setminus A' \& A
		\ar[from=1-1, to=1-2]
		\ar[from=1-1, to=2-1, tail, "{\makecell[r]{
			\topologybar{(0 \subseteq A \setminus A')}\\
			= (0 \subseteq A \setminus A')
			}}"{swap, xshift=-.3mm}]
		\ar[from=1-2, to=2-2, tail]
		\ar[from=2-1, to=2-2, tail]
		\ar[from=1-1, to=2-2, "\lrcorner"{anchor=center, pos=0.125}, draw=none]
	\end{tikzcd}
	\]
	
	In the rest of the proof, we show the bijection between topologies that do not add elements, and nuclei $\phi \colon \labels \to \labels$.
	
	(Topology $\implies$ Nucleus)
	Start with a topology that does not add elements, i.e., only the memberships can change.
	Define $\phi \colon \labels \to \labels$ as the change that happens to the membership of a single element as a subset of $1$, i.e.,
	\[
		\topologybar{\set{ \, \cdot^x } \subseteq 1} = \big( \set{ \, \cdot^{\phi(x)}} \subseteq 1 \big).
	\]
	By \ref{def:topology_on_quasitopos_4_stable_under_PB}, 
	$\topologybar{\set{\,\cdot^x} \subseteq \set{\,\cdot^y}} = \big(\set{\,\cdot^{\phi(x) \meet y}} \subseteq \set{\,\cdot^y} \big)$:
	\[
	\begin{tikzcd}[column sep=5mm, row sep=5mm,ampersand replacement=\&]
		\cdot^{x \meet y} = \cdot^x \& \cdot^x \\
		\cdot^y \& \cdot^\top
		\ar[from=1-1, to=1-2, equal]
		\ar[from=1-1, to=2-1, tail]
		\ar[from=1-2, to=2-2, tail]
		\ar[from=2-1, to=2-2, tail, "\inclusion"']
		\ar[from=1-1, to=2-2, "\lrcorner"{anchor=center, pos=0.125}, draw=none]
	\end{tikzcd}
	\qquad
	\stackrel{\text{\ref{def:topology_on_quasitopos_4_stable_under_PB}}}{\implies}
	\qquad
	\begin{tikzcd}[column sep=5mm, row sep=5mm,ampersand replacement=\&]
		\topologybar{\cdot^x \subseteq \cdot^y} = \cdot^{\phi(x) \meet y} \& \cdot^{\phi(x)} \\
		\cdot^y \& \cdot^\top
		\ar[from=1-1, to=1-2]
		\ar[from=1-1, to=2-1, tail]
		\ar[from=1-2, to=2-2, tail, "\topologybar{\cdot^x \subseteq 1}"]
		\ar[from=2-1, to=2-2, tail, "\inclusion"']
		\ar[from=1-1, to=2-2, "\lrcorner"{anchor=center, pos=0.125}, draw=none]
	\end{tikzcd}
	\]
	This stays true for arbitrary $A' \subseteq A$, when $a^x \in A'$ and $a^y \in A$: then $a$ has membership $\phi(x) \meet y$ in $\topologybar{A' \subseteq A}$.
	This follows from the fact that pullback of $\topologybar{A' \subseteq A}$ and $\set{a^y} \subseteq A$ is $\set{a^{\phi(x) \meet y}}$.
	
	\[
	\begin{tikzcd}[column sep=5mm, row sep=5mm,ampersand replacement=\&]
		a^x \& A' \\
		a^y \& A
		\ar[from=1-1, to=1-2, tail]
		\ar[from=1-1, to=2-1, tail]
		\ar[from=1-2, to=2-2, tail]
		\ar[from=2-1, to=2-2, tail]
		\ar[from=1-1, to=2-2, "\lrcorner"{anchor=center, pos=0.125}, draw=none]
	\end{tikzcd}
	\qquad
	\stackrel{\text{\ref{def:topology_on_quasitopos_4_stable_under_PB}}}{\implies}
	\qquad
	\begin{tikzcd}[column sep=5mm, row sep=5mm,ampersand replacement=\&]
		a^{\phi(x) \meet y} \& \topologybar{A' \subseteq A} \\
		a^y \& A
		\ar[from=1-1, to=1-2, tail]
		\ar[from=1-1, to=2-1, tail, "\topologybar{a^x \subseteq a^y}"']
		\ar[from=1-2, to=2-2, tail]
		\ar[from=2-1, to=2-2, tail] 
		\ar[from=1-1, to=2-2, "\lrcorner"{anchor=center, pos=0.125}, draw=none]
	\end{tikzcd}
	\]
	To prove that $\phi$ is a nucleus, by \cref{lem:alternate_def_nucleus} it suffices to prove that it is 
	increasing~{\normalfont(B)}, satisfies one direction of idempotency~{\normalfont(C)}, is monotone~{\normalfont(E)}, and that it satisfies $\phi(x) \meet y \leq \phi(x \meet y)$~{\normalfont(G)}.
	These properties hold:
	\ref{def:topology_on_quasitopos_1_increasing} implies $\phi$ is increasing.
	\ref{def:topology_on_quasitopos_2_idempotent} implies $\phi$ is idempotent.
	\ref{def:topology_on_quasitopos_3_monotone} implies $\phi$ is monotone.
	By \ref{def:topology_on_quasitopos_4_stable_under_PB}, we have
	$\phi(x \meet y) \meet y = \phi(x) \meet y$, which implies $\phi(x) \meet y \leq \phi(x \meet y)$.
	\begin{align*}
		\begin{tikzcd}[column sep=5mm, row sep=5mm,ampersand replacement=\&]
			\cdot^{x \meet y} \& \cdot^x \\
			\cdot^y \& \cdot^\top
			\ar[from=1-1, to=1-2, tail]
			\ar[from=1-1, to=2-1, tail]
			\ar[from=1-2, to=2-2, tail]
			\ar[from=2-1, to=2-2, tail]
			\ar[from=1-1, to=2-2, "\lrcorner"{anchor=center, pos=0.125}, draw=none]
		\end{tikzcd}
		\qquad
		&\stackrel{\text{\ref{def:topology_on_quasitopos_4_stable_under_PB}}}{\implies}
		\qquad
		\begin{tikzcd}[column sep=5mm, row sep=5mm,ampersand replacement=\&]
			\topologybar{\cdot^{x \meet y} \subseteq \cdot^y} = \cdot^{\phi(x \meet y) \meet y} \& \cdot^{\phi(x)} \\
			\cdot^y \& \cdot^\top
			\ar[from=1-1, to=1-2, tail]
			\ar[from=1-1, to=2-1, tail]
			\ar[from=1-2, to=2-2, tail, "\topologybar{\cdot^x \subseteq 1}"]
			\ar[from=2-1, to=2-2, tail]
			\ar[from=1-1, to=2-2, "\lrcorner"{anchor=center, pos=0.125}, draw=none]
		\end{tikzcd} 
	\end{align*}

	(Nucleus $\implies$ Topology)
	Suppose a nucleus $\phi \colon \labels \to \labels$ is given.
	Define 
	\begin{align*}
		\topology_{(A,\alpha)} \colon \MonoClass(A,\alpha) &\to \MonoClass(A,\alpha) \\
		(A',\alpha') &\mapsto (A', \phi \alpha' \meet \alpha).
	\end{align*}
	We check the axioms of \cref{def:closure_operator}.
	\begin{enumerate}[topsep=2pt]
		\item[\ref{def:topology_on_quasitopos_1_increasing}] 
		Increasing: 
		Because $\phi$ is increasing {\normalfont(B)}, we have $\alpha' \leq \phi \alpha'$.
		And by hypothesis $\alpha' \leq \alpha$.
		Hence $\alpha' \leq \phi \alpha' \meet \alpha$.
		
		\item[\ref{def:topology_on_quasitopos_2_idempotent}] 
		Idempotent:
		We check that $\phi \alpha' \meet \alpha = \phi(\phi\alpha' \meet \alpha) \meet \alpha$.
		
		$(\leq)$
		As said in item (i), $\alpha' \leq \phi \alpha' \meet \alpha$.
		Hence, by monotonicity of $\phi$, we have $\phi \alpha' \leq \phi(\phi\alpha' \meet \alpha)$.
		We take the conjunction with $\alpha$ on both sides to obtain what we want.
		
		$(\geq)$
		By property of a conjunction, $\phi \alpha' \meet \alpha \leq \phi \alpha'$.
		By monotonicity and idempotence of $\phi$, we have $\phi(\phi \alpha' \meet \alpha) \leq \phi \phi \alpha' = \phi \alpha'$.
		We take the conjunction with $\alpha$ on both sides to obtain what we want.
		
		\item[\ref{def:topology_on_quasitopos_3_monotone}]  
		Monotone:
		Suppose $(A',\alpha') \subseteq (A'',\alpha'') \subseteq (A,\alpha)$.
		We check that $\phi \alpha' \meet \alpha \leq \phi \alpha'' \meet \alpha$.
		This follows from $\alpha' \leq \alpha''$, and monotonicity of $\phi$ and of $- \meet \alpha$.
		
		\item[\ref{def:topology_on_quasitopos_4_stable_under_PB}]
		Stable under pullbacks:
		Consider the situation of the diagram.
		\[
		\begin{tikzcd}[column sep=5mm, row sep=5mm,ampersand replacement=\&]
			(A',\alpha') \& (B',\beta') \\
			(A,\alpha) \& (B,\beta)
			\ar[from=1-1, to=1-2, "g\restrict{A'}"]
			\ar[from=1-1, to=2-1, tail, "\inclusion_1"']
			\ar[from=1-2, to=2-2, tail, "\inclusion_2"]
			\ar[from=2-1, to=2-2, "g"']
			\ar[from=1-1, to=2-2, "\lrcorner"{anchor=center, pos=0.125}, draw=none]
		\end{tikzcd}
		\]
		For $(A',\alpha')$ to be the pullback means that $\alpha' = \alpha \meet \beta' g$.
		We check that $\topologybar{(A',\alpha')}$ is the pullback of $g$ and $\topologybar{\inclusion_2}$.
		Since the closure only changes the memberships, it means checking that $\phi\alpha' \meet \alpha = \alpha \meet \big( \phi\beta' \meet \beta \big)g$.
		
		$(\leq)$ 
		We reason as follows
		\begin{align*}
			&& \alpha' = \alpha \meet \beta' g &\leq \beta' g
			\tag*{def.~$\alpha'$ and property of $\meet$} \\
			&\Rightarrow & \phi \alpha' &\leq \phi \beta' g
			\tag*{$\phi$ monotone} \\
			&\Rightarrow & \phi \alpha' \meet \alpha &\leq \alpha \meet \phi \beta' g
			\tag*{$-\meet \alpha$ monotone} \\
			&\Rightarrow & \phi \alpha' \meet \alpha &\leq \alpha \meet \beta g \meet \phi \beta' g
			\tag*{$\alpha \leq \beta g$ by def.~of $g$} \\
			&\Rightarrow & \phi \alpha' \meet \alpha &\leq \alpha \meet (\beta \meet \phi \beta')g
			\tag*{rewriting}
		\intertext{
			$(\geq)$
			By using the hypothesis $\phi(x) \meet y \leq \phi(x \meet y)$ with $x \defeq \beta' g$ and $y \defeq \alpha$, we obtain 
		}
			&&\alpha \meet \phi \beta' g &\leq \phi(\alpha \meet \beta' g)
			\\
			&\Rightarrow & \alpha \meet \phi \beta' g &\leq \phi \alpha'
			\tag*{def.~$\alpha'$} \\
			&\Rightarrow & \alpha \meet \phi \beta' g \meet \beta g &\leq \phi \alpha'
			\tag*{identical, since $\alpha \leq \beta g$ by def.~of $g$} \\
			&\Rightarrow & \alpha \meet (\phi \beta' \meet \beta) g &\leq \phi \alpha' \meet \alpha
			\tag*{$\alpha$ already in the meet of the LHS}
		\end{align*}
		
		\item[\ref{def:topology_on_quasitopos_5_strong}]  
		Preserve strongness:
		Take $(A', \alpha') \subseteq (A,\alpha)$ strong, i.e., with $\alpha' = \alpha\smash{\restrict{A'}}$.
		To see that its closure $(A', \phi\alpha \meet \alpha) \subseteq (A,\alpha)$ is also strong, we check that $\phi\alpha \meet \alpha = \alpha$ on $A'$.
		This follows from $\alpha \leq \phi \alpha$, which holds because $\phi$ is increasing.
	\end{enumerate}
	
	Notice that the two constructions are inverses of each other, because in each case, $\phi(x)$ represents the change of membership of a singleton $\set{ \, \cdot^x } \subseteq 1$.
\end{proof}

\repeatlemma{lem:lnotlnot_monad_plus_inequality_in_DMHeytAlg}

\begin{proof}[Proof of \cref{lem:lnotlnot_monad_plus_inequality_in_DMHeytAlg}]
    The fact that $\phi_{\lnot \lnot}$ is a monad, i.e., increasing, monotone and idempotent, is true in any Heyting algebra \cite[Prop.~1.2.8 (1),(2),(4)]{Borceux_1994_vol3}.
	Suppose now additionally that $(\labels, \leq)$ is a de Morgan Heyting algebra.
	By \cref{lem:alternate_def_nucleus}, in order to prove that $\phi_{\lnot \lnot}$ is a nucleus, it suffices to prove that $({\lnot \lnot} x) \meet y \leq {\lnot \lnot} (x \meet y)$.
	By the definition of Heyting algebra , we have an equivalence:
	\[
		\lnot \lnot x \meet y \leq \lnot \lnot (x \meet y) = (\lnot(x \meet y) \Rightarrow \bot)
		\quad\stackrel{\text{HA def.}}{\iff}\quad
		(\lnot \lnot x \meet y) \meet \lnot (x \meet y) \leq \bot.
	\]
	The right-hand side of this equivalence holds:
	\begin{align*}
		\lnot \lnot x \meet y \meet \lnot (x \meet y)
		&\leq \lnot \lnot x \meet \lnot \lnot y \meet \lnot (x \meet y)
		\tag*{$\lnot \lnot$ increasing} \\
		&= \lnot (\lnot x \join \lnot y) \meet \lnot (x \meet y)
		\tag*{$1$st De Morgan law} \\
		&= \lnot \lnot (x \meet y) \meet \lnot (x \meet y)
		\tag*{$2$nd De Morgan law} \\
		&= \big(\lnot (x \meet y) \to \bot \big) \meet \lnot (x \meet y)
		\tag*{def.~$\lnot$} \\
		&\leq \bot
		\tag*{Modus ponens}
	\end{align*}
	Therefore $({\lnot \lnot} x) \meet y \leq {\lnot \lnot} (x \meet y)$ holds and $\phi_{\lnot\lnot}$ is a nucleus by \cref{lem:alternate_def_nucleus}, which concludes the proof.
\end{proof}

\repeatlemma{lem:sheaves_in_FuzzySet}

\begin{proof}[Proof of \cref{lem:sheaves_in_FuzzySet}]
    Take a topology on $\FuzzySetDefault$ that corresponds to a nucleus $\phi \colon \labels \to \labels$ as described in \cref{lem:topologies_on_FuzzySet}.
	A fuzzy subset $(A',\alpha') \subseteq (A,\alpha)$ is dense if and only if $A' = A$, because the closure does not add elements, and $\phi\alpha' \meet \alpha = \alpha$, which is equivalent to $\alpha \leq \phi \alpha'$.   
	Fix a fuzzy set $(B,\beta)$ and take an arbitrary dense fuzzy subset $(A,\alpha') \subseteq (A,\alpha)$ and arbitrary $f \colon (A,\alpha') \to (B, \beta)$, meaning $\alpha' \leq \beta f$.
	Hence,
	\begin{equation}
		\label{eq:sheaf_requirement}
		\alpha \stackrel{\text{dense}}{\leq} \phi \alpha' \stackrel[\text{+ $\phi$ mon.}]{\text{hyp.~$f$}}{\leq} \phi \beta f.
	\end{equation}
	If a factorisation $g$ of $f$ through $(A,\alpha)$ exists, it is necessarily equal to $f$ and hence unique. 
	This means that every $(B, \beta)$ is separated.
	Meanwhile, for $(B,\beta)$ to be a sheaf requires also the \textit{existence} of a factorisation $g=f \colon (A,\alpha) \to (B,\beta)$, which must be a fuzzy set morphism, i.e., satisfy $\alpha \leq \beta f$.
	We prove the correspondence of the lemma.
	
	\medskip
	\noindent\begin{minipage}{.8\linewidth}
		\begin{itemize}
			\item[$(\Leftarrow)$]
			Suppose $\ima(\beta) \subseteq \ima(\phi)$.
			Then $\phi \beta f = \beta f$ by the assumption and idempotence of $\phi$.
			Combined with \eqref{eq:sheaf_requirement}, we have that $\alpha \leq \beta f$, proving that $(B,\beta)$ is a sheaf (since $\alpha$ and $f$ are arbitrary).
			
			\item[$(\Rightarrow)$]
			Suppose $(B, \beta)$ is a sheaf.
			To prove $\ima(\beta) \subseteq \ima(\phi)$, take $b^x \in B$ with membership $x = \beta(b) \in \ima(\beta)$ and let us prove that $x \in \ima(\phi)$.
			Take $A' \defeq \set{b^x}$ and $A \defeq \set{b^{\phi(x)}}$, and $f = \inclusion$.
			Because $(B,\beta)$ is a sheaf, the factorisation exists, implying that the memberships satisfy $x \leq \phi(x) \leq x$.
			Therefore, $x= \phi(x) \in \ima(\phi)$ as desired.
		\end{itemize}
	\end{minipage}
	\begin{minipage}{.2\linewidth}
		\begin{center}
			\begin{tikzcd}[ampersand replacement=\&]
				b^x \& \\
				b^{\phi(x)} \& B
				\ar[from=1-1, to=2-1, tail, "\inclusion"']
				\ar[from=1-1, to=2-2, tail, "\inclusion"]
				\ar[from=2-1, to=2-2, dotted, tail]
			\end{tikzcd}
		\end{center}
		\qedhere
	\end{minipage}
\end{proof}

\end{document}